\documentclass[letter,11pt]{article}

\usepackage[colorlinks]{hyperref}
  \hypersetup{linkcolor=blue,filecolor=blue,citecolor=blue,urlcolor=blue}

  \newcommand{\fullversion}[1]{#1}
\newcommand{\submversion}[1]{}
\usepackage{xspace}
\fullversion{
\usepackage{fullpage}
}
\usepackage{verbatim}
\usepackage{todonotes}

\fullversion{
}

\fullversion{
\usepackage{amsmath,amsthm}
\usepackage[T1]{fontenc}
}

\newcommand{\prab}[1]{{\color{red} Prab:#1}}

\newcommand{\ket}[1]{|#1\rangle}

\def\cX{{\cal X}}

\fullversion{
\usepackage{amsthm}
\newtheorem{theorem}{Theorem}

\newtheorem{definition}[theorem]{Definition}
\newtheorem{lemma}[theorem]{Lemma}
\newtheorem{claim}[theorem]{Claim}

}

\usepackage{amsmath,amsfonts,amssymb}

\usepackage[utf8]{inputenc}
\usepackage{tabularx}
\usepackage[normalem]{ulem}
\usepackage{setspace}
\usepackage{fancyhdr}
\usepackage{lastpage}
\usepackage{extramarks}
\usepackage{chngpage}
\usepackage{soul}
\usepackage{xspace}
\usepackage{bm}
\usepackage{nicefrac}
\usepackage{paralist}
\usepackage{enumitem}
\usepackage{tikz}

\newcommand{\bfs}{\mathbf{s}}
\newcommand{\bfA}{\mathbf{A}}
\newcommand{\bfe}{\mathbf{e}}

\newcommand{\bfu}{\mathbf{u}}

\newcommand{\secparam}{\lambda}
\newcommand{\negl}{\mathsf{negl}}

\newcommand{\ignore}[1]{}
\newcommand{\la}{\leftarrow}

\fullversion{
\newcommand{\inst}{{\bf y}}
}

\newcommand{\cktclass}{\mathcal{C}}

\newcommand{\poly}{\mathrm{poly}}

\newcommand{\distr}{\mathcal{D}}

\newcommand{\fc}{f} % Functionality

\newcommand{\prvr}{\mathsf{Prover}}
\newcommand{\vrfr}{\mathsf{Verifier}}

\newcommand{\prty}[1]{P_{#1}} % party in 2PC
\newcommand{\inpl}[1]{x_{#1}} % notation for inputs of parties

\newcommand{\rel}{\mathcal{R}}
\newcommand{\otp}{\mathsf{OT}} % OT protocol 
\newcommand{\otpmsg}[1]{OT_{#1}}

\newcommand{\ilength}[1]{\ell_{#1}} % Length of party's input

\newcommand{\otput}{\mathfrak{out}} % party's output

\newcommand{\simr}{\mathsf{Sim}}

\newcommand{\aux}{\mathsf{aux}}

\newcommand{\view}{\mathsf{View}}

\newcommand{\extractor}{\mathsf{Ext}}

\newcommand{\msg}{\mathsf{msg}}

\newcommand{\bit}{\mathsf{bit}}

\newcommand{\gen}{\mathsf{Gen}}

\newcommand{\otpad}{K'}

\newcommand{\wi}{\mathsf{wi}}

\newcommand{\bfr}{\mathbf{r}}
\newcommand{\bfk}{\mathbf{k}}
\newcommand{\bfp}{\mathsf{p}}
\newcommand{\badevent}{\mathsf{BAD}}

\newcommand{\expt}{\mathsf{Expt}}
\newcommand{\adversary}{\mathcal{A}}  % Adversary
\newcommand{\hybrid}{\mathsf{Hyb}} % Hybrid
\newcommand{\prob}{\mathsf{Pr}} % Probability
\newcommand{\setup}{\mathsf{Setup}}

\newcommand{\enc}{\mathsf{Enc}}
\newcommand{\dec}{\mathsf{Dec}}

\newcommand{\pk}{\mathsf{PK}}
\newcommand{\ct}{\mathsf{CT}}

\newcommand{\fhe}{\mathsf{qFHE}}

\newcommand{\gc}{\mathsf{GC}}
\newcommand{\gcgen}{\mathsf{Garble}}

\newcommand{\gceval}{\mathsf{EvalGC}}
\newcommand{\gckt}{\mathcal{GC}}

\newcommand{\sender}{\mathsf{S}}
\newcommand{\receiver}{\mathsf{R}}
\newcommand{\td}{\mathsf{td}}
\newcommand{\invert}{\mathsf{Inv}}
\newcommand{\protwi}{\Pi_{\mathsf{WI}}}

\newcommand{\witness}{\mathbf{w}}

\newcommand{\key}{\mathbf{k}}
\newcommand{\out}{\mathsf{Out}}

\newcommand{\eval}{\mathsf{Eval}}

\newcommand{\sk}{\mathsf{SK}}

\newcommand{\bfd}{\mathbf{d}}
\newcommand{\bfc}{\mathbf{c}}
\newcommand{\comm}{\mathsf{Comm}}
\newcommand{\lobf}{\mathsf{Obf}}
\newcommand{\leval}{\mathsf{ObfEval}}
\newcommand{\lockC}{\mathbf{C}}
\newcommand{\lockclass}{\mathcal{C}}
\newcommand{\obfC}{\widetilde{\lockC}}

\renewcommand{\otpad}{\mathsf{otp}}
\newcommand{\protext}{\mathsf{QEXT}}
\newcommand{\extrel}{\mathcal{R_{\mathrm{EXT}}}}
\newcommand{\protpc}{\mathsf{SFE}}
\newcommand{\trans}{\mathcal{T}}
\newcommand{\rqext}{\mathbf{r}_{\mathrm{qext}}}
\newcommand{\rnd}{t}
\newcommand{\ext}{\mathsf{Ext}}
\newcommand{\locksimr}{\mathsf{LObf}.\mathsf{Sim}}
\newcommand{\lockobfscheme}{\mathbf{LObf}}
\newcommand{\lang}{\mathcal{L}}
\newcommand{\Simu}{\mathsf{Sim}}
\newcommand{\qaux}{\rho_{\aux}}
\newcommand{\tr}{\mathsf{Tr}}
\newcommand{\cE}{\mathcal{E}}
\newcommand{\cD}{\mathcal{D}}
\newcommand{\cF}{\mathcal{F}}
\newcommand{\cH}{\mathcal{H}}
\newcommand{\cM}{\mathcal{M}}
\newcommand{\cC}{\mathcal{C}}
\newcommand{\RegP}{\mathrm{R}_\prvr}
\newcommand{\RegV}{\mathrm{R}_\vrfr}
\newcommand{\RegM}{\mathrm{R}_{\mathsf{M}}}
\newcommand{\cY}{\mathcal{Y}}
\newcommand{\cA}{\mathcal{A}}
\newcommand{\cK}{\mathcal{K}}
\newcommand{\supp}{\mathsf{Supp}}
\newcommand{\inv}{\mathsf{Inv}}
\newcommand{\cR}{\mathcal{R}}
\newcommand{\chk}{\mathsf{Chk}}

\newcommand{\Id}{\mathsf{Id}}
\newcommand{\sfefunc}{\mathbf{F}}

\newcommand{\prover}{\prvr}

\ignore{
\newcommand{\prvr}{\mathsf{Prover}}
\newcommand{\vrfr}{\mathsf{Verifier}}

\newcommand{\prty}[1]{P_{#1}} % party in 2PC
\newcommand{\inpl}[1]{x_{#1}} % notation for inputs of parties

\usepackage{breakcites}

\newcommand{\rel}{\mathcal{R}}
}

\newcommand{\simulator}{\mathsf{Sim}}

\usepackage{breakcites}

\title{Secure Quantum Extraction Protocols}
\fullversion{
\author{Prabhanjan Ananth\\ UCSB \thanks{prabhanjan@cs.ucsb.edu} \and Rolando L. La Placa\\ MIT \thanks{rlaplaca@mit.edu}}
}
\submversion{
\author{Prabhanjan Ananth\inst{1} \and Rolando L. La Placa \inst{2}}
\institute{University of California, Santa Barbara\\ \texttt{prabhanjan@cs.ucsb.edu} \and MIT\\ \texttt{rlaplaca@mit.edu}}
}
\date{}

\begin{document}

\maketitle

\renewcommand{\inst}{\mathbf{z}}
\newcommand{\veck}{\langle \bfk \rangle}
\submversion{
\renewcommand{\inst}{{\bf y}}
}

\begin{abstract}
\noindent Knowledge extraction, typically studied in the classical setting, is at the heart of several cryptographic protocols. The prospect of quantum computers forces us to revisit the concept of knowledge extraction in the presence of quantum adversaries.    
\par We introduce the notion of secure quantum extraction protocols. A secure quantum extraction protocol for an NP relation $\rel$ is a classical interactive protocol between a sender and a receiver, where the sender gets as input the instance $\inst$ and witness $\witness$ while the receiver only gets the instance $\inst$ as input. There are two properties associated with a secure quantum extraction protocol: (a) {\em Extractability}:  for any efficient quantum polynomial-time (QPT) adversarial sender, there exists a QPT extractor that can extract a witness $\witness'$ such that $(\inst,\witness') \in \rel$ and, (b) {\em Zero-Knowledge}: a malicious receiver, interacting with the sender, should not be able to learn any information about $\witness$. 
\par We study and construct two flavors of secure quantum extraction protocols.
\begin{itemize}
 \item {\bf Security against QPT malicious receivers}: First we consider the setting when the malicious receiver is a QPT adversary. In this setting, we construct a secure quantum extraction protocol for NP assuming the existence of quantum fully homomorphic encryption satisfying some mild properties (already satisfied by existing constructions [Mahadev,  FOCS'18, Brakerski CRYPTO'18]) and quantum hardness of learning with errors. The novelty of our construction is a new non-black-box technique in the quantum setting. All previous extraction techniques in the quantum setting were solely based on quantum rewinding. 
    \item {\bf Security against classical PPT malicious receivers}: We also consider the setting when the malicious receiver is a classical probabilistic polynomial time (PPT) adversary. In this setting, we construct a secure quantum extraction protocol for NP solely based on the quantum hardness of learning with errors. Furthermore, our construction satisfies {\em quantum-lasting security}: a malicious receiver cannot later, long after the protocol has been executed, use a quantum computer to extract a valid witness from the transcript of the protocol. 
\end{itemize}
\noindent Both the above extraction protocols are {\em constant round} protocols. 
\par We present an application of secure quantum extraction protocols to zero-knowledge (ZK). Assuming quantum hardness of learning with errors, we present the first construction of ZK argument systems for NP in constant rounds based on the quantum hardness of learning with errors with: (a) zero-knowledge against QPT malicious verifiers and,  (b) soundness against classical PPT adversaries. Moreover, our construction satisfies the stronger (quantum) auxiliary-input zero knowledge property and thus can be composed with other protocols secure against quantum adversaries.     
\end{abstract}

\fullversion{
\newpage
}

%\tableofcontents

\section{Introduction}
Knowledge extraction is a quintessential concept employed to argue the security of  classical zero-knowledge systems and secure two-party and multi-party computation protocols. The seminal work of Feige, Lapidot and Shamir~\cite{FLS99} shows how to leverage knowledge extraction to construct zero-knowledge protocols. The ideal world-real world paradigm necessarily requires the simulator to be able to extract the inputs of the adversaries to argue the security of secure computation protocols. 
\par Typically, knowledge extraction is formalized by defining a knowledge extractor that given access to the adversarial machine, outputs the input of the adversary. The prototypical extraction technique employed in several cryptographic protocols is rewinding. In the rewinding technique, the extractor, with oracle access to the adversary,  rewinds the adversary to a previous state to obtain more than one protocol transcript which in turn gives the ability to the extractor to extract from the adversary. While rewinding has proven to be quite powerful, it has several limitations~\cite{GK90}. Over the years, cryptographers have proposed novel extraction techniques to circumvent the barriers of rewinding. Each time a new extraction technique was invented, it has advanced the field of zero-knowledge and secure computation. As an example, the breakthrough work of Barak~\cite{Bar01} proposed a non-black-box extraction technique -- where the extractor crucially uses the code of the verifier for extraction -- and used this to obtain the first feasibility result on constant-round public-coin zero-knowledge argument system for NP. Another example is the work of Pass~\cite{Pas03} who introduced the technique of super-polynomial time extraction and presented the first feasibility result on 2-round concurrent ZK argument system albeit under a weaker  simulation definition.   

\paragraph{Extracting from Quantum Adversaries.} The prospect of quantum computers introduces new challenges in the design of zero-knowledge and secure computation protocols. As a starting step towards designing these protocols, we need to address the challenge of knowledge extraction against quantum adversaries. So far, the only technique used to extract from quantum adversaries is quantum rewinding~\cite{Wat09}, which has already been studied by a few  works~\cite{Wat09,JKMR06,Unruh12,ARU14,unruh2016computationally} in the context of quantum zero-knowledge protocols. 
\par Rewinding a quantum adversary, unlike its classical counterpart, turns out to be tricky due to two reasons, as stated in Watrous~\cite{Wat09}: firstly, intermediate quantum states of the adversary cannot be copied (due to the universal no-cloning theorem) and secondly, if the adversary performs some measurements then this adversary cannot be rewound since measurements in general are irreversible processes. As a result, the existing quantum rewinding techniques tend to be "oblivious"~\cite{Unruh12}, to rewind the adversary back to an earlier point, the extraction should necessarily forget all the information it has learnt from that point onwards.
As a result of these subtle issues, the analysis of quantum rewinding turns out to be quite involved making it difficult to use it in the security proofs. Moreover, existing quantum rewinding techniques~\cite{Wat09,Unruh12} pose a bottleneck towards achieving a constant round extraction technique; we will touch upon this later.   
\par In order to advance the progress of constructing quantum-secure (or post-quantum) cryptographic protocols, it is necessary that we look beyond quantum rewinding and explore new quantum extraction techniques.  

\subsection{Results}
We introduce and study new techniques that enable us to extract from quantum adversaries. 
\paragraph{Our Notion: Secure Quantum Extraction Protocols.} We formalize this by first introducing the notion of secure quantum extraction protocols. This is a classical interactive protocol between a sender and a receiver and is associated with a NP relation. The sender has an NP instance and a witness while the receiver only gets the NP instance. In terms of properties, we require the following to hold:
\begin{itemize}
   
    \item {\em  Extractability}: An extractor, implemented as a quantum polynomial time algorithm, can extract a valid witness from an adversarial sender. We model the adversarial sender as a quantum polynomial time algorithm that follows the protocol but is allowed to choose its randomness; in the classical setting, this is termed as {\em semi-malicious} and we call this semi-malicious quantum adversaries\footnote{In the literature, this type of semi-malicious adversaries are also referred to as {\em explainable} adveraries.}. 
    \par We also require {\em indistinguishability of extraction}: that is, the adversarial sender cannot distinguish whether it's interacting with the honest receiver or an extractor. In applications, this property is used to argue that the adversary cannot distinguish whether it's interacting with the honest party or the simulator. 

    \item {\em  Zero-Knowledge}: A malicious receiver should not be able to extract a valid witness after interacting with the sender. The malicious receiver can either be a classical probabilistic polynomial time algorithm or a quantum polynomial time algorithm. 
    Correspondingly, there are two notions of quantum extraction protocols we study: quantum extraction protocols secure against quantum adversarial receivers (qQEXT) and quantum extraction protocols secure against classical adversarial receivers (cQEXT).    
    
\end{itemize}

\noindent There are two reasons why we only study extraction against semi-malicious adversaries, instead of malicious adversaries (who can arbitrarily deviate from the protocol): first, even extracting from semi-malicious adversaries turns out to be challenging and we view this as a first step towards extraction from malicious adversaries and second, in the classical setting, there are works that show how to leverage extraction from semi-malicious adversaries to achieve zero-knowledge protocols~\cite{BCPR16,BKP19} or secure two-party computation protocols~\cite{AJ17}.

Quantum extraction protocols are interesting even if we only consider classical adversaries, as they present a new method for proving zero-knowledge. For instance, to demonstrate zero-knowledge, we need to demonstrate a simulator that has a computational capability that a malicious prover doesn't have. Allowing quantum simulators in the classical setting~\cite{KK19} is another way to achieve this asymmetry between the power of the simulator and the adversary besides the few mentioned before (rewinding, superpolynomial, or non-black-box). Furthermore, quantum simulators capture the notion of knowledge that could be learnt if a malicious verifier had access to a quantum computer. \\

\noindent {\em Quantum-Lasting Security.} A potential concern regarding the security of cQEXT protocols is that the classical malicious receiver participating in the cQEXT protocol could later, long after the protocol has been executed, use a quantum computer to learn the witness of the sender from the transcript of the protocol and its own private state. For instance, the transcript could contain an ElGamal encryption of the witness of the sender; while a malicious classical receiver cannot break it, after the protocol is completed, it could later use a quantum computer to learn the witness. This is especially interesting in the event (full-fledged) quantum computers might become available in the future. First introduced by Unruh~\cite{Unruh13}, we study the concept of quantum-lasting security; any quantum polynomial time (QPT) adversary given the transcript and the private state of the malicious receiver, should not be able to learn the witness of the sender. Our construction  will satisfy this security notion and thus our protocol is resilient against the possibility of quantum computers being accessible in the future. 

\paragraph{\underline{Result \#1: Constant Round qQEXT protocols.}} We show the following result. 

\begin{theorem}[Informal]
Assuming quantum hardness of learning with errors and a quantum fully homomorphic encryption scheme (for arbitrary poly-time computations)\footnote{As against leveled quantum FHE,  which can be based on quantum hardness of LWE.}, satisfying, (1) perfect correctness for classical messages and, (2) ciphertexts of poly-sized classical messages have a poly-sized classical description, there exists a constant round quantum extraction protocol secure against quantum poly-time receivers.  
\end{theorem}

\noindent We clarify what we mean by perfect correctness. 
For every public key, every valid fresh ciphertext of a classical message can always be decrypted correctly. Moreover, we require that for every valid fresh ciphertext, of a classical message, the evaluated ciphertext can be decrypted correctly with probability negligibly close to 1. We note that the works of~\cite{mahadev2018classical,Bra18} give candidates for quantum fully homomorphic encryption schemes satisfying both the above properties.
\par En route to proving the above theorem, we introduce a new non black extraction technique in the quantum setting building upon a {\em classical} non-black extraction technique of~\cite{BKP19}. We view identifying the appropriate classical non-black-box technique to also be a contribution of our work. A priori it should not be clear whether classical non-black-box techniques are useful in constructing their quantum analogues. For instance, it is unclear how to utilize the well known non-black-box technique of Barak~\cite{Bar01}; at a high level, the idea of Barak~\cite{Bar01} is to commit to the code of the verifier and then prove using a succinct argument system that either the instance is in the language or it has the code of the verifier. In our setting, the verifier is a quantum circuit which means that we would require succinct arguments for quantum computations which we currently don't know how to achieve. 
\par Non-black-box extraction overcomes the disadvantage quantum rewinding poses in achieving constant round extraction; the quantum rewinding employed by~\cite{Wat09} requires polynomially many rounds (due to sequential repetition) or constant rounds with non-negligible gap between extraction and verification error~\cite{Unruh12}.  
\par This technique was concurrently developed by Bitansky and Shmueli~\cite{BS20} (see ``Comparison with~\cite{BS20}" paragraph) and they critically relied upon this to construct a constant-round zero-knowledge argument system for NP and QMA, thus resolving a long-standing open problem in the round complexity of quantum zero-knowledge. \\

\noindent {\em Subsequent Work.} Many followup works have used the non-black-box extraction technique we introduce in this work to resolve other open problems in quantum cryptography. For instance, our technique was adopted to prove that quantum copy-protection is impossible~\cite{AL20}; resolving a problem that was open for more than a decade. It was also used to prove that quantum VBB for classical circuits is impossible~\cite{AL20,ABDS20}. In yet another exciting follow up work, this technique was developed further to achieve the first constant round post-quantum secure MPC protocol~\cite{ABGKM20}. 

%The novelty of our approach involves identifying the appropriate classical non-black-box extraction technique and then porting it to the quantum setting; in particular, we rely upon the work of~\cite{BKP19} who introduced a new non-black-box technique in the context of designing classical protocols. For instance, it is unclear how to utilize the well known non-black-box technique of Barak~\cite{Bar01}; at a high level, the idea of Barak~\cite{Bar01} is to commit to the code of the verifier and then prove using a succinct argument system that either the instance is in the language or it has the code of the verifier. In our setting, the verifier is a quantum circuit which means that we would require succinct arguments for quantum computations which we currently don't know how to achieve.\\

\paragraph{\underline{Result \#2: Constant Round cQEXT protocols.}} We also present a construction of quantum extraction protocols secure against classical adversaries (cQEXT). This result is incomparable to the above result; on one hand, it is a weaker setting but on the other hand, the security of this construction can solely be based on the hardness of learning with errors.  

\begin{theorem}[Informal]
Assuming quantum hardness of learning with errors, there exists a constant round quantum extraction protocol secure against classical PPT adversaries and satisfying quantum-lasting security.  
\end{theorem}

\noindent Our main insight is to turn the ``test of quantumness" protocol introduced in~\cite{brakerski2018cryptographic} into a quantum extraction protocol using cryptographic tools. In fact, our techniques are general enough that they might be useful to turn any protocol that can verify a quantum computer versus a classical computer into a quantum extraction protocol secure against classical adversaries; the transformation additionally assumes quantum hardness of learning with errors. Our work presents a new avenue for using ``test of quantumness" protocols beyond using them just to test whether the server is quantum or not. 

\par We note that it is conceivable to construct "test of quantumness" protocols from DDH (or any other quantum-{\bf in}secure assumption). The security of the resulting extraction protocol would then be based on DDH and quantum hardness of learning with errors -- the latter needed to argue quantum-lasting security. However, the security of our protocol is solely based on the quantum hardness of learning with errors.

%%%%%%%%%%%%%%%%%%%%%%%%%%%%%%%%%%%%
\begin{comment}
\todo[inline]{This is a good spot to mention that using DDH can be seen as a test of quantumness too.  It will still satisfy ever-lasting security is DDH is used to encode some sort of secret key that is then used inside the SFE.  THe quantum computer can find the secret key, and then get the witness from the SFE.  A classical computer cannot break DDH.  TL;DR  breaking DDH,CDH,Factoring can be seen as tests of quantumness.}
\end{comment}
%%%%%%%%%%%%%%%%%%%%%%%%%%%%%%%%

\paragraph{\underline{ Result \#3: Constant Round QZK for NP with Classical Soundness.}} As an application, we show how to construct constant quantum zero-knowledge argument systems secure against quantum verifiers based on quantum hardness of learning with errors; however, the soundness is still against classical PPT adversaries. %Previously, no such result was known. 
\par Moreover, our protocol satisfies zero-knowledge against quantum  verifiers with arbitrary quantum auxiliary state. Such protocols are also called auxiliary-input zero-knowledge protocols~\cite{GO94} and are necessary for composition. Specifically, our ZK protocol can be composed with other protocols to yield new protocols satisfying quantum security. 

%Our constant round zero-knowledge protocol has a few properties: (1) it is simple, as it is based on the FLS trick~\cite{FLS99}, i.e. we combine the cQEXT with witness indistinguishable proofs to get our desired zero-knowledge construction, (2) it has black-box simulation, and (3) it does not require rewinding of the verifier.

\begin{theorem}[Constant Round Quantum ZK with Classical Soundness; Informal]
Assuming quantum hardness of learning with errors, there exists a constant round black box quantum zero-knowledge system with negligible soundness against classical PPT algorithms. Moreover, our protocol satisfies (quantum) auxiliary-input zero-knowledge property. 
\end{theorem}

\noindent A desirable property from a QZK protocol is if the verifier is classical then the simulator is also classical. While our protocol doesn't immediately satisfy this property, we show, nonetheless, that there is a simple transformation \fullversion{(Section~\ref{sec:class:verifiers})} that converts into another QZK protocol that has this desirable property. 

\paragraph{Application: Authorization with Quantum Cloud.} Suppose Eva wants to convince the cloud services offered by some company that she has the authorization to access a document residing in the cloud. Since the authorization information could leak sensitive information about Eva, she would rather use a zero-knowlede protocol to prove to the cloud that she has the appropriate authorization. While we currently don't have scalable implementations of quantum computers, this could change in the future when organizations (e.g. governments, IBM, Microsoft, etc) could be the first ones to develop a quantum computer. They could in principle then use this to break the zero-knowledge property of Eva's protocol and learn sensitive information about her. In this case, it suffices to  use a QZK protocol but only requiring soundness against malicious classical users; in the nearby future, it is reasonable to assume that even if organizations with enough resources get to develop full-fledged quantum computers, it'll take a while before everyday users will have access to one.

\subsection{Related Work}

\paragraph{Quantum Rewinding.} Watrous~\cite{Wat09} introduced the quantum analogue of the rewinding technique.  Later, Unruh~\cite{Unruh12} introduced yet another notion of quantum rewinding with the purpose of constructing quantum zero-knowledge proofs of knowledge. Unruh's rewinding does have extractability, but it requires that the underlying protocol to satisfy \textit{strict soundness}. Furthermore, the probability that the extractor succeeds is not negligibly close to $1$. The work of~\cite{ARU14} shows that relative to an oracle, many classical zero-knowledge protocols are quantum insecure, and that the strict soundness condition from~\cite{Unruh12} is necessary in order for a sigma protocol to be a quantum proofs of knowledge.

\paragraph{Quantum and Classical Zero-Knowledge.}  Zero-knowledge against quantum adversaries was first studied by Watrous~\cite{Wat09}. He showed how the GMW protocol~\cite{GMW86} for graph 3-colorability is still zero-knowledge against quantum verifiers. Other works~\cite{hallgren2008making,chailloux2008interactive,JKMR06,kobayashi2008general,matsumoto2006simpler,Unruh12} have extended the study of classical protocols that are quantum zero-knowledge, and more recently, Broadbent et al.~\cite{zkqma16} extended the notion of zero-knowledge to QMA languages. By using ideas from \cite{mahadev2018verification} to classically verify quantum computation, the protocol in~\cite{zkqma16} was adapted to obtained classical argument systems for quantum computation in~\cite{vidick2019classical}. All known protocols, with non-negligible soundness error, take non-constant rounds.  
\par On the other hand, zero knowledge proof and argument systems have been extensively studied in classical cryptography. In particular, a series of recent works~\cite{BCPR16,BBKPV16,BKP18,BKP19} resolved the round complexity of zero knowledge argument systems. 

\paragraph{Comparison with~\cite{BS20}.} In a recent exciting work,~\cite{BS20} construct a constant round QZK with soundness against quantum adversaries for NP and QMA. 
\begin{itemize}

\item The non-black-box techniques used in their work was concurrently developed and are similar to the techniques used in our QEXT protocol secure against quantum receivers\footnote{A copy of our QEXT protocol secure against quantum receivers was privately communicated to the authors of~\cite{BS20} on the day of their public posting and our paper was posted online in about two weeks from then~\cite{AL19}.}.  

\item Subsequent to their posting, using completely different techniques, we developed QEXT secure against classical receivers and used it to build a constant round QZK system with classical soundness. There are a few crucial differences between our QZK argument system and theirs:

\begin{enumerate}

\item Our result is based on quantum hardness of learning with errors while their result is based on the existence of quantum fully homomorphic encryption for arbitrary polynomial computations and quantum hardness of learning with errors.

%\item Our ZK protocol is composable while the ZK property in their argument system can only be argued in the presence of specific auxiliary input distributions,

\item The soundness of their argument system is against quantum polynomial time algorithms while ours is only against classical PPT adversaries. 

%\item We show how to construct QZK only for NP, whereas they construct QZK systems for both NP and QMA.
\end{enumerate}
\end{itemize}

%\subsection{Overview of Techniques}

\subsection{Quantum extraction with security against classical receivers: Overview}
\noindent We start with the overview of quantum extraction protocols with security against classical receivers.

\paragraph{Starting Point: Noisy Trapdoor Claw-Free Functions.} Our main idea is to turn the "test of quantumness" from~\cite{brakerski2018cryptographic} into an extraction protocol. Our starting point is a noisy trapdoor claw-free function (NTCF) family~\cite{mahadev2018classical,mahadev2018verification,brakerski2018cryptographic}, parameterized by key space $\cK$, input domain $\cX$ and output domain $\cY$. Using a key $\key \in \cK$, NTCFs allows for computing the functions, denoted by $f_{\key,0}(x) \in \cY$ and $f_{\key,1}(x) \in \cY$~\footnote{The efficient implementation of $f$ only approximately computes $f$ and we denote this by $f'$. We ignore this detail for now.}, where $x \in \cX$. Using a trapdoor $\td$ associated with a key $\key$, any $y$ in the support of $f_{\key,b}(x)$, can be efficiently inverted to obtain $x$. Moreover, there are "claw" pairs $(x
_0,x_1)$ such that $f_{\key,0}(x_0) =f_{\key,1}(x_1)$. Roughly speaking, the security property states that it is computationally hard even for a quantum computer to simultaneously produce $y \in \cY$, values $(b,x_{b})$ and $(d,u)$ such that $f_{\key,b}(x_b)=y$ and $\langle d,J(x_0) \oplus J(x_1) \rangle = u$, where $J(\cdot)$ is an efficienctly computable injective function mapping $\cX$ into bit strings. What makes this primitive interesting is its quantum capability that we will discuss when we recall below the test of~\cite{brakerski2018cryptographic}.

\newcommand{\bits}{\mathbf{a}}
\paragraph{Test of Quantumness~\cite{brakerski2018cryptographic}.} Using NTCFs,~\cite{brakerski2018cryptographic} devised the following test\footnote{As written, this test doesn't have negligible soundness but we can achieve negligible soundness by parallel repetition.}: \begin{itemize} 

\item The classical client, who wants to test whether the server it's interacting with is quantum or classical, first generates a key $\key$ along with a trapdoor $\td$ associated with a noisy trapdoor claw-free function (NTCF) family. It sends $\key$ to the server. 

\item The server responds back with $y \in \cY$.

\item The classical client then sends a {\bf challenge} bit $\bits$ to the server. 

\item If $\bits=0$, the server sends a pre-image $x_b$ along with bit $b$ such that $f_{\key,b}(x_b)=y$. If $\bits=1$, the server sends a vector $d$ along with a bit $u$ satisfying the condition $\langle d,J(x_{0}) \oplus J(x_1) \rangle = u$, where $x_{0},x_1$ are such that $f_{\key,0}(x_0)=f_{\key,1}(x_1)=y$.

\end{itemize}
\noindent The client can check if the message sent by the server is either a valid pre-image or a valid $d$ that is correlated with respect to both the pre-images. 
\par Intuitively, since the (classical) server does not know, at the point when it sends $y$, whether it will be queried for $(b,x_b)$ or $(d,u)$, by the security of NTCFs, it can only answer one of the queries. While the quantum capability of NTCFs allows for a quantum server to maintain a superposition of a claw at the time it sent $y$ and depending on the query made by the verifier it can then perform the appropriate quantum operations to answer the client; thus it will always pass the test. 

\paragraph{From Test of Quantumness to Extraction.} A natural attempt to achieve extraction is the following: the sender takes the role of the client and the receiver takes the role of the server and if the test passes, the sender sends the witness to the receiver. We sketch this attempt below. 
\begin{itemize}
    \item Sender on input instance-witness pair $(\inst,\witness)$ and receiver on input instance $\inst$ run a ``test of quantumness" protocol where the receiver (taking the role of the server) needs to convince the sender (taking the role of the classical client) that it is a quantum computer. 
    \item If the receiver succeeds in the ``test of quantumness" protocol then the sender sender $\witness$, else it aborts. 
\end{itemize}
\noindent Note that a quantum extractor can indeed succeed in the test of quantumness protocol and hence, it would receive $\witness$ while a malicious classical adversary will not. 
\par However, the above solution is not good enough for us. It does not satisfy indistinguishability of extraction: the sender can detect whether it's interacting with a quantum extractor or an honest receiver.

\paragraph{Achieving Indistinguishability of Extraction.} To ensure indistinguishability of extraction, we rely upon a tool called secure function evaluation~\cite{GHV10,BCPR16} that satisfies quantum security. A secure function evaluation (SFE) allows for two parties $P_1$ and $P_2$ to securely compute a function on their inputs in a such a way that only one of the parties, say $P_2$, receives the output of the function. In terms of security, we require that: (i) $P_2$ doesn't get information about $P_1$'s input beyond the output of the function and, (ii)  $P_1$ doesn't get any information about $P_2$'s input (in fact, even the output of the protocol is hidden from $P_1$).
\par The hope is that by combining SFE and test of quantumness protocol, we can guarantee that a quantum extractor can still recover the witness by passing the test of quantumness as before but the sender doesn't even know whether the receiver passed or not. To implement this, we assume a structural property from the underlying test of quantumness protocol: until the final message of the protocol, the client cannot distinguish whether it's talking to a quantum server or a classical server. This structural property is satisfied by the test of quantumness protocol~\cite{brakerski2018cryptographic} sketched above. \par Using this structural property and SFE, here is another attempt to construct a quantum extraction protocol: let the test of quantumness protocol be a $k$-round protocol. 
\begin{itemize}
    \item Sender on input instance-witness pair $(\inst,\witness)$ and receiver on input instance $\inst$ run the first $(k-1)$ rounds of the test of quantumness protocol where the receiver (taking the role of the server) needs to convince the sender (taking the role of the receiver) that it can perform quantum computations. 
    \item Sender and receiver then run a SFE protocol for the following functionality $G$: it takes as input $\witness$ and the first $(k-1)$ rounds of the test of quantumness protocol from the sender, the $k^{th}$ round message from the receiver\footnote{It follows without loss of generality that the server (and thus, the receiver of the quantum extraction protocol) computes the final message of the test of quantumness protocol.} and outputs $\witness$ if indeed the test passed, otherwise output $\bot$. Sender will take the role of $P_1$ and the receiver will take the role of $P_2$ and thus, only the receiver will receive the output of $G$. 
\end{itemize}
\noindent Note that the security of SFE guarantees that the output of the protocol is hidden from the sender and moreover, the first $(k-1)$ messages of the test of quantumness protocol doesn't reveal the information about whether the receiver is a quantum computer or not. These two properties ensure the sender doesn't know whether the receiver passed the test or not. Furthermore, the quantum extractor still succeeds in extracting the witness $\witness$ since it passes the test. 
\par The only remaining property to prove is zero-knowledge. 

\paragraph{Challenges in Proving Zero-Knowledge.} How do we ensure that a malicious classical receiver was not able to extract the witness? The hope would be to invoke the soundness of the test of quantumness protocol to argue this. However, to do this, we need all the $k$ messages of the test of quantumness protocol. 
\par To understand this better, let us recall how the soundness of the test of quantumness works: the client sends a challenge bit $\bits=0$ to the server who responds back with $(b,x_b)$, then the client rewinds the server and instead sends the challenge bit $\bits=1$ and it receives $(d,u)$: this contradicts the security of NTCFs since a classical PPT adversary cannot simultaneously produce both a valid pre-image $(b,x_b)$ and a valid correlation vector along with the prediction bit $(d,u)$. 
\par Since the last message is fed into the secure function evaluation protocol and inaccessible to the simulator, we cannot use this rewinding strategy to prove the zero-knowledge of the extraction protocol.

\paragraph{Final Template: Zero-Knowledge via Extractable Commitments~\cite{PRS02,PW09}.} To overcome this barrier, we force the receiver to commit, using an extractable commitment scheme, to the $k^{th}$ round of the test of quantumness protocol before the SFE protocol begins. An extractable commitment scheme is one where there is an extractor who can extract an input $x$ being committed from the party committing to $x$. Armed with this tool, we give an overview of our construction below. 
\begin{itemize}
    \item Sender on input instance-witness pair $(\inst,\witness)$ and receiver on input instance $\inst$ run the first $(k-1)$ rounds of the test of quantumness protocol where the receiver (taking the role of the server) needs to convince the sender (taking the role of the receiver) that it can perform quantum computations. 
    \item The $k^{th}$ round of the test of quantumness protocol is then committed by the receiver, call it $\mathbf{c}$, using the extractable commitment scheme\footnote{In the technical sections, we use a specific construction of extractable commitment scheme by~\cite{PRS02,PW09} since we additionally require security against quantum adversaries.}.  
    \item Finally, the sender and the receiver then run a SFE protocol for the following functionality $G$: it takes as input $\witness$ and the first $(k-1)$ rounds of the test of quantumness protocol from the sender, the decommitment of $\mathbf{c}$ from the receiver and outputs $\witness$ if indeed the test passed, otherwise output $\bot$. Sender will take the role of $P_1$ and the receiver will take the role of $P_2$ and thus, only the receiver will receive the output of $G$. 
\end{itemize}
\noindent Let us remark about  zero-knowledge since we have already touched upon the other properties earlier. To argue zero-knowledge, construct a simulator that interacts honestly with the malicious receiver until the point the extraction protocol is run. Then, the simulator runs the extractor of the commitment scheme to extract the final message of the test of quantumness protocol. It then rewinds the test of quantumness protocol to the point where the simulator sends a different challenge bit (see the informal description of~\cite{brakerski2018cryptographic} given before) and then runs the extractor of the commitment scheme once again to extract the $k^{th}$ round message of the test of quantumness protocol. Recall that having final round messages corresponding to two different challenge bits is sufficient to break the security of NTCFs; the zero-knowledge property then follows.
\par A couple of remarks about our simulator. Firstly, the reason why our simulator is able to rewind the adversary is because the adversary is a classical PPT algorithm. Secondly, our simulator performs {\em double rewinding} -- not only does the extractor of the commitment scheme perform rewinding but also the test of quantumness protocol is rewound.  
 
\subsection{Constant Round QZK Argument Systems with Classical Soundness}
We show how to use the above  quantum extraction protocol secure against classical receivers (cQEXT) to construct an interactive argument system satisfying classical soundness and quantum ZK.  

\paragraph{From Quantum Extraction to Quantum Zero-Knowledge.} As a starting point, we consider the quantum analogue of the seminal FLS technique~\cite{FLS99} to transform a quantum extraction protocol into a quantum ZK protocol. A first attempt to construct quantum ZK is as follows: let the input to the prover be instance $\inst$ and witness $\witness$ while the input to the verifier is  $\inst$.  
\begin{itemize}
    \item The verifier commits to some trapdoor $\td$. Call the commitment $\bfc$ and the corresponding decommitment $\bfd$.  
    \item The prover and verifier then execute a quantum extraction protocol with the verifier playing the role of the sender, on input $(\bfc,\bfd)$, while the prover plays the role of the receiver on input $\bfc$. 
    \item The prover and the verifier then run a witness-indistinguishable protocol where the prover convinces the verifier that either $\inst$ belongs to the language or it knows $\td$.  
\end{itemize}
\par At first sight, it might seem that the above template should already give us the result we want; unfortunately, the above template is insufficient. The verifier could behave maliciously in the quantum extraction protocol but the quantum extraction protocol only guarantees security against semi-malicious senders. Hence, we need an additional mechanism to protect against malicious receivers. Of course, we require witness-indistinguishability to hold against quantum verifiers and we do know candidates satisfying this assuming quantum hardness of learning with errors~\cite{Blu86,LS19}.   

\paragraph{Handling Malicious Behavior in QEXT.} To check that the verifier behaved honestly in the quantum extraction protocol, we ask the verifier to reveal the inputs and random coins used in the quantum extraction protocol. At this point, the prover can check if the verifier behaved honestly or not. Of course, this would then violate soundness: the malicious prover upon receiving the random coins from the verifier can then recover $\td$ and then use this to falsely convince the verifier to accept its proof. We overcome this by forcing the prover to commit (we again use the extractable commitment scheme of~\cite{PW09}) to some string $\td'$ just before the verifier reveals the inputs and random coins used in the quantum extraction protocol. Then we force the prover to use the committed $\td'$ in the witness-indistinguishable protocol; the prover does not gain any advantage upon seeing the coins of the verifier and thus, ensuring soundness.
\par One aspect we didn't address so far is the aborting issue of the verifier: if the verifier aborts in the quantum extraction protocol, the simulator still needs to produce a transcript indistinguishable from that of the honest prover. Luckily for us, the quantum extraction protocol we constructed before already allows for simulatability of aborting adversaries.
\par To summarise, our ZK protocol consists of the following steps: (i) first, the prover and the verifier run the quantum extraction protocol, (ii) next the prover commits to a string $\td'$ using~\cite{PW09}, (iii) the verifier then reveals the random coins used in the extraction protocol and, (iv) finally, the prover and the verifier run a quantum WI protocol where the prover convinces the verifier that it either knows a trapdoor $\td'$ or that $\inst$ is a YES instance.

\subsection{Quantum extraction with security against quantum receivers: Overview}
We show how to construct extraction protocols where we prove security against quantum receivers. At first sight, it might seem that quantum extraction and quantum zero-knowledge properties are contradictory since the extractor has the same computational resources as the malicious receiver. However, we provide more power to the extractor by giving the extractor non-black-box access to the semi-malicious sender. There is a rich literature on non-black-box techniques in the classical setting starting with the work of~\cite{Bar01}. 

\paragraph{Quantum Extraction via Circular {\bf In}security of QFHE.} The main tool we employ in our protocol is a fully homomorphic encryption $\fhe$ scheme\footnote{Recall that a classical FHE scheme~\cite{Gen09,BV12} allows for publicly evaluating an encryption of a message $x$ using a circuit $C$ to obtain an encryption of $C(x)$.} that allows for public homomorphic evaluation of quantum circuits. Typically, we require a fully homomorphic encryption scheme to satisfy semantic security. However, for the current discussion, we require that $\fhe$ to satisfy a stronger security property called 2-circular {\bf in}security: 
\begin{quote}
    Given $\fhe.\enc(\pk_1,\allowbreak SK_2)$ (i.e., encryption of $SK_2$ under $\pk_1$), $\fhe.\enc(PK_2,\sk_1)$, where $(\pk_1,\sk_1)$ and $(PK_2,SK_2)$ are independently generated public key-secret key pairs, we can efficiently recover $\sk_1$ and $SK_2$.
    \end{quote}
    Later, we show how to get rid of 2-circular {\bf in}security property by using lockable obfuscation~\cite{GKW17,WZ17}.  Here is our first attempt to construct the extraction protocol: 
\begin{itemize}

    \item The sender, on input instance $\inst$ and witness $\witness$, sends three ciphertexts:  $\ct_1 \leftarrow \fhe.\enc(\pk_1,\allowbreak \td)$, $\ct_2 \leftarrow \fhe.\enc(\pk_1,\witness)$ and $\ct_3 \leftarrow \fhe.\enc(PK_2,\allowbreak \sk_1)$. 
    \item The receiver sends $\td'$. 
    \item If $\td'=\td$ then the sender sends $SK_2$. 
\end{itemize}
\noindent A quantum extractor with non-black-box access to the private (quantum) state of the semi-malicious sender $S$ does the following: 
\begin{itemize}
    \item It first encrypts the private (quantum) state of $S$ under public key $\pk_1$.
    \item Here is our main insight: the extractor can homomorphically evaluate the next message function of $S$ on $\ct_1$ and the encrypted state of $S$. The result is $\ct_1^* = \fhe.\enc(\pk_1,S(\td))$. But note that $S(\td)$ is nothing but $SK_2$; note that $S$ upon receiving $\td'=\td$ outputs $SK_2$. Thus, we have $\ct^*_1=\fhe.\enc(\pk_1,SK_2)$.
    \item Now, the extractor has both $\ct_3 = \fhe.\enc(PK_2,\sk_1)$ and $\ct_1^*=\fhe.\enc(\pk_1,\allowbreak SK_2)$. It can then use the circular {\bf in}security of $\fhe$ to recover $\sk_1,SK_2$. 
    \item Finally, it decrypts $\ct_2$ to obtain the witness $\witness$!
\end{itemize}
\noindent The correctness of extraction alone is not sufficient; we need to argue that the sender cannot distinguish whether it's interacting with the honest receiver or the extractor. This is not true in our protocol since the extractor will always compute the next message function of $S$ on $\td'=\td$ whereas an honest receiver will send $\td'=\td$ only with negligible probability. 

\paragraph{Indistinguishability of Extraction: SFE strikes again.} We already encountered a similar issue when we were designing extraction protocols with security against classical receivers and the tool we used to solve that issue was secure function evaluation (SFE); we will use the same tool here as well. 
\par Using SFE, we make another attempt at designing the quantum extraction protocol. 
\begin{itemize}
    \item The sender, on input instance $\inst$ and witness $\witness$, sends three ciphertexts:  $\ct_1 \leftarrow \fhe.\enc(\pk_1,\allowbreak \td)$, $\ct_2 \leftarrow \fhe.\enc(\pk_1,\witness)$ and $\ct_3 \leftarrow \fhe.\enc(\allowbreak PK_2,\sk_1)$. 
    \item The sender and the receiver executes a secure two-party computation protocol, where the receiver feeds $\td'$ and the sender feeds in $(\td,\witness)$. After the protocol finishes, the receiver recovers $\witness$ if $\td'= \td$, else it recovers $\bot$. The sender doesn't receive any output.   
\end{itemize}
The above template guarantees indistinguishability of extraction property\footnote{There is a subtle point here that we didn't address: the transcript generated by the extractor is encrypted under $\fhe$. But after recovering the secret keys, the extractor could decrypt the encrypted transcript.}. 
\par We next focus on zero-knowledge. To do this, we need to argue that the $\td'$ input by the malicious receiver can never be equal to $\td$. One might falsely conclude that the semantic security of $\fhe$ would imply that $\td$ is hidden from the sender and hence the argument follows. This is not necessarily true; the malicious receiver might be able to ``maul" the ciphertext $\ct_1$ into the messages of the secure function evaluation protocol in such a way that the implicit input committed by the receiver is $\td'$. We need to devise a mechanism to prevent against such mauling attacks. 

\paragraph{Preventing Mauling Attacks.} We prevent the mauling attacks by forcing the receiver to commit to random strings $(r_1,\ldots,r_{\ell})$ in the first round, where $|\td|=\ell$, even before it receives the ciphertexts $(\ct_1,\ct_2,\ct_3)$ from the sender. Once it receives the ciphertexts, the receiver is supposed to commit to every bit of the trapdoor using the randomness $r_1,\ldots,r_{\ell}$; that is, the $i^{th}$ bit of $\td$ is committed using $r_i$. 
\par Using this mechanism, we can then provably show that if the receiver was able to successfully maul the $\fhe$ ciphertext then it violates the semantic security of $\fhe$ using a non-uniform adversary. 

\paragraph{Replacing Circular {\bf In}security with Lockable Obfuscation~\cite{GKW17,WZ17}.} While the above protocol is a candidate for quantum extraction protocol secure against quantum receivers; it is still unsatisfactory since we assume a quantum FHE scheme satisfying 2-circular {\bf in}security. We show how to replace 2-circular insecure QFHE with {\em any} QFHE scheme (satisfying some mild properties already satisfied by existing candidates) and lockable obfuscation for classical circuits. A lockable obfuscation scheme is an obfuscation scheme for a specific class of functionalities called compute-and-compare functionalities; a compute-and-compare functionality is parameterized by  $C,\alpha$ (lock), $\beta$ such that on input $x$, it outputs $\beta$ if $C(x)=\alpha$. As long as $\alpha$ is sampled uniformly at random and independently of $C$, lockable obfuscation completely hides the circuit $C$, $\alpha$ and $\beta$. The idea to replace 2-circular insecure QFHE with lockable obfuscation\footnote{It shouldn't be too surprising that lockable obfuscation can be used to replace circular insecurity since one of the applications~\cite{GKW17,WZ17} of lockable obfuscation was to demonstrate counter-examples for circular security,} is as follows: obfuscate the circuit, with secret key $SK_2$, ciphertext $\fhe.\enc(SK_2,r)$ hardwired, that takes as input $\fhe.\enc(\pk_1,SK_2)$, decrypts it to obtain $SK'_2$, then decrypts $\fhe.\enc(SK_2,r)$ to obtain $r'$ and outputs $\sk_1$ if $r'=r$. If the adversary does not obtain $\fhe.\enc(\pk_1,SK_2)$ then we can first invoke the security of lockable obfuscation to remove $\sk_1$ from the obfuscated circuit and then it can replace $\fhe.\enc(\pk_1,\witness)$ with $\fhe.\enc(\pk_1,\bot)$. The idea of using fully homomorphic encryption along with lockable obfuscation to achieve non-black-box extraction was first introduced, in the classical setting, by~\cite{BKP19}.
\par Unlike our cQEXT construction, the non-black-box technique used for qQEXT does not directly give us a constant round quantum zero-knowledge protocol for NP.  This is because an adversarial verifier that aborts can distinguish between the extractor or the honest prover (receiver in qQEXT). The main issue is that the extractor runs the verifier homomorphically, so it cannot detect if the verifier aborted at any point in the protocol without decrypting. But if the verifier aborted, the extractor wouldn't be able to decrypt in the first place -- it could attempt to rewind but then this would destroy the initial quantum auxiliary state.

\section{Preliminaries}
\label{sec:prelims}
\noindent We denote the security parameter by $\secparam$. We denote (classical) computational indistiguishability of two distributions $\distr_0$ and $\distr_1$ by $\distr_0 \approx_{c,\varepsilon} \distr_1$. In the case when $\varepsilon$ is negligible, we drop $\varepsilon$ from this notation. 

\paragraph{Languages and Relations.} A language $\lang$ is a subset of $\{0,1\}^*$. A relation $\rel$ is a subset of $\{0,1\}^* \times \{0,1\}^*$. We use the following notation:
\begin{itemize}

\item Suppose $\rel$ is a relation. We define $\rel$ to be {\em efficiently decidable} if there exists an algorithm $A$ and fixed polynomial $p$ such that $(x,w) \in \rel$ if and only if $A(x,w)=1$ and the running time of $A$ is upper bounded by $p(|x|,|w|)$. 

\item Suppose $\rel$ is an efficiently decidable relation. We say that $\rel$ is a NP relation if $\lang(\rel)$ is a NP language, where $\lang(\rel)$ is defined as follows: $x \in \lang(R)$ if and only if there exists $w$ such that $(x,w) \in \rel$ and $|w| \leq p(|x|)$ for some fixed polynomial $p$. 

\end{itemize}

\subsection{Learning with Errors}
\label{sec:prelims:lwe}

\noindent In this work, we are interested in the decisional learning with errors (LWE) problem. This problem, parameterized by $n,m,q,\chi$, where $n,m,q \in \mathbb{N}$, and for a distribution $\chi$ supported over $\mathbb{Z}$ is to distinguish between the distributions $(\bfA,\bfA \bfs + \bfe)$ and $(\bfA,\bfu)$, where $\bfA \xleftarrow{\$} \mathbb{Z}_q^{m \times n},\bfs \xleftarrow{\$} \mathbb{Z}_q^{n \times 1},\bfe \xleftarrow{\$} \chi^{m \times 1}$ and $\bfu \leftarrow \mathbb{Z}_q^{m \times 1}$. Typical setting of $m$ is $n \log(q)$, but we also consider $m=\poly(n \log(q))$. 
\par We base the security of our constructions on the quantum hardness of learning with errors problem.

\subsection{Notation and General Definitions}
\label{ssec:notation}

For completeness, we present some of the basic quantum definitions, for more details see \cite{nielsen2002quantum}.
\paragraph{Quantum states and channels.} Let $\cH$ be any finite Hilbert space, and let $L(\cH):=\{\cE:\cH \rightarrow \cH \}$ be the set of all linear operators from $\cH$ to itself (or endomorphism). Quantum states over $\cH$ are the positive semidefinite operators in $L(\cH)$ that have unit trace. Quantum channels or quantum operations acting on quantum states over $\cH$ are completely positive trace preserving (CPTP) linear maps from $L(\cH)$ to $L(\cH')$ where $\cH'$ is any other finite dimensional Hilbert space.

A state over $\cH=\mathbb{C}^2$ is called a qubit. For any $n \in \mathbb{N}$, we refer to the quantum states over $\cH = (\mathbb{C}^2)^{\otimes n}$ as $n$-qubit quantum states. To perform a standard basis measurement on a qubit means projecting the qubit into $\{\ket{0},\ket{1}\}$. A quantum register is a collection of qubits. A classical register is a quantum register that is only able to store qubits in the computational basis.

A unitary quantum circuit is a sequence of unitary operations (unitary gates) acting on a fixed number of qubits. Measurements in the standard basis can be performed at the end of the unitary circuit. A (general) quantum circuit is a unitary quantum circuit with $2$ additional operations: $(1)$ a gate that adds an ancilla qubit to the system, and $(2)$ a gate that discards (trace-out) a qubit from the system. A quantum polynomial-time algorithm (QPT) is a uniform collection of quantum circuits $\{C_n\}_{n \in \mathbb{N}}$.

\paragraph{Quantum Computational Indistinguishability.}

When we talk about quantum distinguishers, we need the following definitions, which we take from \cite{Wat09}.
\begin{definition}[Indistinguishable collections of states] Let $I$ be an infinite subset $I \subset \{0,1\}^*$, let $p : \mathbb{N} \rightarrow \mathbb{N}$ be a polynomially bounded function, and let $\rho_{x}$ and $\sigma_x$ be $p(|x|)$-qubit states. We say that $\{\rho_{x}\}_{x \in I}$ and $\{\sigma_x\}_{x\in I}$ are \textbf{quantum computationally indistinguishable collections of quantum states} if for every QPT $\cE$ that outputs a single bit, any polynomially bounded  $q:\mathbb{N}\rightarrow \mathbb{N}$, and any auxiliary $q(|x|)$-qubits state $\nu$, and for all $x \in I$, we have that
$$\left|\Pr\left[\cE(\rho_x\otimes \nu)=1\right]-\Pr\left[\cE(\sigma_x \otimes \nu)=1\right]\right| \leq \epsilon(|x|) $$
for some negligible function $\epsilon:\mathbb{N}\rightarrow [0,1]$. We use the following notation 
$$\rho_x \approx_{Q,\epsilon} \sigma_x$$
and we ignore the $\epsilon$ when it is understood that it is a negligible function.
\end{definition}

\begin{definition}[Indistinguishability of channels] Let $I$ be an infinite subset $I \subset \{0,1\}^*$, let $p,q: \mathbb{N} \rightarrow \mathbb{N}$ be polynomially bounded functions, and let $\cD_x,\cF_x$
be quantum channels mapping $p(|x|)$-qubit states to $q(|x|)$-qubit states. We say that $\{\cD_x\}_{x \in I}$ and $\{\cF_x\}_{x \in I}$ are \textbf{quantum computationally indistinguishable collection of channels} if for every QPT $\cE$ that outputs a single bit, any polynomially bounded $t : \mathbb{N} \rightarrow \mathbb{N}$, any $p(|x|)+t(|x|)$-qubit quantum state $\rho$, and for all $x\in I$, we have that
$$ \left|\Pr\left[\cE\left((\cD_x\otimes \Id)(\rho)\right)=1\right]-\Pr\left[\cE\left((\cF_x\otimes \Id)(\rho)\right)=1\right]\right|\leq \epsilon(|x|) $$
for some negligible function $\epsilon:\mathbb{N}\rightarrow [0,1]$. We will use the following notation
$$ \cD_x(\cdot) \approx_{Q,\epsilon} \cF_x(\cdot)$$
and we ignore the $\epsilon$ when it is understood that it is a negligible function.

\end{definition}

\paragraph{Interactive Models.} We model an interactive protocol between a prover, $\prvr$, and a verifier, $\vrfr$, as follows. There are 2 registers $\RegP$ and $\RegV$ corresponding to the prover's and the verifier's private registers, as well as a message register, $\RegM$, which is used by both $\prvr$ and $\vrfr$ to send messages. In other words, both prover and verifier have access to the message register. We denote the size of a register $\mathsf{R}$ by $|\mathsf{R}|$ -- this is the number of bits or qubits that the register can store.  We will have 2 different notions of interactive computation. Our honest parties will perform classical protocols, but the adversaries will be allowed to perform quantum protocols with classical messages.

\begin{enumerate}
    \item \textbf{Classical protocol:} An interactive protocol is classical if $\RegP$, $\RegV$, and $\RegM$ are classical, and $\prvr$ and $\vrfr$ can only perform classical computation.
    \item \textbf{Quantum protocol with classical messages:} An interactive protocol is quantum with classical messages if either one of $\RegP$ or $\RegV$ is a quantum register, and $\RegM$ is classical. $\prvr$ and $\vrfr$ can perform quantum computations if their respective private register is quantum, but they can only send classical messages.
\end{enumerate}
When a protocol has classical messages, we can assume that the adversarial party will also send classical messages. This is without loss of generality, because the honest party can enforce this condition by always measuring the message register in the computational basis before proceeding with its computations.

\newcommand{\qst}{\mathsf{qst}}

\paragraph{Non-Black-Box Access.} Let $S$ be a QPT party (e.g. either prover or verifier in the above descriptions) involved in specific quantum protocol. In particular, $S$ can be seen as a collection of QPTs, $S=(S_1,...,S_{\ell})$, where $\ell$ is the number of rounds of the protocol, and $S_i$ is the quantum operation that $S$ performs on the $i$th round of the protocol.

We say that a QPT $Q$ has \textit{non-black-box access} to $S$, if $Q$ has access to an efficient classical description for the operations that $S$ performs in each round, $(S_1,...,S_{\ell})$, as well as access to the initial auxiliary inputs of $S$.

%We will use  $\qst_{S}[i]$ to denote the private state of $S$ after performing $S_i$. In this notation, $\qst_{S}[0]$ corresponds to the auxiliary input to $S$. 

\paragraph{Interaction Channel.} For a particular protocol $(\prvr, \vrfr)$, the interaction between $\prvr$ and $\vrfr$ on input $\inst$ induces a quantum channel $\cE_{\inst}$  acting on their private input states, $\rho_{\prvr}$ and $\sigma_{\vrfr}$. We denote the view of $\vrfr$ when interacting with $\prvr$  by $$\view_{\vrfr}\left(\left\langle \prvr\left(\inst, \rho_\prvr\right),\vrfr\left(\inst, \sigma_{\vrfr}\right) \right \rangle \right),$$ and this view is defined as the verifiers output. Specifically, 
$$ \view_{\vrfr}\left(\left\langle \prvr\left(\inst, \rho_\prvr\right),\vrfr\left(\inst, \sigma_{\vrfr}\right) \right \rangle \right) :=\tr_{\RegP} \left[\cE_{\inst}\left(\rho_\prvr \otimes  \sigma_\vrfr\right)\right]. $$ From the verifier's point of view, the interaction induces the channel $\cE_{\inst,V}(\sigma)=\cE_{\inst}(\sigma \otimes \rho_{\prvr})$ on its private input state.

\fullversion{
\subsection{Perfectly Binding Commitments}
\label{sec:prelims:commit}
\noindent A commitment scheme consists a classical PPT algorithm\footnote{Typically, commitment schemes are also associated with a opening algorithm; we don't use the opening algorithm in our work.} $\comm$ that takes as input security parameter $1^{\secparam}$, input message $ x$ and outputs the commitment $\bfc$. There are two properties that need to be satisfied by a commitment scheme: binding and hiding. In this work, we are interested in commitment schemes that are perfectly binding and computationally hiding; we define both these notions below. We adapt the definition of computational hiding to the quantum setting.  
 
\begin{definition}[Perfect Binding]
A commitment scheme $\comm$ is said to be perfectly binding if for every security parameter $\secparam \in \mathbb{N}$, there does not exist two messages $x,x'$ with $x \neq x'$ and randomness $r,r'$ such that $\comm(1^{\secparam},x;r)=\comm(1^{\secparam},x';r')$.  
\end{definition}

\begin{definition}[Quantum-Computational Hiding]
A commitment scheme $\comm$ is said to be computationally hiding if for sufficiently large security parameter $\secparam \in \mathbb{N}$, for any two messages $x,x'$, the following holds: 
$$\left\{ \comm \left( 1^{\secparam},x \right) \right\} \approx_{Q} \left\{ \comm \left( 1^{\secparam},x' \right) \right\} $$
\end{definition}

\paragraph{Instantiation.} A construction of perfectly binding non-interactive commitments was presented in the works of~\cite{GHKW17,LS19} assuming the hardness of learning with errors. Thus, we have the following: 

\begin{lemma}[\cite{GHKW17,LS19}]
Assuming the quantum hardness of learning with errors, there exists a construction of perfectly binding quantum-computational hiding non-interactive commitment schemes. 
\end{lemma}

\subsection{Noisy Trapdoor Claw-Free Functions}
\label{sec:ntcfs}
Noisy trapdoor claw-free functions is a useful tool in quantum cryptography. Most notably, they are a key ingredient in the construction of certifiable randomness protocols \cite{brakerski2018cryptographic}, classical client quantum homomorphic encryption \cite{mahadev2018classical}, and classifal verification of quantum computation \cite{mahadev2018verification}.  We present the formal definition directly from \cite{brakerski2018cryptographic}.

\begin{definition}[Noisy Trapdoor Claw-Free Functions] Let $\cX$ and $\cY$ be finite sets, let $D_{\cY}$ be the set of distributions over $\cY$, and let $\cK$ be a finite set of keys. A collection of functions $\{f_{\key,b}:\cX \rightarrow D_{\cY} \}_{\key \in \cK, b \in \{0,1\}}$ is noisy trapdoor claw-free if
\begin{itemize}
    \item \textbf{(Key-Trapdoor Generation):}  There is a PPT $\gen(1^\secparam)$ to generate a key and a corresponding trapdoor, $\key,\td_{\key} \leftarrow \gen(1^\secparam)$.
    \item For all $\key \in \cK$
    \begin{itemize}
        \item \textbf{(Trapdoor):} For all $b\in\{0,1\}$, and any distinct $x,x'\in \cX$, we have that $\supp(f_{\key,b}(x)) \cap \supp(f_{\key,b}(x')) = \emptyset$. There is also an efficient deterministic algorithm $\inv$, that for any $y \in \supp(f_{\key,b}(x))$, outputs $x \leftarrow \inv(\td_{\key},b,y)$.
        \item \textbf{(Injective Pair):} There exists a perfect matching $\cR_\key \subseteq \cX \times \cX$ such that $f_{\key,0}(x_0)=f_{\key,1}(x_1)$ if and only if $(x_0,x_1) \in \cR_{\key}$
    \end{itemize}
    \item \textbf{(Efficient Range Superposition):} For all $\key \in \cK$ and $b \in \{0,1\}$, there exists functions $f'_{\key,b}:\cX \rightarrow D_{\cY}$ such that the following holds.
    \begin{itemize}
        \item For all $(x_0,x_1) \in \cR_{\key}$, and all $y \in \supp(f'_{\key,b}(x_b))$, the inversion algorithm still works, i.e. $x_b \leftarrow \inv(\td_{\key},b,y)$ and $x_{b\oplus1} \leftarrow \inv(\td_{\key}, b\oplus1,y)$.
        \item There is an efficient deterministic checking algorithm $\chk:\cK \times \{0,1\} \times \cX \times \cY \rightarrow \{0,1\}$ such that $\chk(\key,b,x,y)=1$ iff $y \in \supp(f'_{\key,b}(x))$
        \item For every $\key \in \cK$ and $b \in \{0,1\}$, $$\underset{x \leftarrow \cX}{\mathbb{E}}\left(H^2\left(f_{\key,b}(x),f'_{\key,b}(x)\right)\right) \leq \mu(\secparam) $$  
        for some negligible function $\mu$, and where $H^2$ is the Hellinger distance.
        \item For any $\key \in \cK$ and $b \in \{0,1\}$, there exists an efficient way to prepare the superposition $$\frac{1}{\sqrt{|\cX|}}\underset{x\in \cX, y\in \cY}{\sum}\sqrt{f'_{\key,b}(x)(y)} \ket{x}\ket{y} $$
    \end{itemize}
    
    \item \textbf{(Adaptive Hardcore Bit)}: for all keys $\key \in \cK$, for some polynomially bounded $w:\mathbb{N}\rightarrow \mathbb{N}$, the following holds.
    \begin{itemize}
        \item For all $b \in \{0,1\}$ and for all $x \in \cX$ there exists a set $G_{\key,b,x} \subseteq \{0,1\}^{w(\secparam)}$, s.t. $\underset{d\leftarrow \{0,1\}^{w(\secparam)}}{\Pr}\left[d \notin G_{\key,b,x} \right] \leq \text{negl}(\secparam)$. Furthermore, membership in $G_{\key,b,x}$ can be checked given $t_\key, \key, b$ and $x$.
        \item There is an efficiently computable injection $J:\cX \rightarrow \{0,1\}^{w(\secparam)}$, that can be inverted efficiently in its range, and for which the following holds. Let
        $$H_\key := \left\{\left(b, x_b, d, d\cdot \left(J(x_0)\oplus J(x_1)\right)\right) | b\in \{0,1\},(x_0,x_1) \in \cR_{\key}, d \in G_{\key,0,x_0} \cap G_{\key,1,x_1} \right\}  $$
        $$\overline{H_\key} := \left\{\left(b, x_b, d, c\right) | (b,x,d,c \oplus1) \in H_\key \right\}  $$
        For any QPT $\cA$ there is a negligible function $\mu$ s.t.
        $$\left|\underset{\key,\td_\key}{\Pr}\left[\cA(\key) \in H_\key\right]-\underset{\key,\td_\key}{\Pr}\left[\cA(\key) \in \overline{H_\key}\right] \right|\leq \mu(\secparam)$$
        
    \end{itemize}
\end{itemize}
\end{definition}

\paragraph{Instantiation.}  The work of~\cite{brakerski2018cryptographic} presented a construction of noisy trapdoor claw-free functions from learning with errors.

\subsection{Quantum Fully Homomorphic Encryption}
\label{ssec:prelims:qfhe}
Quantum Homomorphic Encryption schemes have the same syntax as traditional classical homomorphic encryption schemes, but are extended to support quantum operations and to allow plaintexts and ciphertexts to be quantum states. We take our definition directly from \cite{broadbent2015quantum}.

\begin{definition}
A quantum fully homomorphic encryption scheme is a tuple of QPT $\fhe=(\gen,\enc,\dec,\eval)$ satisfying
\begin{itemize}
    \item $\fhe.\gen(1^\secparam)$: outputs a a public and a secret key, $(\pk,\sk)$, as well as a quantum state $\rho_{evk})$, which can serve as an evaluation key.
    \item $\fhe.\enc(\pk,\cdot):L(\cM)\rightarrow L(\cC)$:  takes as input a qubit $\rho$ and outputs a ciphertext $\sigma$
    \item $\fhe.\dec(\sk,\cdot):L(\cC)\rightarrow L(\cM)$: takes a quantum ciphertext $\sigma$ in  correct, and outputs a qubit $\rho$ in the message space $L(\cM)$.
    \item $\fhe.\eval(\cE, \cdot ):L(\cR_{evk}\otimes \cC^{\otimes n})\rightarrow L(\cC^{\otimes m})$: takes as input a quantum circuit $\cE: L(\cC^{\otimes n}) \rightarrow L(\cC^{\otimes m})$, and a ciphertext in $L(\cC^{\otimes n})$ and outputs a ciphertext in $L(\cC^{\otimes m})$, possibly consuming the evaluation key $\rho_{evk}$ in the proccess.

\end{itemize}
\end{definition}
\noindent Semantic security and compactness are defined analogously to the classical setting, and we defer to~\cite{broadbent2015quantum} for a definition.
\noindent We require an $\fhe$ scheme to satisfy the following properties. 

\paragraph{(Perfect) Correctness of classical messages.} We require the following properties to hold: for every quantum circuit $\cE$ acting on $\ell$ qubits, message $x$, every $r_1,r_2 \in \{0,1\}^{\poly(\secparam)}$,
\begin{itemize}
    \item $\prob[x \leftarrow \fhe.\dec(\sk,\fhe.\enc(\pk,x)):(\pk,\sk) \leftarrow \fhe.\gen(1^{\secparam})]=1$
    \item $\prob[\fhe.\dec(\sk,\fhe.\eval(\pk,\cE,\ct))=\cE(x)] \geq 1-\negl(\secparam)$, for some negligible function $\negl$, where: (1) $(\pk,\sk) \leftarrow \fhe.\setup(1^{\secparam};r_1)$ and, (2) $\ct \leftarrow \fhe.\enc(\pk,x;r_2)$. The probability is defined over the randomness of the evaluation procedure.  
\end{itemize}

\paragraph{Instantiation.} The works of~\cite{mahadev2018classical,Bra18} give lattice-based candidates for quantum fully homomorphic encryption schemes; we currently do not know how to base this on learning with errors alone\footnote{Brakerski~\cite{Bra18} remarks that the security of their candidate can be based on a circular security assumption that is also used to argue the security of existing constructions of unbounded depth multi-key FHE~\cite{CM15,MW16,PS16,BP16}.}. There are two desirable propertiess required from the quantum FHE schemes and the works of~\cite{mahadev2018classical,Bra18} satisfy both of them. We formalize them in the lemma below. 

\begin{lemma}[\cite{mahadev2018classical,Bra18}]
There is a quantum fully homomorphic encryption scheme that satisfies: (1) perfect correctness of classical messages and, (2) ciphertexts of classical poly-sized messages have a poly-sized classical description. 
\end{lemma}

\subsection{Quantum-Secure Function Evaluation}
\label{ssec:prelims:sfe}
As a building block in our construction, we consider a secure function evaluation protocol~\cite{GHV10} for classical functionalities. A secure function evaluation protocol is a two message two party secure computation protocol; we designate the parties as sender and receiver (who receives the output of the protocol). Unlike prior works, we require the secure function evaluation protocol to be secure against polynomial time quantum adversaries.

\paragraph{Security.} We require malicious (indistinguishability) security against a quantum adversary $\receiver$ and semantic security against a quantum adversary $\sender$. We define both of them below. 
\par First, we define an indistinguishability security notion against malicious $\receiver$. To do that, we employ an extraction mechanism to extract $\receiver$'s input $\inpl{1}^*$. We then argue that $\receiver$ should not be able to distinguish whether $\sender$ uses $\inpl{2}^0$ or $\inpl{2}^1$ in the protocol as long as $\fc(\inpl{1}^*,\inpl{2}^0)=\fc(\inpl{1}^*,\inpl{2}^1)$. We don't place any requirements on the computational complexity of the extraction mechanism. 

\begin{definition}[Indistinguishability Security: Malicious Quantum $\receiver$]
\label{def:indmalp1}
	Consider a secure function evaluation protocol for a functionality $\fc$ between a sender $\sender$ and a receiver $\receiver$. We say that the secure evaluation protocol satisfies {\bf indistinguishability security against malicious $\receiver^*$} if for every adversarial QPT  $\receiver^*$, there is an extractor $\extractor$ (not necessarily efficient) such the following holds. Consider the following experiment: \\
	
\noindent \uline{$\expt(1^{\secparam},b)$}:
\begin{itemize}

\item $\receiver^*$ outputs the first message $\msg_1$.

\item Extractor $\extractor$ on input $\msg_1$ outputs $\inpl{1}^*$.  

\item Let $\inpl{2}^0,\inpl{2}^1$ be two inputs such that $\fc(\inpl{1}^*,\inpl{2}^0)=\fc(\inpl{1}^*,\inpl{2}^1)$. Party $\sender$ on input $\msg_1$ and $\inpl{2}^b$, outputs the second message $\msg_2$. 

\item $\receiver^*$ upon receiving the second message outputs a bit $\otput$. 
\item Output $\otput$.
	
\end{itemize}
We require that, 
$$\left| \prob[1 \la \expt(1^{\secparam},0)] - \prob[1 \la \expt(1^{\secparam},1)]\right| \leq \negl(\secparam),$$
for some negligible function $\negl$. 
\end{definition}

\noindent We now define semantic security against $\sender$. We insist that $\sender$ should not be able to distinguish which input $\sender$ used to compute its messages. Note that $\sender$ does not get to see the output recovered by the receiver. 

\begin{definition}
	[Semantic Security against Quantum $\sender^*$]
	\label{def:indmalp2}
	Consider a secure function evaluation protocol for a functionality $\fc$ between a sender $\sender$ and a receiver $\receiver$ where $\receiver$ gets the output. We say that the secure function evaluation protocol satisfies {\bf semantic security against $\sender^*$} if for every adversarial QPT $\sender^*$, the following holds: Consider two strings $\inpl{1}^0$ and $\inpl{2}^1$. Denote by $\distr_b$ the distribution of the first message (sent to $\sender^*$) generated using $\inpl{1}^b$ as $\receiver$'s input. The distributions $\distr_0$ and $\distr_1$ are computationally indistinguishable.
\end{definition}

\paragraph{Instantiation.} A secure function evaluation protocol can be built from garbled circuits and oblivious transfer that satisfies indistinguishability security against malicious receivers. Garbled circuits can be based on the hardness of learning with errors by suitably instantiating the symmetric encryption in the construction of Yao's garbled circuits~\cite{Yao86} with one based on the hardness of learning with errors~\cite{Reg05}. Oblivious transfer with indistinguishability security against malicious receivers based on learning with errors was presented in a recent work of Brakerski et al.~\cite{BD18}. Thus, we have the following lemma. 

\begin{lemma}[\cite{Yao86,Reg05,BD18}]
Assuming the quantum hardness of learning with errors, there exists a quantum-secure function evaluation protocol for polynomial time classical functionalities. 
\end{lemma}

\subsection{Lockable Obfuscation}
\label{ssec:prelims:lobfs}
\noindent We first recall the definition of circuit obfuscation schemes~\cite{BGIRSVY01}. A circuit obfuscation scheme associated with the class of circuits $\cktclass$ consists of the classical PPT algorithms $(\lobf,\leval)$ defined below:  
\begin{itemize}
    \item {\bf Obfuscation,}  $\lobf(1^{\secparam},C)$: it takes as input the security parameter $\secparam$, circuit $C$ and produces an obfuscated circuit $\obfC$.  
    \item {\bf Evaluation,} $\leval(\obfC,x)$: it takes as input the obfuscated circuit $\obfC$, input $x$ and outputs $y$. 
    
\end{itemize}

\label{revisit perfect correctness.}
\paragraph{Perfect Correctness.} A program obfuscation scheme $(\lobf,\leval)$ is said to be correct if for every circuit $C\in \cktclass$ with $C:\{0,1\}^{\ell_{in}} \rightarrow \{0,1\}^{\ell_{out}}$, for every input $x \in \{0,1\}^{\ell_{in}}$, we have $\obfC(x)=C(x)$. \\

\noindent We are interested in program obfuscation schemes that are (i) defined for a special class of circuits called compute-and-compare circuits and, (ii) satisfy distributional virtual black box security notion~\cite{BGIRSVY01}. Such obfuscation schemes were first introduced by~\cite{WZ17,GKW17} and are called lockable obfuscation schemes. We recall their definition, adapted to quantum security, below. 

\begin{definition}[Quantum-Secure Lockable Obfuscation]
An obfuscation scheme $(\lobf,\leval)$ for a class of circuits $\cktclass$ is said to be a \textbf{quantum-secure lockable obfuscation scheme} if the following properties are satisfied: 
\begin{itemize}
    \item It satisfies the above mentioned correctness property. 
    \item {\bf Compute-and-compare circuits}: Each circuit $\lockC$ in $\cktclass$ is parameterized by strings $\alpha \in \{0,1\}^{\poly(\secparam)},\beta \in \{0,1\}^{\poly(\secparam)}$ and a poly-sized circuit $C$ such that on every input $x$, $\lockC(x)$ outputs $\beta$ if and only if $C(x)=\alpha$. 
    \item {\bf Security}: For every polynomial-sized circuit $C$, string $\beta \in \{0,1\}^{\poly(\secparam)}$,for every QPT adversary $\adversary$ there exists a QPT simulator $\simulator$ such that the following holds: sample $\alpha \xleftarrow{\$}  \{0,1\}^{\poly(\secparam)}$,
    $$\left\{ \lobf \left( 1^{\secparam},\lockC \right) \right\} \approx_{Q,\varepsilon} \left\{\simulator\left(1^{\secparam},1^{|C|} \right) \right\},$$
    where $\lockC$ is a circuit parameterized by $C,\alpha,\beta$ with $\varepsilon \leq \frac{1}{2^{|\alpha|}}$.  
\end{itemize}
%\prab{parameterize $\approx_{Q,\varepsilon}$ with $\varepsilon$ in the original deinfition and say that $\varepsilon$ is ignored if $\varepsilon=\negl$}

\end{definition}

\paragraph{Instantiation.} The works of~\cite{WZ17,GKW17,GKVW19} construct a lockable obfuscation scheme based on polynomial-security of learning with errors (see Section~\ref{sec:prelims:lwe}). Since learning with errors is conjectured to be hard against QPT algorithms, the obfuscation schemes of~\cite{WZ17,GKW17,GKVW19} are also secure against QPT algorithms. Thus, we have the following theorem. 

\begin{theorem}[\cite{GKW17,WZ17,GKVW19}]
Assuming quantum hardness of learning with errors, there exists a quantum-secure lockable obfuscation scheme. 
\end{theorem}
}

%\subsection{Witness Indistinguishability}
%\label{ssec:prelims:wi}

%\newcommand{\cL}{\mathcal{L}}
\newcommand{\protextcl}{\mathsf{c}\mathsf{QEXT}}
\newcommand{\protextq}{\mathsf{qQEXT}}

\section{Secure Quantum Extraction Protocols}
\label{sec:qext}
\noindent We define the notion of quantum extraction protocols below. An extraction protocol, associated with an NP relation, is a {\em classical} interactive protocol between a sender and a receiver.The sender has an NP instance and a witness; the receiver only has the NP instance. 
\par In terms of properties,  we require the property that there is a QPT extractor that can extract the witness from a semi-malicious sender (i.e., follows the protocol but is allowed to choose its own randomness) even if the sender is a QPT algorithm. Moreover, the semi-malicious sender should not be able to distinguish whether it's interacting with the extractor or the honest receiver. 
\par In addition, we require the following property (zero-knowledge): the interaction of any malicious receiver with the sender should be simulatable without the knowledge of the witness. The malicious receiver can either be classical or quantum and thus, we have two notions of quantum extraction protocols corresponding to both of these cases. 
\par In terms of properties required, this notion closely resembles the concept of zero-knowledge argument of knowledge (ZKAoK) systems. There are two important differences: 
\begin{itemize}
    \item Firstly, we do not impose any completeness requirement on our extraction protocol. 
    \item In ZKAoK systems, the prover can behave maliciously (i.e., deviates from the protocol) and the argument of knowledge property states that the probability with which the extractor can extract is negligibly close to the probability with which the prover can convince the verifier. In our definition, there is no guarantee of extraction if the sender behaves maliciously. 
\end{itemize}

%\noindent \fullversion{discuss the semi-malicious setting; in particular have a discussion about why semi-malicious and what means in the quantum world -- mention how the algorithm can choose the randomness as a function of the messages.}\\
%\noindent \prab{change perfect to semi-malicious.}

\begin{definition}[Quantum extraction protocols secure against quantum adversaries]
\label{def:qqext}
A \textbf{quantum extraction protocol secure against quantum adversaries}, denoted by $\protextq$ is a classical protocol between two classical PPT algorithms, sender $\sender$ and a receiver $\receiver$ and is associated with an NP relation $\rel$. The input to both the parties is an instance $\inst \in \lang(\rel)$. In addition, the sender also gets as input the witness $\witness$ such that $(\inst,\witness) \in \rel$. At the end of the protocol, the receiver gets the output $\witness'$. The following properties are satisfied by $\protextq$: 
\begin{itemize}
   \item {\bf Quantum Zero-Knowledge}: Let $p:\mathbb{N}\rightarrow\mathbb{N}$ be any polynomially bounded function. For every $(\inst,\witness) \in \rel$, for any QPT algorithm $\receiver^*$ with private quantum register of size $|\mathrm{R}_{\receiver^*}|=p(\secparam)$, for any large enough security parameter $\secparam \in \mathbb{N}$, there exists a QPT simulator $\Simu$ such that, 
    $$\view_{\receiver^*}\left( \langle \sender(1^\secparam,\inst,\witness),\receiver^*(1^\secparam,\inst, \cdot)\rangle\right) \approx_{Q} \Simu(1^\secparam,\receiver^*,\inst,\cdot).$$

    \item {\bf Semi-Malicious Extractability}: Let $p:\mathbb{N}\rightarrow \mathbb{N}$ be any polynomially bounded function. For any large enough security parameter $\secparam \in \mathbb{N}$, for every $(\inst,\witness)\in \lang(\rel)$, for every semi-malicious\footnote{A QPT algorithm is said to be semi-malicious in the quantum extraction protocol if it follows the protocol but is allowed to choose the randomness for the protocol.} QPT $\sender^*$ with private quantum register of size $|\mathrm{R}_{\sender^*}|=p(\secparam)$, there exists a QPT extractor $\ext=(\ext_1,\ext_2)$ (possibly using the code of $\sender^*$ in a non-black box manner), the following holds: 
    \begin{itemize}
      %  \item The output of the extractor is the concatenation of outputs of $\ext_1$ and $\ext_2$. 
        \item {\bf Indistinguishability of Extraction}: $\view_{\sender^*}\left( \langle \sender^*(1^\secparam,\inst,\witness, \cdot), \receiver(1^\secparam,\inst) \rangle\right) \approx_{Q} \ext_1 \left(1^\secparam,\sender^*,\inst, \cdot \right)$

        \item The probability that $\ext_2$ outputs $\witness'$ such that $(\inst,\witness') \in \rel$ is negligibly close to 1.
      
    \end{itemize}
\end{itemize}
\end{definition}

\begin{definition}[Quantum extraction protocols secure against classical adversaries]
\label{def:qext:cl}
A \textbf{quantum extraction protocol secure against classical adversaries} $\protextcl$ is defined the same way as in Definition~\ref{def:qqext} except that instead of quantum zero-knowledge, $\protextcl$ satisfies classical zero-knowledge property defined below:
\begin{comment}
a classical protocol between two parties, sender $\sender$ and a receiver $\receiver$ and is associated with an NP relation $\rel$. 

The input to both the parties is an instance $\inst \in \lang(\rel)$. In addition, the sender also gets as input the witness $\witness$ such that $(\inst,\witness) \in \rel$. At the end of the protocol, the receiver gets the output $\witness'$. The following properties are satisfied by $\protextcl$:
\end{comment}
\begin{itemize}
   \item {\bf Classical Zero-Knowledge}: Let $p:\mathbb{N}\rightarrow\mathbb{N}$ be any polynomially bounded function. For any large enough security parameter $\secparam \in \mathbb{N}$, for every $(\inst,\witness) \in \rel$, for any classical PPT algorithm $\receiver^*$ with auxiliary information $\aux \in \{0,1\}^{\poly(\secparam)}$, there exists a classical PPT simulator $\Simu$  such that 
    $$\view_{\receiver^*}\left( \langle \sender(1^\secparam,\inst,\witness),\receiver^*(1^\secparam,\inst, \aux)\rangle\right) \approx_{c} \Simu(1^\secparam,\receiver^*,\inst,\aux).$$

    %\prab{need to define what is classical indistinguishability}

%    \item {\bf Perfect Extractability}: Let $p:\mathbb{N}\rightarrow \mathbb{N}$ be any polynomially bounded function. For every $(\inst,\witness)\in \lang(\rel)$, for every QPT $\sender^*$ with private quantum register of size $|\mathrm{R}_{\sender^*}|=p(\secparam)$, there exists a QPT extractor $\ext=(\ext_1,\ext_2)$ (possibly using the code of $\sender^*$ in a non black box manner), such that for any large enough security parameter $\secparam \in \mathbb{N}$, the following holds: 
    %\begin{itemize}
      %  \item The output of the extractor is the concatenation of outputs of $\ext_1$ and $\ext_2$. 
        %\item$\view_{\sender^*}\left( \langle \sender^*(1^\secparam,\inst,\witness, \cdot), \receiver(1^\secparam,\inst) \rangle\right) \approx_{Q} \ext_1 \left(1^\secparam,\sender^*,\inst, \cdot \right)$
        
       % \item If $\sender^*$ is semi-malicious\footnote{A QPT algorithm is said to be semi-malicious in the quantum extraction protocol if it follows the protocol but is free to choose its own randomness.} then the probability that $\ext_2$ outputs $\witness'$ such that $(\inst,\witness') \in \rel$ is 1.
      
  %  \end{itemize}
\end{itemize}
\end{definition}

\paragraph{Quantum-Lasting Security.} A desirable property of cQEXT protocols is that a classical malicious receiver, long after the protocol has been executed cannot use a quantum computer to learn the witness of the sender from the transcript of the protocol along with its own private state. We call this property  {\em quantum-lasting security}; first introduced by Unruh~\cite{Unruh13}. We formally define quantum-lasting security below. 

\begin{definition}[Quantum-Lasting Security]
A cQEXT protocol is said to be {\bf quantum-lasting secure} if the following holds: for any large enough security parameter $\secparam \in \mathbb{N}$, for any classical PPT $\receiver^*$, for any QPT adversary $\adversary^*$, for any auxiliary information $\aux \in \{0,1\}^{\poly(\secparam)}$, for any auxiliary state of polynomially many qubits,  $\rho$, there exist a QPT simulator $\Simu^*$ such that:  
$$\adversary^*\left(\view_{\receiver^*} \left\langle \sender(1^\secparam,\inst,\witness),\receiver^*(1^\secparam,\inst, \aux) \right\rangle, \rho \right) \approx_{Q} \Simu^*(1^\secparam, \inst, \aux, \rho)$$
 \end{definition}

\section{QEXT Secure Against Classical Receivers}
In this section, we show how to construct quantum extraction protocols secure against classical adversaries based solely on the quantum hardness of learning with errors. 

\paragraph{Tools.}

%%%%%%%%%%%%
\fullversion{
\begin{itemize}
    \item   Quantum-secure computationally-hiding and perfectly-binding non-interactive commitments, $\comm$ (see Section~\ref{sec:prelims:commit}). 
    \par We instantiate the underlying commitment scheme in~\cite{PW09} using $\comm$ to obtain a quantum-secure extractable commitment scheme. Instead of presenting a definition of quantum-secure extractable commitment scheme and then instantiating it, we directly incorporate the construction of~\cite{PW09} in the construction of the extraction protocol. 
    
  \item Noisy trapdoor claw-free functions $\{f_{\key,b}:\cX \rightarrow D_{\cY} \}_{\key \in \cK, b \in \{0,1\}}$ (see Section~\ref{sec:ntcfs}). 
  
  %  \item Quantum trapdoor finder scheme $\qtf=(\qtf.\gen,\qtf.\tdfinder,\qtf.\chck)$ (see Section~\ref{sec:qtfs}).
    \item Quantum-secure secure function evaluation protocol  $\protpc=(\protpc.\sender,\protpc.\receiver)$ (see Section~\ref{ssec:prelims:sfe}). 
\end{itemize}
}

%%%%% Submission version
\submversion{
\begin{itemize}
    \item   Quantum-secure computationally-hiding and perfectly-binding non-interactive commitments, $\comm$. 
    \par We instantiate the underlying commitment scheme in~\cite{PW09} using $\comm$ to obtain a quantum-secure extractable commitment scheme. Instead of presenting a definition of quantum-secure extractable commitment scheme and then instantiating it, we directly incorporate the construction of~\cite{PW09} in the construction of the extraction protocol. 
    
  \item Noisy trapdoor claw-free functions $\{f_{\key,b}:\cX \rightarrow D_{\cY} \}_{\key \in \cK, b \in \{0,1\}}$. 
  
  %  \item Quantum trapdoor finder scheme $\qtf=(\qtf.\gen,\qtf.\tdfinder,\qtf.\chck)$ (see Section~\ref{sec:qtfs}).
    \item Quantum-secure secure function evaluation protocol  $\protpc=(\protpc.\sender,\protpc.\receiver)$. 
\end{itemize}
}

\paragraph{Construction.} We present the construction of the quantum extraction protocol $(\sender,\receiver)$ in Figure~\ref{fig:QEXTcons:classical} for an NP language $\lang$. %Our extraction mechanism of committing to the shares of the values and later releasing only one of the shares is borrowed from~\cite{PW09}. 

\renewcommand{\bit}{\mathbf{b}}

\begin{figure}[!htb]
\begin{center}
\begin{tabular}{|p{12.8cm}|}
\hline 
\begin{center}
    $\sfefunc$
\end{center}
Input of sender: $\left(  \left\{\bfc_{i,0}^{(j)},\bfc_{i,1}^{(j)},(sh_{i,w_i}^{(j)})',(\bfd_{i,w_i}^{(j)})', \td_i, \key_i,y_i, v_i, w_i^{(j)} \right\}_{i,j \in [k]},\witness \right)$\\
Input of receiver: $\left(  \left\{sh_{i,\overline{w_i}}^{(j)},\bfd_{i,\overline{w_i}}^{(j)} \right\}_{i,j \in [k]} \right)$
\ \\
\begin{itemize}
    \setlength\itemsep{1em}
    \item If for any $i,j \in [k]$, $\bfc_{i,w_i}^{(j)} \neq \comm\left( 1^{\secparam},(sh_{i,w_i}^{(j)})';(\bfd_{i,w_i}^{(j)})' \right)$ or $\bfc_{i,\overline{w_i}}^{(j)} \neq \comm \left(1^{\secparam},sh_{i,\overline{w_i}}^{(j)};\bfd_{i,\overline{w_i}}^{(j)} \right)$, output $\bot$.  
    
    \item For every $i \in [k]$, let $(x_{i,0},x_{i,1}) \leftarrow \invert(\key_i,\td_i,y_i)$. 
    \begin{itemize}
        \setlength\itemsep{1em}
        \item {\em Check if the commitments commit to the same message}:  Output $\bot$ if the following does not hold: for every $j,j' \in [k]$,
        we have $\left(sh_{i,w_i}^{(j)} \right)'\oplus sh_{i,w_i}^{(j)} = \left(sh_{i,w_i}^{(j')} \right)'\oplus sh_{i,w_i}^{(j')}$. 
        
        \item If $v_i=0$: let $(b_i,J(x'_{i,b_i}))=(sh_{i,w_i}^{(j)})' \oplus sh_{i,\overline{w_i}}^{(j)}$, where $J(\cdot)$ is the injection in the definition of NTCF. Since $J(\cdot)$ can be efficiently inverted, recover $x'_{i,b_i}$. If $x'_{i,b_i} \neq x_{i,b_i}$, output $\bot$.   
        \item If $v_i=1$: let $(u_i,d_i)=\left(sh_{i,w_i}^{(j)} \right)' \oplus sh_{i,\overline{w_i}}^{(j)}$. If $\langle d_i, J(x_{i,0}) \oplus J(x_{i,1}) \rangle \neq u_i$, or if $d_i \notin G_{\key_i,0,x_{i,0}}\cap G_{\key_i,1,x_{i,1}}$ output $\bot$.
    \end{itemize}
    
    \item Otherwise, output $\witness$. 
\end{itemize}
\ \\
\hline
\end{tabular}
\end{center}
\caption{Description of the function $\sfefunc$ associated with the $\protpc$.}
\label{fig:sfefunc:classical}
\end{figure}

\begin{comment}
\begin{figure}[!htb]
\begin{tabular}{|p{15cm}|}
\hline 
\begin{center}
    $\sfefunc$
\end{center}
Input of sender: $\left(\left\{\bfc_{i,0},\bfc_{i,1},x_{i,0},x_{i,1},sh'_{i,b_i},\bfd'_{i,b_i} \right\}_{i \in [k]},\witness \right)$\\
Input of receiver: $\left(  \left\{sh_{i,\overline{b_i}},\bfd_{i,\overline{b_i}} \right\}_{i \in [k]} \right)$
\ \\
\begin{itemize}
    \item If for any $i \in [k]$, $\bfc_{i,b_i} \neq \comm\left( 1^{\secparam},sh'_{i,b_i};\bfd'_{i,b_i} \right)$ or $\bfc_{i,\overline{b_i}} \neq \comm \left(1^{\secparam},sh_{i,\overline{b_i}};\bfd_{i,\overline{b_i}} \right)$, output $\bot$.  
    
    \item If for any $i \in [k]$, $\langle sh_{i,0} \oplus sh_{i,1}, x_{i,0} \oplus x_{i,1} \rangle \neq 0$, output $\bot$. 
    
    \item Otherwise, output $\witness$. 
    
\end{itemize}
\ \\
\hline
\end{tabular}
\caption{Description of the function $f$ associated with the $\protpc$.}
\label{fig:sfefunc:classical}
\end{figure}
\end{comment}

\begin{figure}[!htb]
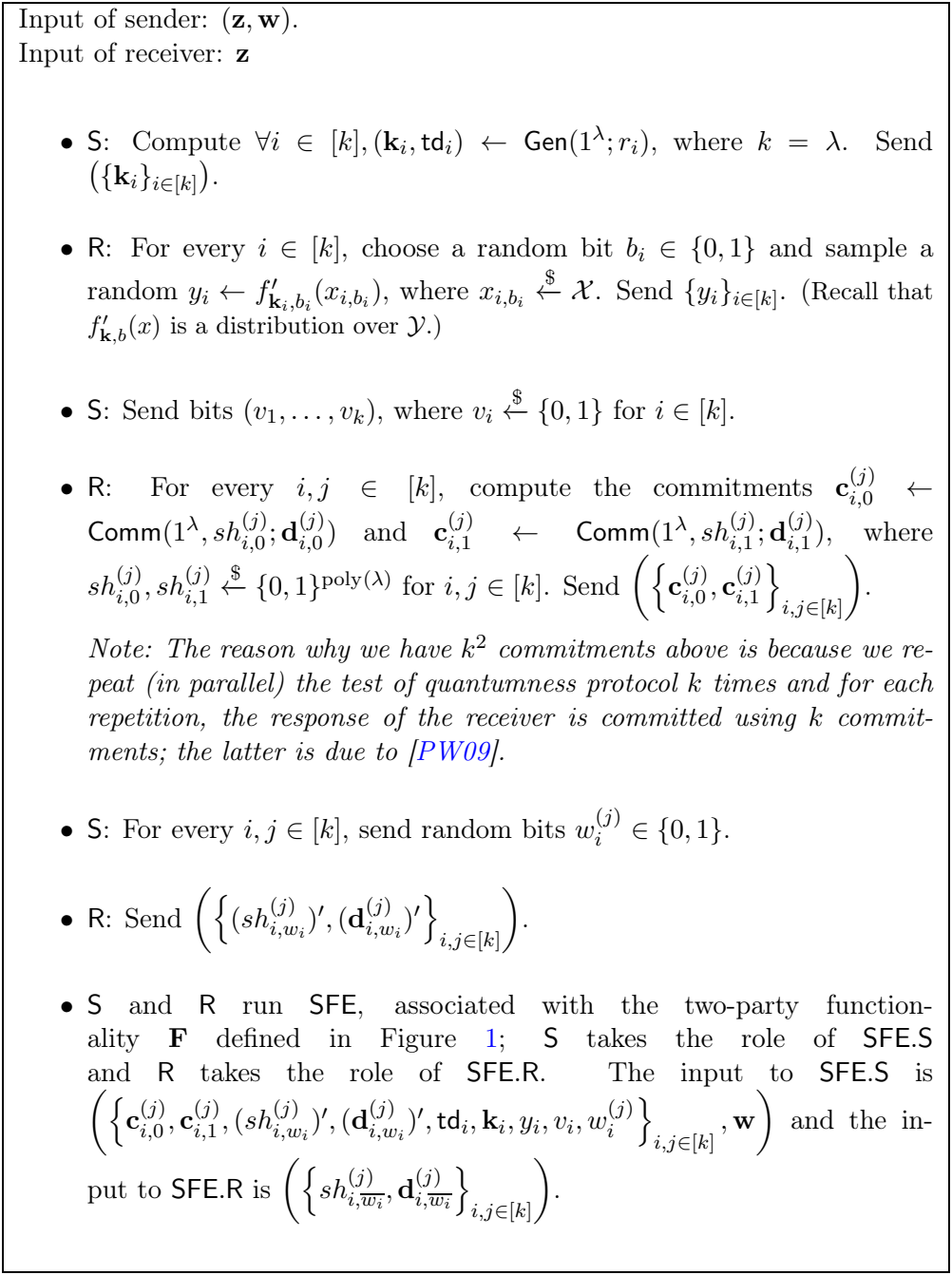

\begin{center}
\begin{tabular}{|p{12.8cm}|}
\hline 

Input of sender: $(\inst,\witness)$. \\
Input of receiver: $\inst$ \\

\begin{itemize}
\setlength\itemsep{1em}

    \item  $\sender$: Compute $\forall i\in[k],(\key_i,\td_i) \leftarrow \gen(1^{\secparam};r_i)$, where $k=\secparam$. Send $\left( \{\key_i\}_{i \in [k]} \right)$. 
    
    \item $\receiver$: For every $i \in [k]$, choose a random bit $b_i \in \{0,1\}$ and sample a random $y_i \leftarrow f'_{\key_i,b_i}(x_{i,b_i})$, where $x_{i,b_i} \xleftarrow{\$} \cX$. Send $\{y_i\}_{i \in [k]}$.
    {(\small Recall that $f'_{\key,b}(x)$ is a distribution over $\cY$.)}
    
    \item $\sender$: Send bits $(v_1,\ldots,v_{k})$, where $v_i \xleftarrow{\$} \{0,1\}$ for $i \in [k]$. 
    
    \item $\receiver$: For every $i,j \in [k]$, compute the commitments $\bfc_{i,0}^{(j)} \leftarrow \comm(1^{\secparam},sh_{i,0}^{(j)};\bfd_{i,0}^{(j)})$ and $\bfc_{i,1}^{(j)} \leftarrow \comm(1^{\secparam},sh_{i,1}^{(j)};\bfd_{i,1}^{(j)})$, where $sh_{i,0}^{(j)},sh_{i,1}^{(j)} \xleftarrow{\$} \{0,1\}^{\poly(\secparam)}$ for $i,j \in [k]$. Send $\left(\left\{ \bfc_{i,0}^{(j)},\bfc_{i,1}^{(j)} \right\}_{i,j \in [k]} \right)$.
    
    {\em Note: The reason why we have $k^2$ commitments above is because we repeat (in parallel) the test of quantumness protocol $k$ times and for each repetition, the response of the receiver is committed using $k$ commitments; the latter is due to~\cite{PW09}.}
    
   % $\sender$: Compute $(\qtf.\pk, \qtf.\sk) \leftarrow \qtf.\gen(1^{\secparam};r)$, where $k=\secparam$. Send $\qtf.\pk$. 

  %  \item $\receiver$:
    %\begin{itemize}
    
       % \item Compute $\bfc^* \leftarrow \comm(1^{\secparam},0;\bfd^*)$. 
        %\item Compute $\bfc_{0} \leftarrow \comm(1^{\secparam},sh_{0};\bfd_{0})$ and $\bfc_{1} \leftarrow \comm(1^{\secparam},sh_{1};\bfd_{1})$, where $sh_{0},sh_{1} \xleftarrow{\$} \{0,1\}^{\poly(\secparam)}$.
      
   % \end{itemize}
    %  Send $\left(\bfc_{0},\bfc_{1} \right)$. 
    
   % \item $\sender$: Send a bit $b \in \{0,1\}$ uniformly at random.
    %\item $\receiver$: Send $\left(\left\{ sh'_{b},\bfd' \right\}\right)$. 
    
     \item $\sender$: For every $i,j \in [k]$, send random bits $w_i^{(j)} \in \{0,1\}$.
    
    \item $\receiver$: Send $\left(\left\{ (sh_{i,w_i}^{(j)})',(\bfd_{i,w_i}^{(j)})' \right\}_{i,j \in [k]} \right)$.  
    
  %  \item If $\bfc_{i,b_i} \neq \comm(1^{\secparam},sh_{i,b_i};\bfd_{i,b_i})$ for any $i \in [k]$ then $\sender$ aborts.
    
    \item $\sender$ and $\receiver$ run $\protpc$, associated with the two-party functionality $\sfefunc$ defined in Figure~\ref{fig:sfefunc:classical}; $\sender$ takes the role of $\protpc.\sender$ and $\receiver$ takes the role of $\protpc.\receiver$.  The input to $\protpc.\sender$ is $\left(  \left\{\bfc_{i,0}^{(j)},\bfc_{i,1}^{(j)},(sh_{i,w_i}^{(j)})',(\bfd_{i,w_i}^{(j)})', \td_i, \key_i,y_i, v_i, w_i^{(j)} \right\}_{i,j \in [k]},\witness \right)$ and the input to $\protpc.\receiver$ is $\left(  \left\{sh_{i,\overline{w_i}}^{(j)},\bfd_{i,\overline{w_i}}^{(j)} \right\}_{i,j \in [k]} \right)$.
    
  %  \item If $\bfc_{i,b_i} \neq \comm(1^{\secparam},sh_{i,b_i};\bfd_{i,b_i})$ for any $i \in [k]$ then $\sender$ aborts.
    
  %  \item $\sender$ and $\receiver$ run $\protpc$, associated with the two-party functionality $\sfefunc$ defined in Figure~\ref{fig:sfefunc:classical}; $\sender$ takes the role of $\protpc.\sender$ and $\receiver$ takes the role of $\protpc.\receiver$.  The input to $\protpc.\sender$ is $\left(  \left\{\bfc_{0},\bfc_{1},\qtf.\sk,sh'_{b},\bfd'_{b} \right\},\witness \right)$ and the input to $\protpc.\receiver$ is $\left(  \left\{sh_{\overline{b}},\bfd_{\overline{b}} \right\} \right)$.
 
\end{itemize}
\ \\
\hline
\end{tabular}
\end{center}
\caption{Quantum Extraction Protocol $(\sender,\receiver)$ secure against classical receivers.}
\label{fig:QEXTcons:classical}
\end{figure}

%%%%%%%%%%%%%%%%%%%%%%%%%%%%%%%%%%%%%%%%%%%%%%%%%%%%%%%%%%%%%%%%%%%%%%%%%%%%%%%%%%%%%%%%%%%%%%%%%%%%%%%%%%%%%%%%%%%%%%%%%%%%%%%%%%%%%%%%%%%%%%%%%%%%%%%%%%%%%%%%%%

\submversion{

\noindent We prove the following lemma in the full version. 

\begin{lemma}
\label{lem:proof:clqext}
Assuming the quantum security of $\comm,\protpc$ and NTCFs, the protocol $(\sender,\receiver)$ is a quantum extraction protocol secure against classical adversaries for NP. Moreover, $(\sender,\receiver)$ satisfies quantum-lasting security. 
\end{lemma}

}

%%%%%%%%%%%%%%%%%%%%%%%%%%%%%%%%%%%%%%%%%%%%%%%%%%%%%%%%%%%%%%%%%%%%%%%%%%%%%%%%%%%%%%%%%%%%%%%%%%%%%%%%%%%%%%%%%%%%%%%%%%%%%%%%%%%%%%%%%%%%%%%%%%%%%%%%%%%%%%%%%%

%%%%%%%%%%%%%%%%%%%%%%%%%%%%%%%%%%%%%%%%%%%%%%%%%%%%%%%%%%%%%%%%%%%%%%%%%%%%%%%%%%%%%%%%%%%%%%%%%%%%%%%%%%%%%%%%%%%%%%%%%%%%%%%%%%%%%%%%%%%%%%%%%%%%%%%%%%%%%%%%%%
\fullversion{
\begin{lemma}
Assuming the quantum security of $\comm,\protpc$ and NTCFs, the protocol $(\sender,\receiver)$ is a quantum extraction protocol secure against classical adversaries for NP, and it is also quantum-lasting secure.
\end{lemma}
\begin{proof}
\ 
\paragraph{Classical Zero-Knowledge.} Let $\receiver^*$ be a classical PPT algorithm. We first describe a classical simulator $\Simu$ such that $\receiver^*$ cannot distinguish whether it's interacting with $\sender$ or with $\Simu$. 

\paragraph{Description of $\Simu$.}
\begin{itemize}
 \item Until the $\protpc$ protocol is executed, it behaves as the honest sender would. That is, 
 \begin{itemize}
\item  For every $ i\in[k]$, it computes $(\key_i,\td_i) \leftarrow \gen(1^{\secparam};r_i)$. Send $\left( \{\key_i\}_{i \in [k]} \right)$.  
    
    \item It receives $\{y_i\}_{i \in [k]}$ from $\receiver^*$.
    
    \item It sends bits $(v_1,\ldots,v_{k})$, where $v_i \xleftarrow{\$} \{0,1\}$ for $i \in [k]$. 
    
    \item It receives $\left(\left\{ \bfc_{i,0}^{(j)},\bfc_{i,1}^{(j)} \right\}_{i,j \in [k]} \right)$ from $\receiver^*$.
    
   % $\sender$: Compute $(\qtf.\pk, \qtf.\sk) \leftarrow \qtf.\gen(1^{\secparam};r)$, where $k=\secparam$. Send $\qtf.\pk$. 

  %  \item $\receiver$:
    %\begin{itemize}
    
       % \item Compute $\bfc^* \leftarrow \comm(1^{\secparam},0;\bfd^*)$. 
        %\item Compute $\bfc_{0} \leftarrow \comm(1^{\secparam},sh_{0};\bfd_{0})$ and $\bfc_{1} \leftarrow \comm(1^{\secparam},sh_{1};\bfd_{1})$, where $sh_{0},sh_{1} \xleftarrow{\$} \{0,1\}^{\poly(\secparam)}$.
      
   % \end{itemize}
    %  Send $\left(\bfc_{0},\bfc_{1} \right)$. 
    
   % \item $\sender$: Send a bit $b \in \{0,1\}$ uniformly at random.
    %\item $\receiver$: Send $\left(\left\{ sh'_{b},\bfd' \right\}\right)$. 
    
     \item For every $i,j \in [k]$, it sends random bits $w_i^{(j)}\in \{0,1\}$.
    
    \item It receives $\left(\left\{ (sh_{i,w_i}^{(j)})',(\bfd_{i,w_i}^{(j)})' \right\}_{i,j \in [k]} \right)$ from $\receiver^*$. 
    
  %  \item If $\bfc_{i,b_i} \neq \comm(1^{\secparam},sh_{i,b_i};\bfd_{i,b_i})$ for any $i \in [k]$ then $\sender$ aborts.
    \end{itemize}
    \item It then executes $\protpc$ with $\receiver^*$, associated with the two-party functionality $\sfefunc$ defined in Figure~\ref{fig:sfefunc:classical}; the input of $\Simu$ in $\protpc$ is $\bot$. 
    
\end{itemize}

\noindent We prove the following by a sequence of hybrids. For some arbitrary auxiliary information $\aux \in \{0,1\}^{\poly(\secparam)}$,
$$\view_{\receiver^*}\left( \langle \sender(1^\secparam,\inst,\witness),\receiver^*(1^\secparam,\inst, \aux)\rangle\right) \approx_{Q} \Simu(1^\secparam,\receiver^*,\inst,\aux),$$
In other words, that no QPT distinguisher can distinguish between the view of $\receiver^*$ when interacting with $\sender$ from the output of $\Simu$.  This is stronger than what we need to argue classical ZK, as it would be enough to show that $\receiver^*$, a PPT machine (not QPT), cannot distinguish.  However, the stronger indistinguishability result makes it easier to show that the scheme is quantum-lasting secure. \\

\noindent $\underline{\hybrid_1}$: The output of this hybrid is $\view_{\receiver^*}\left( \langle \sender(1^\secparam,\inst,\witness),\receiver^*(1^\secparam,\inst, \aux)\rangle\right)$. \\

\noindent $\underline{\hybrid_2}$: Consider the following sender, $\hybrid_2.\sender$, that behaves as follows:
\begin{enumerate}
    \item $\receiver^*$: Sends $\{y_i\}_{i \in [k]}$.
    \item $\hybrid_2.\sender$: Sends $(v_1,\ldots,v_k)$ uniformly at random. If $\receiver^*$ aborts in this step, $\hybrid_2.\sender$ aborts.
    \item $\receiver^*$: Sends $\left\{ \left(\bfc_{i,0}^{(j)}, \bfc_{i,1}^{(j)}\right)\right\}_{i,j \in [k]}$. If $\receiver^*$ aborts in this step, $\hybrid_2.\sender$ aborts.
    \item $\hybrid_2.\sender$: Sends $w_{i}^{(j)} \in \{0,1\}$ uniformly at random for all $i,j \in [k]$.
    \item $\receiver^*$: Opens up the commitments queried, $\left\{\left(sh_{i,w_i}^{(j)}, \bfd_{i,w_i}^{(j)} \right)\right\}_{i,j \in [k]}$. If $\receiver^*$ aborts in this step, $\hybrid_2.\sender$ aborts. If $\bfc_{i,w_i}^{(j)} \neq \comm(1^{\secparam},sh_{i,w_i}^{(j)};\bfd_{i,w_i}^{(j)})$ for any $i,j \in [k]$, continue the execution of the protocol as in Step 11.   
    \item $\hybrid_2.\sender$: Keep rewinding ($\poly(k)$ times) to Step 4, until it is able to recover another commitment accepting transcript. A commitment accepting transcript is one for which all the commitments opened in Step 5 are valid, i.e. that $\bfc_{i,w_i}^{(j)}= \comm(1^\secparam, sh_{i,w_i}^{(j)}; \bfd_{i,w_i}^{(j)})$. Let $\{(w_{i}^{(j)})'\}$ be the queries sent in the second recovered commitment accepting transcript. If for any $i \in [k]$, it is the case that for every $j \in [k]$, it holds that $(w_i^{(j)})'=w_{i}^{(j)}$, then abort.
    \item If $\hybrid_2.\sender$ did not abort in the previous step, then for every $i \in [k]$, there is $j_i \in [k]$, s.t. $(w_{i}^{(j_i)})' \neq w_i^{(j_i)}$. From these two transcripts, it extracts the committed value.
    \item $\hybrid_2.\sender$: (We call this step the NTCF condition check). From the commited values recovered, check if they satisfy the desired NTCF conditions. I.e. for every $i \in [k]$, if $v_i=0$, check if the decommited value if a valid preimage $(b_i,J(x_{i,b_i}))$, and if $v_i=1$ check if the decommited value is a valid correlation $(u_i, d_i)$. If the check do not pass, continue as before. If the check pass,
    \begin{itemize}
        \item Keep rewinding ($\poly(k)$ times) until Step 2, repeating the proccess above, including the rewinding phase for the commitment challenges. The rewinding continues until we get another transcript, for which the NTCF check passes. Let $(v_1',\ldots,v_k')$ be the messages sent at Step 2 in the new transcript. 
    \end{itemize}
    \item $\hybrid_2.\sender$: If $(v_1,\ldots,v_{k})$ and $(v_1',\ldots,v_{k}')$ are different in less than $\omega(\log(k))$ coordinates, then abort. 
    \item If $\hybrid_2.\sender$ has not aborted so far, let $S$ be the set of indices at which both $(v_1,\ldots,v_{k})$ and $(v_1',\ldots,v_{k}')$ differ. For $i \in S$, let $(b_i,x_i)$ and $(d_i,u_i)$ be the values recovered from the commitment accepting transcripts associated with bits $v_i$ and $v'_i$. Denote  $T=\{(b_i,x_i,d_i,u_i):i \in S\}$. Moreover, $|T|=\omega(\log(k))$
    \item Now, continue the execution of the protocol on the original thread; i.e., when the $\hybrid_2.\sender$  queries $(w_1,\ldots,w_k)$ and $(v_1,\ldots,v_{k})$. 
\end{enumerate}

\par The only difference between $\hybrid_1$ and $\hybrid_2$ is that $\hybrid_2.\sender$ aborts on some transcripts; conditioned on $\hybrid_2.\sender$ not aborting, the transcript produced by the receiver when interacting with $\sender$ is identical to the transcript produced by $\hybrid_2.\sender$. We claim that the probability that $\hybrid_2.\sender$ aborts, conditioned on the event that $\receiver^*$ does not abort, is negligibly small.

\begin{claim}
\label{clm:hyb2:abort}
$\Pr[\hybrid_2.\sender \text{ aborts}|\receiver^*\text{ does not abort}]=\negl(k)$
\end{claim}
\begin{proof}
To argue this, we first establish some terminology. Let $p_1$ be the probability with which $\receiver^*$ produces a commitment accepting transcript and $p_2$ be the probability with which $\receiver^*$ passes the NTCF condition check. We call the rewinding performed in Step 4 to be "inner rewinding" and the the rewinding performed in Step 8 to be "outer rewinding". 
\par In the rest of the proof, we condition on the event that $\receiver^*$ does not abort. Consider the following claims.

\begin{claim}
The probability that the number of outer rewinding operations performed is greater than $k$ is negligible. 
\end{claim}
\begin{proof}
Note that the outer rewinding is performed till the point it can recover a transcript that passes the NTCF check. Since the probability that $\receiver^*$ produces a transcript that passes the NTCF check is $p_2$, we have that the expected number of outer rewinding operations to be $(1-p_2) + p_2 \cdot \frac{1}{p_2} \leq 2$. By Chernoff, the probability that the number of outer rewinding operations is greater than $k$ is negligible.    
\end{proof}

\begin{claim}
\label{clm:rewind:inner}
The probability that the number of inner rewinding operations performed is greater than $k^2$ is negligible. 
\end{claim}
\begin{proof}
Note that for every NTCF transcript, $\comm$ is rewound many times until $\hybrid_2.\sender$ can indeed recover another commitment-accepting transcript. For a given NTCF transcript, since the probability that $\receiver^*$ produces a commitment accepting transcript is $p_1$, we have that the expected number of inner rewinding operations to be $(1-p_1) + p_1 \cdot \frac{1}{p_1} \leq 2$. And thus by Chernoff, for a given NTCF transcript, the probability that the number of inner rewinding operations is greater than $k$ is negligible. Since the number NTCF transcripts produced is at most $k$ with probability negligibly close to 1, we have that the total number of inner rewinding operations is at most $k^2$ with probability neglibly close to 1.    
\end{proof}

\noindent We now argue about the probability that $\hybrid_2.\sender$ aborts on an NTCF transcript (Step 9) and the probability that it aborts on the transcript of $\comm$ (Step 6). 

\begin{claim}
\label{clm:step9:abort}
The probability that $\hybrid_2.\sender$ aborts in Step 9 is negligible. 
\end{claim}
\begin{proof}
Note that $\hybrid_2.\sender$ aborts in Step 9 only if: (i) it received a valid transcript on the original thread of execution, (ii) it rewinds until the point it receives another valid NTCF transcript and, (iii) the challenge $(v'_1,\ldots,v'_k)$ on which the second transcript was accepted differs from $(v_1,\ldots,v_k)$ only in $\omega(\log(k))$ co-ordinates. Thus, the probability that it aborts is the following quantity: 
\begin{eqnarray*}
& & p_2(p_2+p_2(1-p_2)+p_2(1-p_2)^2+\cdots)\cdot \mathsf{Pr}[\substack{(v_1,\ldots,v_k)\text{ and }(v'_1,\ldots,v'_k)\\ \text{ differ in less than }\omega(\log(k))\text{ co-ordinates}}]\\
& \leq & p_2^2 \left( \frac{1}{p_2} \right) \cdot \mathsf{Pr}[\substack{(v_1,\ldots,v_k)\text{ and }(v'_1,\ldots,v'_k)\\ \text{ differ in less than }\omega(\log(k))\text{ co-ordinates}}] \\
& = & p_2 \cdot \negl(k)\ \ (\text{By Chernoff Bound})
\end{eqnarray*}
\end{proof}

\begin{claim}
The probability that $\hybrid_2.\sender$ aborts in Step 6 is negligible. 
\end{claim}
\begin{proof}
Since step 6 is executed for multiple NTCF transcripts, we need to argue that for any of NTCF transcripts, the probability that $\hybrid_2.\sender$ aborts in Step 6 is negligible. Since we already argued in Claim~\ref{clm:rewind:inner} that the number of inner rewinding operations is $\poly(k)$, by union bound, it suffices to argue the probability that for any given NTCF transcript, the probability that $\hybrid_2.\sender$ aborts in Step 6 is negligible. This is similar to the argument in Claim~\ref{clm:step9:abort}: the probability that $\hybrid_2.\sender$ aborts in Step 6 is $p_1^2 \cdot \frac{1}{p_1} \cdot \prob \left[ \exists i \in [k], \forall j \in [k]: \left( w_{i}^{(j)} \right)' = \left( w_i^{(j)} \right)  \right] = p_1 \cdot 2^{-k}$.   
\end{proof}

\noindent Observe that $\hybrid_2.\sender$ only aborts in Steps 6 and 9; recall that we have already conditioned on the even that $\receiver^*$ does not abort. Thus, we have the proof of the claim.  

\end{proof}
\par This claim shows that $\hybrid_1$ and $\hybrid_2$ are indistinguishable:
$$\view_{\receiver^*}\left( \langle \sender(1^\secparam,\inst,\witness),\receiver^*(1^\secparam,\inst, \aux)\rangle\right) \approx_{Q} \view_{\receiver^*}\left( \langle \hybrid_2.\sender(1^\secparam,\inst,\witness),\receiver^*(1^\secparam,\inst, \aux)\rangle\right).$$\\

\noindent $\underline{\hybrid_3}$: In this hybrid, $\hybrid_3.\sender$ will do as $\hybrid_2.\sender$ except as follows: once it gets to step 8, if the NTCF check passes, it continues as usual, but if the NTCF check does not pass, it inputs $\bot$ in the $\protpc$.\\

\noindent The indistinguishability of $\hybrid_2$ and $\hybrid_3$ follows from the security of the $\protpc$ against malicious quantum receivers, and we have:
$$\view_{\receiver^*}\left( \langle \hybrid_2.\sender(1^\secparam,\inst,\witness),\receiver^*(1^\secparam,\inst, \aux)\rangle\right) \approx_{Q} \view_{\receiver^*}\left( \langle \hybrid_3.\sender(1^\secparam,\inst,\witness),\receiver^*(1^\secparam,\inst, \aux)\rangle\right),$$
This is because the following holds in the event that the above check does not pass:

\small $$\sfefunc\left(\left(  \left\{\bfc_{i,0}^{(j)},\bfc_{i,1}^{(j)},(sh_{i,w_i}^{(j)})',(\bfd_{i,w_i}^{(j)})', \td_i, \key_i,y_i, v_i, w_i^{(j)} \right\}_{i,j \in [k]},\witness \right),\ \left(  \left\{sh_{i,\overline{w_i}}^{(j)},\bfd_{i,\overline{w_i}}^{(j)} \right\}_{i,j \in [k]} \right) \right)$$
$$= \sfefunc\left( \left(\bot\right), \left(  \left\{sh_{i,\overline{w_i}}^{(j)},\bfd_{i,\overline{w_i}}^{(j)} \right\}_{i,j \in [k]} \right) \right).$$\\

\normalsize
\noindent $\underline{\hybrid_4}$: In this hybrid, $\hybrid_4.\sender$ always inputs $\bot$ in the $\protpc$.\\

We have the following: 
$$\view_{\receiver^*}\left( \langle \hybrid_3.\sender(1^\secparam,\inst,\witness),\receiver^*(1^\secparam,\inst, \aux)\rangle\right) \approx_{Q} \view_{\receiver^*}\left( \langle \hybrid_4.\sender(1^\secparam,\inst,\witness),\receiver^*(1^\secparam,\inst, \aux)\rangle\right)$$

\noindent This is because either $\hybrid_3.\sender$ inputs $\bot$ into the $\protpc$ or it can find $T=\{(b_i,x_i,u_i,d_i):i \in S\}$ (see $\hybrid_2$) such that both $(b_i,x_i)$ and $(u_i,d_i)$ pass the NTCF checks corresponding to the $i^{th}$ instantiation. Moreover, recall that $|T|=\omega(\log(k))$. This contradicts the security of NTCFs: by the adaptive hardcore bit property of the NTCF, a PPT classical adversary can break a given instantiation with probability negligibly close to 1/2 and thus, it can break $\omega(\log(k))$ instantiations only with negligible probability.\\

\noindent $\underline{\hybrid_5}$: Now the hybrid sender, $\hybrid_5.\sender$ does as $\hybrid_4.\sender$, but it does not rewind $\receiver^*$.\\

\noindent The statistical distance between $\hybrid_4$ and $\hybrid_5$ is negligible in $k$; this follows from Claim~\ref{clm:hyb2:abort}.

%\todo[inline]{There seem to be some formatting issues right around here?}

\paragraph{Quantum-Lasting Security.} We have shown that for any auxiliary information $\aux \in \{0,1\}^{\poly(\secparam)}$,
$$\view_{\receiver^*}\left( \langle \sender(1^\secparam,\inst,\witness),\receiver^*(1^\secparam,\inst, \aux)\rangle\right) \approx_{Q} \Simu(1^\secparam,\receiver^*,\inst,\aux).$$

Let $\cA^*$ be any QPT adversary that is given the transcript, $\view_{\receiver^*}\left( \langle \sender(1^\secparam,\inst,\witness),\receiver^*(1^\secparam,\inst, \aux)\rangle\right)$. Consider the $\Simu^*$ that first runs $\Simu(1^\secparam,\receiver^*,\inst,\aux)$, and then runs $\cA^*$, i.e. $\Simu^*$ is the QPT that on a polynomial sized quantum states $\rho$ acts as
$$\Simu^*\left(1^\secparam, \cA^*, \receiver^*, \inst, \aux, \rho \right) = \cA^*\left(\Simu(1^\secparam,\receiver^*,\inst,\aux),\rho \right).$$
Since $\cA^*$ is QPT, it can't distinguish if it is given the actual transcript or the output of $\Simu$. In particular, we have that

$$\cA^*\left(\view_{\receiver^*}\left( \langle \sender(1^\secparam,\inst,\witness),\receiver^*(1^\secparam,\inst, \aux)\rangle\right),\rho\right) \approx_Q \Simu^*\left(1^\secparam, \cA^*, \receiver^*, \inst, \aux, \rho \right).$$

\paragraph{Extractability.} Let $\sender^*$ be the semi-malicious sender. We define our quantum extractor $\ext$ as follows.

\paragraph{Description of $\ext$.} The input to $\ext$ is the instance $\inst$. 

\begin{itemize}
    \item Run $\sender^*$ to obtain $\{\key_i\}_{i \in [k]}$.
    \item For all $i \in [k]$,
    \begin{itemize}
    \item Prepare the superpostion $$\frac{1}{\sqrt{2|\cX|}} \underset{b, x\in \cX, y \in \cY}{\sum} \sqrt{f'_{\key_i,b}(x)(y)} \ket{b,x,y}$$
    which can be done efficiently by the required properties of NTCF. 
    \item Measure the $y$ register, to obtain outcome $y_i$. Denote the postmeasurement quantum state by $\ket{\Psi_i}$. By NTCF,
    $$\ket{\Psi_i} = \frac{\ket{0,x_{i,0}}+\ket{1,x_{i,1}}}{\sqrt{2}} $$ where $(x_{i,0},x_{i,1})\leftarrow \inv(\key_i,\td_i,y_i)$.
    \item Compute $J$ into a new register, $\ket{b,x,0} \rightarrow \ket{b,x,J(x)}$, and then uncompute the register containing $x$ by performing $J^{-1}$, i.e. $\ket{b,x,J(x)}\rightarrow \ket{b,x\oplus J^{-1}(J(x)),J(x)}$. The resulting transformation is $\ket{b,x,0} \rightarrow \ket{b,0,J(x)}$.
    \item Discard the second register, and keep the first register containing $b$ and the third register with $J(x)$. At this point, the extractor has the states $$\ket{\Psi'_i}=\frac{\ket{0,J(x_{i,0})}+\ket{1,J(x_{i,1})}}{\sqrt{2}} $$
    \end{itemize}
    \item Send $\{y_i\}_{i \in [k]}$ to $\sender^*$, and let $\{v_i\}_{i \in [k]}$ be the message received from $\sender*$.
    \item  For all $i \in [k]$:
    \begin{itemize}
        \item if $v_i = 0$, measure $\ket{\Psi'_i}$ in the standard basis, to obtain $(b_i,J(x_{i,b_i}))$.
        \item if $v_i = 1$, apply the Hadamard transformation to $\ket{\Psi'_i}$, and measure in standard basis to obtain $(u_i, d_i)$
    \end{itemize}
    \item For all $i,j \in [k]$, choose the shares $(sh_{i,0}^{(j)},sh_{i,1}^{(j)})$ uniformly at random conditioned on either $(b_i,J(x_{i,b_i}))= sh_{i,0}^{(j)}\oplus sh_{i,1}^{(j)}$ or $(u_i,d_i)=sh_{i,0}^{(j)}\oplus sh_{i,1}^{(j)}$ if $v_i =0$ or $v_i=1$ respectively.
    \item Perform the rest of the protocol as the honest receiver would. Output the outcome of the SFE protocol.
    \end{itemize}
    
\begin{claim}
Assuming NTCFs, perfect correctness and security of $\protpc$,  the probability that $\ext$ extracts from the semi-malicious sender ie negligibly close to $1$.
\end{claim}
\begin{proof}
We first claim that with probability negligibly close to 1, the following is satisfied for every $v_i \in [k]$:
\begin{itemize}
        \item If $v_i = 0$, let $(b_i,J(x_{i,b_i}))$ be the value obtained by measuring $\ket{\Psi'_i}$ in the standard basis. Then,  $f'_{\key_i,b_i}(x_{i,b})=y_i$, 
        \item If $v_i = 1$, let $(u_i, d_i)$ be the value obtained by applying the Hadamard transformation to $\ket{\Psi'_i}$, and measuring it in the standard basis. Then $\langle d_i,J(x_{i,0}) \oplus J(x_{i,1}) \rangle = u_i$ and $d_i \notin G_{\key_i,0,x_{i,0}}\cap G_{\key_i,1,x_{i,1}}$.
    \end{itemize}
This follows from the union bound and Lemma 5.1 of the protocol of~\cite{brakerski2018cryptographic}. By perfect correctness of $\protpc$, it follows that if the extractor inputs shares $sh_{i,0}^{(j)},sh_{i,1}^{(j)}$ that answer correctly each challenge, the output it will receive from the $\protpc$ will be the witness $\witness$. 

%The NTCF properties hold for any keys $\key \in \cK$.  In particular, this means that the extractor can correctly prepare the initial superposition (by efficient range superposition property). After measuring the $y$ register, the postmeasurement outcome as described in the extractor follows from the trapdoor and the injective pair property, where we have $x_{i,b} \leftarrow \inv(\td_i,b,y_i)$ for each $b \in \{0,1\}$.

%At this point, if $a_i=0$, the extractor's standard basis measurement would yield a valid preimage.
\end{proof}

\begin{claim}
\label{clm:ind:ext}
$\view_{\sender^*}\left( \langle \sender^*(1^\secparam,\inst,\witness, \cdot), \receiver(1^\secparam,\inst) \rangle\right) \approx_{Q} \ext_1 \left(1^\secparam,\sender^*,\inst, \cdot \right)$
\end{claim}
\begin{proof}
Consider the following hybrids. \\

\noindent $\underline{\hybrid_1}$: The output of this hybrid is $\view_{\sender^*}\left( \langle \sender^*(1^\secparam,\inst,\witness, \cdot), \receiver(1^\secparam,\inst) \rangle\right)$.  \\

\noindent $\underline{\hybrid_2}$: We define a hybrid receiver $\hybrid_2.\receiver$ who sets the input to $\protpc$ to be $\bot$. 
\par The following holds from the semantic security of $\protpc$ against QPT senders: 
$$\view_{\sender^*}\left( \langle \sender^*(1^\secparam,\inst,\witness, \cdot), \receiver(1^\secparam,\inst) \rangle\right) \approx_{Q} \view_{\sender^*}\left( \langle \sender^*(1^\secparam,\inst,\witness, \cdot), \hybrid_2.\receiver(1^\secparam,\inst) \rangle\right) $$

\noindent $\underline{\hybrid_3}$: We define a hybrid receiver $\hybrid_3.\receiver$ that behaves as $\hybrid_2.\receiver$, but it samples $\{y_i\}_{i \in [k]}$ as the extractor would, by preparing the claw-free superpositions, and then measuring the $y$ register. We claim that the distribution over $y_i$'s is the same in $\hybrid_2$ and $\hybrid_3$. To see this, note that $\hybrid_3$ samples from the distribution $y_i$ from the distribution: $\frac{1}{2|\cX|}\underset{b\in \{0,1\},x \in X}{\sum}f'_{\key_i,b}(x)(y)$. To sample from this distribution, we can first sample $b \in \{0,1\}$, then an $x_{i,b} \in \cX$ and then sampling $y_i$ from the distribution $f'_{\key_i,b}(x_{i,b})$. \\

\noindent $\underline{\hybrid_4}$:  We define a hybrid receiver $\hybrid_4.\receiver$ who computes $\{y_i\}_{i\in [k]}$ by performing the quantum operations that the extractor does, and then computes, for all $i \in[k]$, either $(b_i,J(x_{i,b_i}))$ or $(u_i,d_i)$ according to whether $v_i = 0$ or $v_i=1$ respectively. In other words, $\hybrid_4.\receiver$ compute correct answers to the test of quantumness, then it commits to appropriate shares,
\[
  sh_{i,0}^{(j)}\oplus sh_{i,1}^{(j)} =
  \begin{cases}
    (b_i,J(x_{i,b})) & \text{if $v_i=0$} \\
    (u_i,d_i) & \text{if $v_i=1$}
  \end{cases}
\]
$\hybrid_4.\receiver$ uses these shares for commitment $\bfc_{i,0}^{(j)} = \comm(1^\secparam, sh_{i,0}^{(j)}; \bfd_{i,0}^{(j)})$ and $\bfc_{i,1}^{(j)} = \comm(1^\secparam, sh_{i,1}^{(j)}; \bfd_{i,1}^{(j)})$ The rest of the steps are the same as $\hybrid_3.\receiver$. 
\par The following holds from the computational hiding property of $\comm$ by a similar argument to the one in~\cite{PW09}: 
$$\view_{\sender^*}\left( \langle \sender^*(1^\secparam,\inst,\witness, \cdot), \hybrid_3.\receiver(1^\secparam,\inst) \rangle\right) \approx_{Q} \view_{\sender^*}\left( \langle \sender^*(1^\secparam,\inst,\witness, \cdot), \hybrid_4.\receiver(1^\secparam,\inst) \rangle\right) $$

\noindent $\underline{\hybrid_5}$: We define a hybrid receiver $\hybrid_5.\receiver$ who 
sets the input in $\protpc$ to be   $\left(\left\{sh_{i,\overline{w_i}}^{(j)},\bfd_{i,\overline{w_i}}^{(j)} \right\}_{i \in [k]} \right) $, where $\{w_i\}_{i \in [k]}$ are the bit queried by $\sender^*$ when asking the receiver to reveal commitments. Note that the output distribution of $\hybrid_5.\receiver$ is identical to that of the extractor $\ext$. \\

\par The following holds from the semantic security of $\protpc$ against quantum senders: 
$$\view_{\sender^*}\left( \langle \sender^*(1^\secparam,\inst,\witness, \cdot), \hybrid_4.\receiver(1^\secparam,\inst) \rangle\right) \approx_{Q} \view_{\sender^*}\left( \langle \sender^*(1^\secparam,\inst,\witness, \cdot), \hybrid_5.\receiver(1^\secparam,\inst) \rangle\right) \equiv \ext_1 \left(1^\secparam,\sender^*,\inst, \cdot \right)$$

\end{proof}

\end{proof}

\paragraph{Indistinguishability of Extraction Against Malicious Senders.} We observe that our construction satifies a stronger property than claimed. Our protocol satisfies indistinguishability of extraction against {\em malicious} senders, and not just semi-malicious senders. However, the extractability is still required against semi-malicious senders. 
\par We formalize this in the claim below. 
%stronger extractability property: the semi-malicious sender cannot distinguish whether its interacting with the extractor or the honest receiver even if it is allowed to abort. If at any point in time, the sender aborts, so does the extractor and note that from the same arguments as above, the view of the sender when interacting with the honest receiver will still be indistinguishable (against a quantum polynomial time adversary) from the view of the sender when interacting with the extractor. We formalize this in the claim below. 

\begin{claim}
\label{clm:aborting}
The quantum extraction protocol $(S,R)$ described in Figure~\ref{fig:QEXTcons:classical} satisfies indistinguishability of extraction (Definition~\ref{def:qqext}) against malicious senders. 
\end{claim}

\noindent We omit the proof of the above claim since it is identical to the proof of Claim~\ref{clm:ind:ext}. The indistinguishability of the hybrids in the proof of Claim~\ref{clm:ind:ext} already hold against malicious senders; in the proof, we never used the fact that the sender was semi-malicious. 
\par The only caveat missing in the proof of Claim~\ref{clm:ind:ext} but comes up in the proof of the above claim is the fact that the malicious sender could abort. If the malicious sender aborts, then so does the extractor; since the extractor is straightline, the view of the sender until that point will still be indistinguishable from the view of the sender when interacting with the honest receiver. 
}

%%%%%%%%%%%%%%%%%%%%%%%%%%%%%%%%%%%%%%%%%%%%%%%%%%%%%%%%%%%%%%%%%%%%%%%%%%%%%%%%%%%%%%%%%%%%%%%%%%%%%%%%%%%%%%%%%%%%%%%%%%%%%%%%%%%%%%%%%%%%%%%%%%%%%%%%%%%%%%%%%%

%%% APPLICATION TO ZK 
%%% APPLICATION TO ZK 
%%% APPLICATION TO ZK 
%%% APPLICATION TO ZK 
%%% APPLICATION TO ZK 

%%%%%%%%%%%%%%%%%%%%%%%%%%%%%%%%%%%%%%%%%%%%%%
%%%%%%%%%%%%%%%%%%%%%%%%%%%%%%%%%%%%%%%%%%%%%%
%%%%%%%%%%%%%%%%%%%%%%%%%%%%%%%%%%%%%%%%%%%%%%
%%%%%%%%%%%%%%%%%%%%%%%%%%%%%%%%%%%%%%%%%%%%%%

\submversion{
\input{submver-applzk}
}

%%%%%%%%%%%%%%%%%%%%%%%%%%%%%%%%%%%%%%%%%%%%%%
%%%%%%%%%%%%%%%%%%%%%%%%%%%%%%%%%%%%%%%%%%%%%%
%%%%%%%%%%%%%%%%%%%%%%%%%%%%%%%%%%%%%%%%%%%%%%
%%%%%%%%%%%%%%%%%%%%%%%%%%%%%%%%%%%%%%%%%%%%%%
\fullversion{ 
\subsection{Application: QZK with classical soundness}
\noindent In this section, we show how to construct a quantum zero-knowledge, classical prover, argument system for NP secure against quantum verifiers; that is, the protocol is classical, the malicious prover is also a classical adversary but the malicious verifier can be a polynomial time quantum algorithm. To formally define this notion, consider the following definition.  

\begin{definition}[Classical arguments for NP]
\label{def:clargnp}
A classical interactive protocol $(\prvr,\vrfr)$ is a \textbf{classical ZK argument system} for an NP language $\lang$, associated with an NP relation $\lang(\rel)$, if the following holds:
\begin{itemize}

    \item {\bf Completeness}: For any $(\inst,\witness) \in \lang(\rel)$, we have that $\Pr[\langle \prvr(1^{\secparam},\inst,\witness),\vrfr(1^{\secparam},\inst) \rangle=1] \geq 1-\negl(\secparam),$
    for some negligible function $\negl$. 
    \item {\bf Soundness}: For any $\inst \notin \lang$, any PPT classical adversary $\prvr^*$, and any polynomial-sized auxiliary information $\aux$, we have that $\Pr[\langle \prvr^*(1^{\secparam},\inst,\aux),\vrfr(1^{\secparam},\inst) \rangle=1] \leq \negl(\secparam)$, for some negligible function $\negl$. 
    \end{itemize}
\end{definition}

\newcommand{\Sim}{\mathsf{S}}

\noindent We say that a classical argument system for NP is a QZK (quantum zero-knowledge) classical argument system for NP if in addition to the above properties, a classical interactive protocol satisfies zero-knowledge against malicious receivers. 
%The following definition extends the notion of quantum zero-knowledge to classical arguments in a similar way. We want the honest verifier to be a classical machine, but an adversarial verifier could try to learn from the honest prover by using a quantum computer. 

\begin{definition}[QZK classical argument system for NP] A classical interactive protocol $(\prvr,\vrfr)$ is a \textbf{quantum zero-knowledge classical argument system} for a language $\lang$, associated with an NP relation $\lang(\rel)$ if both of the following hold.

\begin{itemize}
    \item $(\prvr,\vrfr)$ is a classical argument for $\lang$ (Definition~\ref{def:clargnp}).
    \item {\bf Quantum Zero-Knowledge}: Let $p:\mathbb{N}\rightarrow \mathbb{N}$ be any polynomially bounded function. For any QPT $\vrfr^*$ that on instance $\inst \in \lang$ has private register of size $|\mathsf{R}_{\vrfr^*}|= p(|\inst|)$, there exist a QPT $\Simu$ such that the following two collections of quantum channels are quantum computationally indistinguishable,
    \begin{itemize}
        \item     $\left\{ \Simu(\inst,\vrfr^*, \cdot )\right\}_{\inst \in \lang}$
        \item $\left\{\view_{\vrfr^*}( \langle \prvr(\inst,\aux_1),\vrfr^*(\inst,\cdot)\rangle)\right\}_{\inst \in \lang}$.
    \end{itemize}
    
    In other words, that for every $\inst \in \lang$, for any bounded polynomial $q:\mathbb{N}\rightarrow \mathbb{N}$, for any QPT distinguisher $\cD$ that outputs a single bit, and any $p(|\inst|)+q(|\inst|)$-qubits quantum state $\rho$,
\end{itemize}
$$ \big|\Pr\left[\cD\left(\Simu(\inst,\vrfr^*, \cdot )\otimes I)(\rho)\right)=1\right]\ \ \ \ \ \ \ \ \ \ \ \ \ \ \ \ \ \ \ \ \ \ \ \ \ \ \ \ \ \ \ \ \ \ \ \ \ $$ 
$$\ \ \ \ \ \ \ \ \ \ \ \ -\Pr\left[\cD\left((\view_{\vrfr^*}( \langle \prvr(\inst,\aux_1),\vrfr^*(\inst,\cdot)\rangle)\otimes I)(\rho)\right)=1\right]\big|\leq \epsilon(|\inst|) $$
\end{definition}

\paragraph{Witness-Indistinguishability against quantum verifiers.}
We also consider witness indistinguishable (WI) argument systems for NP languages secure against quantum verifiers. We define this formally below. 

\begin{definition}[Quantum WI for an $\lang \in \text{NP}$] 
\label{def:qwi}
A classical protocol $(\prvr,\vrfr)$ is a \textbf{quantum witness indistinguishable argument system} for an NP language $\lang$ if both of the following hold.
\begin{itemize}
    \item $(\prvr,\vrfr)$ is a classical argument for $\lang$ (Definition~\ref{def:clargnp}).
    \item {\bf Quantum WI}: Let $p:\mathbb{N}\rightarrow \mathbb{N}$ be any polynomially bounded function. For every $\inst \in \lang$, for any two valid witnesses $\witness_1$ and $\witness_2$, for any QPT $\vrfr^*$ that on instance $\inst$ has private quantum register of size $|\mathsf{R}_{\vrfr^*}|=p(|\inst|)$, we require that
    $$\view_{\vrfr^*}( \langle \prvr(\inst,\witness_1),\vrfr^*(\inst,\cdot)\rangle) \approx_{Q} \view_{\vrfr^*}( \langle \prvr(\inst,\witness_2),\vrfr^*(\inst,\cdot)\rangle).$$
\end{itemize}
If $(\prvr, \vrfr)$ is a quantum proof system (sound against unbounded provers), we say that $(\prvr,\vrfr)$ is a \textbf{quantum witness indistinguishable proof system} for $\lang$.
\end{definition}

\paragraph{Instantiation.} By suitably instantiating the constant round WI argument system of Blum~\cite{Blu86} with perfectly binding quantum computational hiding commitments, we achieve a constant round quantum WI classical argument system assuming quantum hardness of learning with errors. 

%%%%%%%%%%%%%%%%%%%%%%%%%%%%%%%%%%%%%%%%%%%%%%%%%%%%%%%%%%%%%%%%%%%%%%%%%%%%%%%%%%%%%%%%%%%%%%%%%%%%%%%%
%%%%%%%%%%%%% construction
\subsubsection{Construction}
\noindent We present a construction of constant round quantum zero-knowledge classical argument system for NP. 

\paragraph{Tools.}

\begin{itemize}
    \item Perfectly-binding and quantum-computational hiding non-interactive commitments $\comm$ (see Section~\ref{sec:prelims:commit}). 
    \item Quantum extraction protocol secure against classical adversaries $\protextcl=(\sender,\receiver)$ associated with the relation $\extrel$ below. More generally, $\protextcl$ could be any quantum extraction protocol secure against classical adversaries satisfying Claim~\ref{clm:aborting} (indistinguishability of extraction against malicious senders). 
    $$\extrel = \left\{ \left(\bfc ,\ (\bfd,\td) \right)\ :\ \bfc = \comm(1^{\secparam},\td;\bfd) \right\} $$
    \item Quantum witness indistinguishable classical argument of knowledge system $\protwi=(\protwi.\prvr,\allowbreak \protwi.\vrfr)$ for the relation $\rel_{\wi}$ (Definition~\ref{def:qwi}). 
    
\end{itemize}

\begin{figure}[!htb]
\begin{tabular}{|p{16cm}|}
\hline \\
Instance: $\left(\inst,\td,\left\{(\bfc_0^{(j)})^*,(\bfc_1^{(j)})^*\right\}_{j \in [k]}\right)$ \\
Witness: $\left(\witness,\left\{(sh_0^{(j)},\bfd_0^{(j)},sh_1^{(j)},\bfd_1^{(j)})\right\}_{j \in[k]} \right)$ \\
NP verification: Accept if one of the following two conditions are satisfied: 
\begin{itemize}
    \item $(\inst,\witness) \in \rel$. 
    \item If for every $j \in [k]$, it holds that \small $$\left((\bfc_0^{(j)})^* = \comm(1^{\secparam},sh_0^{(j)};\bfd_0^{(j)})\right) \bigwedge \left((\bfc_1^{(j)})
    ^* = \comm(1
    ^{\secparam},sh_1^{(j)};\bfd_1^{(j)})\right) \bigwedge \left( \td= sh_0^{(j)} \oplus sh_1^{(j)} \right).$$ \normalsize
\end{itemize}

\ \\
\hline
 \end{tabular}
\caption{Relation $\rel_{\wi}$ associated with $\protwi$.} 
\label{fig:relwi}
\end{figure}

\paragraph{Construction.} Let $\lang$ be an NP language. We describe a classical interactive protocol $\left(\prvr,\vrfr \right)$ for $\lang$ in Figure~\ref{fig:qzk}. 

\begin{figure}[!htb]
\begin{tabular}{|p{15cm}|}
\hline 
\begin{itemize}
    \item {\bf Trapdoor Committment Phase}: $\vrfr$: sample $\td  \leftarrow \{0,1\}^{\secparam}$. Compute $\bfc \leftarrow \comm(1^{\secparam},\td;\bfd)$, where $\bfd \leftarrow \{0,1\}^{\mathrm{poly}(\secparam)}$ is the randomness used in the commitment. Send $\bfc$ to $\prvr$. 
    \item {\bf Trapdoor Extraction Phase}: $\prvr$ and $\vrfr$ run the quantum extraction protocol $\protextcl$ with $\vrfr$ taking the role of the sender $\protextcl.\sender$ and $\prvr$ taking the role of the receiver $\protextcl.\receiver$. The input of $\protextcl.\sender$ is  $(1^{\secparam},\bfc,(\bfd,\td);\rqext)$ and the input of $\protextcl.\receiver$ is $\left( 1^{\secparam},\bfc \right)$, where $\rqext$ is the randomness used by the sender in $\protextcl$. Let the transcript generated during the execution of $\protextcl$ be $\trans_{\vrfr \rightarrow \prvr}$.
    
    \textit{Note: The trapdoor extraction phase will be used by the simulator, while proving zero-knowledge, to extract the trapdoor from the malicious verifier.}
    
  %  \item {\bf Extracting from the Prover}: $\prvr$ and $\vrfr$ run $\protext_2$ with $\vrfr$ taking the role of the receiver $\protext_2.\receiver$ and $\prvr$ taking the role of the sender $\protext_2.\sender$. The input of $\protext_2.\sender$ is $(1^{\secparam},\bfc,\bot)$ and the input of $\protext_2.\receiver$ is $(1^{\secparam},\bfc)$. Let $\rqext$ be the randomness used by $\protext_2.\sender$ in $\protext_2$. Let the transcript generated during the execution of $\protext_2$ be $\trans_{\prvr \rightarrow \vrfr}$. 
    % \textit{This phase will be used while arguing the soundness of the protocol.}
    
    \item Let $k=\secparam$. For every $j \in [k]$, $\prvr$ sends $(\bfc_0^{(j)})^*=\comm(1^{\secparam},sh_0^{(j)};\bfd_0^{(j)})$ and $(\bfc_1^{(j)})^*=\comm(1^{\secparam},sh_1^{(j)};\bfd_1^{(j)})$, where $sh_0^{(j)},sh_1^{(j)} \xleftarrow{\$} \{0,1\}^{\poly(\secparam)}$.   
    
    \item For every $j \in [k]$, $\vrfr$ sends bit $b^{(j)} \xleftarrow{\$} \{0,1\}$ to $\prvr$. 
    
    \item $\prvr$ sends $(sh_{b^{(j)}}^{(j)},\bfd_{b^{(j)}}^{(j)})$ to $\vrfr$. 
    
    \item $\vrfr$ sends $\rqext,\bfd,\td$ to $\prvr$. Then $\prvr$ checks the following: 
    \begin{itemize}
    \item Let $\trans_{\vrfr \rightarrow \prvr}$ be $(m_1^S,m_1^R,\ldots,m_{\rnd'}^S,m_{\rnd'}^R)$, where the message $m_i^R$ (resp., $m_i^S$) is the message sent by the receiver (resp., sender) in the $i^{th}$ round\footnote{We remind the reader that in every round, only one party speaks.} and $\rnd'$ is the number of rounds of $\protextcl$. Let the message produced by $\protextcl.\sender\left( 1^{\secparam},\bfc,(\bfd,\td);\rqext \right)$ in the $i^{th}$ round be $\widetilde{m}_{i}^S$. 
    \item If for any $i \in [\rnd']$, $\widetilde{m}_i^S \neq {m}_i^S$ then $\prvr$ aborts If $\bfc \neq \comm(1^{\secparam},\td;\bfd)$ then abort. 
    \end{itemize}
    
    \item {\bf Execute Quantum WI}: $\prvr$ and $\vrfr$ run $\protwi$ with $\prvr$ taking the role of $\protwi$ prover $\protwi.\prvr$ and $\vrfr$ taking the role of $\protwi$ verifier $\protwi.\vrfr$. The input to $\protwi.\prvr$ is the security parameter $1^{\secparam}$, instance $\left(\inst,\td,\left\{(\bfc_0^{(j)})^*,(\bfc_1^{(j)})^*\right\}_{j \in [k]}\right)$ and witness $(\witness,\bot)$. The input to $\protwi.\vrfr$ is the security parameter $1^{\secparam}$ and instance $\left(\inst,\td,\left\{(\bfc_0^{(j)})^*,(\bfc_1^{(j)})^*\right\}_{j \in [k]}\right)$.  %\prab{need to change the relation associated with quantum WI.}
    
   % \item {\bf Release last message of $\protwi$ if $\vrfr$ behaved correctly}: $\prvr$ computes the last round message of $\protwi.\prvr\left(1^{\secparam},\left( \inst,\bfc,\trans_{\prvr \rightarrow \vrfr} \right)\right)$. Denote this by $\wi_{\rnd}$. $\prvr$ {\em does not} send $\wi_{\rnd}$ to $\vrfr$. Instead, $\prvr$ and $\vrfr$ run $\protpc$, associated with the two-party functionality $G$ (defined in Figure~\ref{fig:funcG}), with $\prvr$ taking the role of $\protpc$ sender $\protpc.\sender$ and $\vrfr$ taking the role of $\protpc$ receiver $\protpc.\receiver$. The input to $\protpc.\sender$ is the transcript $\trans$, commitment $\bfc$ and $\wi_{\rnd}$. The input to $\protpc.\receiver$ is the decommitment $\bfd$ and randomness $\rqext$. 
 
\item {\bf Decision step}: $\vrfr$ computes the decision step of $\protwi.\vrfr$. 
 
\end{itemize}
\ \\
\hline
\end{tabular}
\caption{(Classical Prover) Quantum Zero-Knowledge Argument Systems for NP.}
\label{fig:qzk}
\end{figure}

\begin{lemma}
\label{lem:qzk:classical}
The classical interactive protocol $(\prvr,\vrfr)$ is a quantum zero-knowledge, classical prover, argument system for NP.
\end{lemma}
\begin{proof}
\ 
The completeness is straightforward. We prove soundness and zero-knowledge next. 
\ 
\paragraph{Soundness.} Let $\prvr^*$ be a classical PPT algorithm. We prove that $\prvr^*(1^{\secparam},\inst,\aux)$, for $\inst \notin \lang$ and auxiliary information $\aux$, can convince $\vrfr(1^{\secparam},\inst)$ with only negligible probability. Consider the following hybrids. \\

\noindent \underline{$\hybrid_1$}: The output of this hybrid is the view of the prover $\view_{\prvr^*}( \langle \prvr^*(1^{\secparam},\inst,\aux),\vrfr(1^{\secparam},\inst)\rangle)$ along with the decision bit of $\vrfr$. \\

\noindent \underline{$\hybrid_2$}: We consider the following hybrid verifier $\hybrid_2.\vrfr$ which executes the trapdoor commitment phase and the trapdoor extraction phase with $\prvr^*$ honestly. It then receives $\{((\bfc_0^{(j)})^*,(\bfc_1^{(j)})^*))\}_{j \in [k]}$ from the prover. $\hybrid_2.\vrfr$ sends random bits $\{b^{(j)}\}_{j \in [k]}$ to $\prvr^*$ and it then receives $(sh_{b^{(j)}}^{(j)},\bfd_{b^{(j)}}^{(j)})$. At this point, $\hybrid_2.\vrfr$ will rewind until it can extract $\td^*$ from the commitments; if it extracted multiple values or it didn't extract any value, set $\td^*=\bot$. This is done similarly to the cQEXT case and the argument from~\cite{PW09}.\\

\par The output distribution of this hybrid is identical to the output distribution of $\hybrid_1$. 

\clearpage

\noindent The following holds: 
\begin{eqnarray*}
\prob \left[ 1 \leftarrow \langle P^*(1^{\secparam},\inst,\aux),\hybrid_2.\vrfr(1^{\secparam},\inst)\rangle  \right] & = &  \prob \left[ \substack{1 \leftarrow \langle P^*(1^{\secparam},\inst,\aux),\hybrid_2.\vrfr(1^{\secparam},\inst) \rangle \\ \bigwedge\\ \left( \td^* = \td \bigvee \td^* \neq \td \right)}\ :\ \td^* \leftarrow \ext(1^{\secparam},\qaux) \right] \\
& & \\
& & \\
& \leq & \underbrace{\prob \left[ \substack{1 \leftarrow \langle P^*(1^{\secparam},\inst,\aux),\hybrid_2.\vrfr(1^{\secparam},\inst) \rangle \\ \bigwedge\\ \left( \td^* = \td \right)}\ :\ \td^* \leftarrow \ext(1^{\secparam},\qaux) \right]}_{\varepsilon_1}\\
& & + \underbrace{\prob \left[ \substack{1 \leftarrow \langle P^*(1^{\secparam},\inst,\aux),\hybrid_2.\vrfr(1^{\secparam},\inst) \rangle \\ \bigwedge\\ \left( \td^* \neq \td \right)}\ :\ \td^* \leftarrow \ext(1^{\secparam},\qaux) \right]}_{\varepsilon_2}
\end{eqnarray*}

%\prab{need to check why is there a 12 in the probability above..}

\noindent We prove the following claims. 

\begin{claim}
$\varepsilon_1 \leq \negl(\secparam)$, for some negligible function $\negl$. 
\end{claim}
\begin{proof}
Consider the following hybrids.\\

\noindent \underline{$\hybrid_3$}: We define a hybrid verifier $\hybrid_3.\vrfr$ that performs the trapdoor commitment phase honestly. In the trapdoor extraction phase, it executes $\protext_1.\simulator(1^{\secparam})$, instead of $\protext_1.\sender(1^{\secparam},\bfc,(\bfd,\td))$, while interacting with $\prvr^*$. The rest of the steps of $\hybrid_3.\vrfr$ is as defined in $\hybrid_2.\vrfr$. 
\par Let $\td^*$ be the trapdoor extracted as before. From the zero-knowledge property of $\protextcl$, the following holds:
\begin{eqnarray}
\varepsilon_1 \leq \prob \left[ \substack{1 \leftarrow \langle P^*(1^{\secparam},\inst,\aux),\hybrid_3.\vrfr(1^{\secparam},\inst) \rangle \\ \bigwedge\\ \left( \td^* = \td \right)}\ :\ \td^* \leftarrow \ext(1^{\secparam},\qaux) \right]   + \negl(\secparam) 
\end{eqnarray}

\noindent \underline{$\hybrid_4$}: We define the hybrid verifier $\hybrid_4.\vrfr$ that performs the same steps as $\hybrid_3.\vrfr$ execpt that it computes $\bfc$ as $\comm(1^{\secparam},\mathbf{0};\bfd)$ instead of $\comm(1^{\secparam},\td;\bfd)$, where $\mathbf{0}$ is a $\secparam$-length string of all zeroes. 
\par Let $\td^*$ be the trapdoor extracted as before. From the quantum hiding property of $\comm$, the following holds:
\begin{eqnarray}
\prob \left[ \substack{1 \leftarrow \langle P^*(1^{\secparam},\inst,\aux),\hybrid_3.\vrfr(1^{\secparam},\inst) \rangle \\ \bigwedge\\ \left( \td^* = \td \right)}\ :\ \td^* \leftarrow \ext(1^{\secparam},\qaux) \right] \ \ \ \ \ \ \ \ \ \ \ \ \ \ \ \ \ \ \ \ \ \ \ \ \ \ \ \ \\ \leq \prob \left[ \substack{1 \leftarrow \langle P^*(1^{\secparam},\inst,\aux),\hybrid_4.\vrfr(1^{\secparam},\inst) \rangle \\ \bigwedge\\ \left( \td^* = \td \right)}\ :\ \td^* \leftarrow \ext(1^{\secparam},\qaux) \right]   + \negl(\secparam) 
\end{eqnarray}

\noindent \underline{$\hybrid_5$}: We define the hybrid verifier $\hybrid_5.\vrfr$ that performs the same steps as $\hybrid_4.\vrfr$ except that it samples $\td$ {\em after} it completes its interaction with the $\prvr^*$. 
\par Note that the output distributions of $\hybrid_4$ and $\hybrid_5$ are identical. Moreover, the probability that $\hybrid_5.\vrfr$ accepts and $\td^* = \td$ is at most $\frac{1}{2^{\secparam}}$. Thus we have, 
\begin{eqnarray*}
& & \prob \left[ \substack{1 \leftarrow \langle P^*(1^{\secparam},\inst,\aux),\hybrid_4.\vrfr(1^{\secparam},\inst) \rangle \\ \bigwedge\\ \left( \td^* = \td \right)}\ :\ \td^* \leftarrow \ext(1^{\secparam},\qaux) \right] \\
& = & \prob \left[ \substack{1 \leftarrow \langle P^*(1^{\secparam},\inst,\aux),\hybrid_5.\vrfr(1^{\secparam},\inst) \rangle \\ \bigwedge\\ \left( \td^* = \td \right)}\ :\ \td^* \leftarrow \ext(1^{\secparam},\qaux) \right] \\ 
& \leq & \negl(\secparam)
\end{eqnarray*}

\noindent From the above hybrids, it follows that $\varepsilon_1 \leq \negl(\secparam)$. 

\end{proof}

\begin{claim}
$\varepsilon_2 \leq \negl(\secparam)$, for some negligible function $\negl$. 
\end{claim}
\begin{proof}
Since the trapdoor $\td^*$ extracted from $\prvr^*$ is not equal to $\td$, this means that there is a $j \in [k]$ s.t. $sh_0^{(j)} \oplus sh_1^{(j)} \neq \td$, where $sh_0^{(j)}$ and $sh_1^{(j)}$ are the unique values (uniqueness follows from perfect binding) committed to in $(\bfc_0^{(j)})^*$ and $(\bfc_1^{(j)})^*$ respectively. 
\par From the soundness of $\protwi$, it then follows that the probability that the verifier accepts is negligible.   %Moreover, from the perfect binding property of $\comm$, there does not exist shares $(sh^{(j)}_0)',(sh^{(j)}_1)'$ such that $(sh^{(j)}_0)' \neq sh^{(j)}_0$ or $(sh^{(j)}_1)' \neq sh^{(j)}_1$ such that the commitments of $(sh^{(j)}_0)'$ and $(sh^{(j)}_1)'$, for some fixed random strings, are $(\bfc_0^{(j)})^*$ and $(\bfc_1^{(j)})^*$ respectively.
\end{proof}

\paragraph{Zero-Knowledge.} Let $\vrfr^*$ be the malicious QPT verifier. We describe the simulator $\simulator$ as follows. 
\begin{itemize}
    \item It receives $\bfc$ from $\vrfr^*$. 
    \item Suppose $\ext$ be the extractor of $\protextcl$ associated with $\protextcl.\sender^*$, where $\protextcl.\sender^*$ is the adversarial sender algorithm computed by $\vrfr^*$. Compute $\ext(1^{\secparam},\protextcl.\sender^*,\cdot)$ to obtain $\td
    ^*$. At any time, if $\vrfr^*$ aborts, $\simulator$ also aborts with the output, the current private state of $\vrfr^*$.  
    \item For every $j \in [k]$, it samples $sh_0^{(j)},sh_1^{(j)}$ uniformly at random subject to $sh_0^{(j)} \oplus sh_1^{(j)} = \td^*$. It then computes $(\bfc_0^{(j)})^*=\comm(1^{\secparam},sh_0^{(j)};\bfd_0^{(j)})$ and $(\bfc_1^{(j)})^*=\comm(1^{\secparam},sh_1^{(j)};\bfd_1^{(j)})$ and sends $((\bfc_0^{(j)})^*,(\bfc_1^{(j)})^*)$ to $\vrfr^*$.
    \item It receives bits $\{b^{(j)}\}_{j \in[k]}$ from $\vrfr^*$. 
    \item It sends $(sh_{b^{(j)}}^{(j)},\bfd_{b^{(j)}}^{(j)})$ from $\vrfr^*$. 
    \item It receives $(\rqext,\bfd,\td)$ from $\vrfr^*$. It then checks the following:
      \begin{itemize}
    \item Let $\trans_{\vrfr \rightarrow \prvr}$ be $(m_1^S,m_1^R,\ldots,m_{\rnd'}^S,m_{\rnd'}^R)$, where the message $m_i^R$ (resp., $m_i^S$) is the message sent by the receiver (resp., sender) in the $i^{th}$ round\footnote{We remind the reader that in every round, only one party speaks.} and $\rnd'$ is the number of rounds of $\protextcl$. Let the message produced by $\protextcl.\sender\left( 1^{\secparam},\bfc,(\bfd,\td);\rqext \right)$ in the $i^{th}$ round be $\widetilde{m}_{i}^S$. 
    \item If for any $i \in [\rnd']$, $\widetilde{m}_i^S \neq {m}_i^S$ then $\simulator$ aborts. If $\td \neq \td^*$ then $\simulator$ aborts. 
    \end{itemize}
    \item $\simulator$ executes $\protwi$ with $\vrfr^*$ on input instance  $\left(\inst,\td,\left\{(\bfc_0^{(j)})^*,(\bfc_1^{(j)})^*\right\}_{j \in [k]}\right)$. The witness $\simulator$ uses in $\protwi$ is $\left(\bot,\left\{(sh_0^{(j)},\bfd_0^{(j)},sh_1^{(j)},\bfd_1^{(j)})\right\}_{j \in[k]} \right)$. If $\vrfr$ aborts at any point in time, $\simulator$ also aborts and outputs the current state of the verifier. 
    \item Otherwise, output the current state of the verifier. 
   
\end{itemize}

\noindent We prove the indistinguishability of the view of the verifier when interacting with the honest prover versus the view of the verifier when interacting with the simulator. Consider the following hybrids. \\

\noindent $\underline{\hybrid_1}$: The output of this hybrid is the view of $\vrfr^*$ when interacting with $\prvr$. That is, the output of the hybrid is $\view_{\vrfr^*}\left( \langle \prvr(1^{\secparam},\inst,\witness),\vrfr^*(1^{\secparam},\inst,\cdot) \rangle \right)$.\\ 

\noindent $\underline{\hybrid_2}$: We define a hybrid prover $\hybrid_2.\prvr$ as follows: it first receives $\bfc$ from $\vrfr^*$. It computes  $\ext(1^{\secparam},\protextcl.\sender^*,\cdot)$ to obtain $\td^*$. It then sends $(\bfc_0^{(j)})^*$ and $(\bfc_1^{(j)})^*$, where $(\bfc_0^{(j)})^*$ and $(\bfc_1^{(j)})^*$ are commitments of $sh_0^{(j)},sh_1^{(j)}$ respectively and $sh_0^{(j)},sh_1^{(j)}$ are sampled uniformly at random. It receives $b$ from $\vrfr^*$. It then sends $(sh_b^{(j)},\bfd_b^{(j)})$ to $\vrfr^*$. It then receives $(\rqext,\bfd,\td)$ from $\vrfr^*$. It then checks the following: 
\begin{itemize}
    \item Let $\trans_{\vrfr \rightarrow \prvr}$ be $(m_1^S,m_1^R,\ldots,m_{\rnd'}^S,m_{\rnd'}^R)$, where the message $m_i^R$ (resp., $m_i^S$) is the message sent by the receiver (resp., sender) in the $i^{th}$ round and $\rnd'$ is the number of rounds of $\protextcl$. Let the message produced by $\protextcl.\sender\left( 1^{\secparam},\bfc,(\bfd,\td);\rqext \right)$ in the $i^{th}$ round be $\widetilde{m}_{i}^S$. 
    \item If for any $i \in [\rnd']$, $\widetilde{m}_i^S \neq {m}_i^S$ then $\hybrid_2.\prover$ aborts. If $\td \neq \td^*$ then $\hybrid_2.\prover$ aborts. 
    \end{itemize}
\noindent $\hybrid_2.\prvr$ finally executes $\protwi$ with $\vrfr^*$; it still uses $\witness$ in $\protwi$. 
\par We claim the following holds: 
$$\view_{\vrfr^*} \left( \langle \prvr(1^{\secparam},\inst,\witness),\vrfr^*(1^{\secparam},\inst,\cdot) \rangle  \right) \approx_{Q} \view_{\vrfr^*} \left( \hybrid_2.\prvr(1^{\secparam},\inst,\witness),\vrfr^*(1^{\secparam},\inst,\cdot) \right) $$
\noindent There are two cases:
\begin{itemize}
    \item $\protextcl.\sender^*$ does not behave according to the protocol (i.e., not semi-malicious): The view of the verifier when interacting with $\hybrid_2.\prover$ is indistinguishable from the view of the verifier when interacting with the honest prover, from the indistinguishability of extraction against malicious senders property (Claim~\ref{clm:aborting}). 
    \item $\protextcl.\sender^*$ behaves according to the protocol (i.e., it is semi-malicious): In this case, $\protextcl.\ext$ is able to extract $\td$ with probability negligibly close to 1. Moreovoer, as before, the view of the verifier when interacting with the honest prover is indistinguishable from $\hybrid_{2}.\prover$ from Claim~\ref{clm:aborting}. 
    
\end{itemize} %Moreover, at any point in time if $\protextcl.\sender^*$ aborts, by Claim~\ref{clm:aborting}, we have that the state of $\vrfr^*$ output by $\simulator$ is indistinguishable from the state of $\vrfr^*$ when interacting with the honest prover. \\ 
\ \\
\noindent $\underline{\hybrid_3}$: We define a hybrid prover $\hybrid_3.\prvr$ as follows: it behaves exactly like $\hybrid_2.\prvr$ except that it computes the commitments $(\bfc_0^{(j)})^*$ and $(\bfc_1^{(j)})^*$ as commitments of $sh_0^{(j)}$ and $sh_1^{(j)}$, where $sh_0^{(j)} \oplus sh_1^{(j)} = \td$.

\par The following holds from the quantum-computational hiding property of $\comm$ following the same argument as~\cite{PW09}: 
$$\view_{\vrfr^*} \left( \langle \hybrid_2.\prvr(1^{\secparam},\inst,\witness),\vrfr^*(1^{\secparam},\inst,\cdot) \rangle  \right) \approx_{Q} \view_{\vrfr^*} \left( \hybrid_3.\prvr(1^{\secparam},\inst,\witness),\vrfr^*(1^{\secparam},\inst,\cdot) \right) $$

\noindent $\underline{\hybrid_4}$: We define a hybrid prover $\hybrid_4.\prvr$ as follows: it behaves exactly like $\hybrid_3.\prvr$ except that it uses the witness $(\bot,(sh_0^{(j)},\bfd_0^{(j)},sh_1^{(j)},\bfd_1^{(j)}))$ in $\protwi$ instead of $(\witness,\bot)$. Note that the description of $\hybrid_4.\prvr$ is identical to the description of $\simulator$.  
\par The following holds from the quantum witness indistinguishability property of $\protwi$: 
\begin{eqnarray*}
& & \view_{\vrfr^*} \left( \langle \hybrid_3.\prvr(1^{\secparam},\inst,\witness),\vrfr^*(1^{\secparam},\inst,\cdot) \rangle  \right)\\
& \approx_{Q} & \view_{\vrfr^*} \left( \hybrid_4.\prvr(1^{\secparam},\inst,\witness),\vrfr^*(1^{\secparam},\inst,\cdot) \right) \\
& \equiv & \simulator(1^{\secparam},\inst,\cdot)
\end{eqnarray*}

\end{proof}
}

\fullversion{
\subsubsection{On Classical Verifiers}
\label{sec:class:verifiers}
A desirable property from a QZK protocol is if the verifier is classical then so is the simulator. Our protocol as described above doesn't satisfy this property. That is, our simulator is still a QPT algorithm even if the malicious verifier is classical. However, we can do a simple modification to our QZK protocol (Figure~\ref{fig:qzk}) to satisfy this desired property.
\par The modification is as follows: in addition to the cQEXT protocol, also sequentially execute a constant round classical extractable commitment scheme satisfying perfectly binding~\cite{PW09}. In the classical scheme, the verifer takes the role of the committer committing to $\bfc$ and $\bfd$; note that these are the same values it commits to in the cQEXT protocol as well. Note that this wouldn't affect soundness; the classical malicious prover will still be unable to learn $\bfd$ from the classical extractable commitment scheme, from its hiding property. 
\par To argue zero-knowledge, first consider the following two simulators: 
\begin{itemize}
    \item $\simr_{c}$: This simulator runs the extractor in the classical extractable commitment scheme to extract $\bfd$. It then runs the honest receiver to interact with the verifier in the cQEXT protocol. The rest of the steps is identical to the simulator described in the proof of Lemma~\ref{lem:qzk:classical}.
    \item $\simr_{q}$: This simulator runs the honest receiver to interact with the verifier in the classical extractable commitment scheme. It then runs the extractor in the cQEXT protocol to extract $\bfd$. The rest of the steps is identical to the simulator described in the proof of Lemma~\ref{lem:qzk:classical}.  
\end{itemize}
If the malicious verifier is classical PPT then $\simr_c$ can successfully carry out the simulation whereas if the malicious verifier is QPT then $\simr_q$ is successful. While we wouldn't know whether the malicious verifier is classical PPT or not, we know for a fact that one of two simulators will succeed. 
}

%%%%%%%%%%%%%%%%%%%%%%%%%%%%%%%%%%%%%%%%%%%%%%%%%%%%%%%%%%%%%%%%%%%%%%%%%%%%%%%%%%%%%%%%%%%%%%%%%%%%%%%%%%%%%%%%%%%%%%%%%%%%%%%%%%%%%%%%%%%%%%%%%%%%%%%%%%%%%%%%%%
\section{QEXT Secure Against Quantum Adversaries}
\label{sec:qext}

\renewcommand{\otpad}{\mathsf{otp}}
\newcommand{\qwi}{\mathsf{qwi}} 
\subsection{Construction of $\protext$}
\label{ssec:qext:cons}
\noindent We present a construction of quantum extraction protocols secure against quantum adversaries, denoted by $\protextq$. First, we describe the tools used in this construction. 

\fullversion{
\paragraph{Tools.} 
\begin{itemize}    
    \item Quantum-secure computationally-hiding and perfectly-binding non-interactive commitments $\comm$ (see Section~\ref{sec:prelims:commit}). 
    \item Quantum fully homomorphic encryption scheme with some desired properties, $(\fhe.\gen,\allowbreak \fhe.\enc,\allowbreak  \fhe.\dec,\allowbreak  \fhe.\eval)$. 
    \begin{itemize}
        \item It admits homomorphic evaluation of arbitrary computations, 
        \item It admits perfect correctness, 
        \item The ciphertext of a classical message is also classical.  
     %   \item 
    \end{itemize}
    \noindent We show in Section~\ref{ssec:prelims:qfhe} that there are $\fhe$ schemes satisfying the above properties. 
    %\todo[inline]{Formalize this properties}
    
    \item Quantum-secure two-party secure computation $\protpc$ with the following properties (see Section~\ref{ssec:prelims:sfe}): 
    \begin{itemize}
        \item Only one party receives the output. We designate the party receiving the output as the receiver $\protpc.\receiver$ and the other party to be $\protpc.\sender$. 
    
        \item Security against quantum passive senders. 
        \item IND-Security against quantum malicious receivers. 
    \end{itemize}
    
    \item Quantum-secure lockable obfuscation $\lockobfscheme=(\lobf,\leval)$ for $\lockclass$, where every circuit $\lockC$, parameterized by $(\bfr,\bfk,\sk_1,\ct^*)$, in $\lockclass$ is defined in Figure~\ref{fig:lockfunc}. Note that $\lockclass$ is a compute-and-compare functionality (see Section~\ref{ssec:prelims:lobfs}).

\end{itemize}
}
\submversion{
\begin{itemize}    
    \item Quantum-secure computationally-hiding and perfectly-binding non-interactive commitments $\comm$. 
    \item Quantum fully homomorphic encryption scheme with some desired properties, $(\fhe.\gen,\allowbreak \fhe.\enc,\allowbreak  \fhe.\dec,\allowbreak  \fhe.\eval)$. 
    \begin{itemize}
        \item It admits homomorphic evaluation of arbitrary computations, 
        \item It admits perfect correctness, 
        \item The ciphertext of a classical message is also classical.  
     %   \item 
    \end{itemize}
    \noindent We show in the full version that there are $\fhe$ schemes satisfying the above properties. 
    %\todo[inline]{Formalize this properties}
    
    \item Quantum-secure two-party secure computation $\protpc$ with the following properties: 
    \begin{itemize}
        \item Only one party receives the output. We designate the party receiving the output as the receiver $\protpc.\receiver$ and the other party to be $\protpc.\sender$. 
    
        \item Security against quantum passive senders. 
        \item IND-Security against quantum malicious receivers. 
    \end{itemize}
    
    \item Quantum-secure lockable obfuscation $\lockobfscheme=(\lobf,\leval)$ for $\lockclass$, where every circuit $\lockC$, parameterized by $(\bfr,\bfk,\sk_1,\ct^*)$, in $\lockclass$ is defined in Figure~\ref{fig:lockfunc}. Note that $\lockclass$ is a compute-and-compare functionality.

\end{itemize}

}

\submversion{
\clearpage
}

\begin{figure}[!htb]
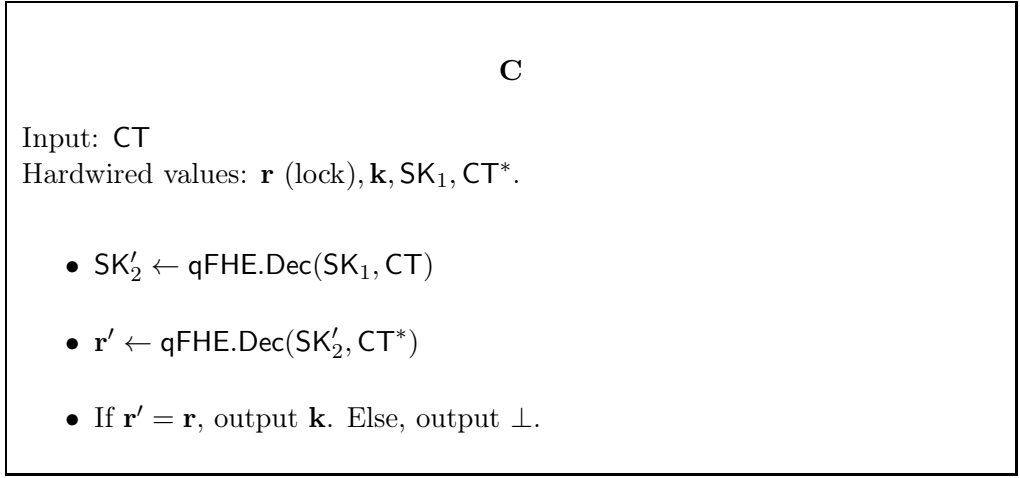

\begin{tabular}{|p{13cm}|}
\hline 
\begin{center}
    $\lockC$
\end{center}
Input: $\ct$\\
Hardwired values: $\bfr \  (\mathrm{lock}),\bfk,\sk_1,\ct^*$.\\ 
\begin{itemize}
 \setlength\itemsep{1em}
    \item $\sk'_2 \leftarrow \fhe.\dec(\sk_1,\ct)$
    \item $\bfr' \leftarrow \fhe.\dec(\sk'_2,\ct^*)$
    \item If $\bfr'=\bfr$, output $\bfk$. Else, output $\bot$. 
\end{itemize}
\\
\hline
\end{tabular}  
    \caption{Circuits used in the lockable obfuscation}
    \label{fig:lockfunc}
\end{figure}

\newcommand{\recin}{\mathbf{y}}
\begin{figure}[!htb]
\begin{tabular}{|p{13cm}|}
\hline 
\begin{center}
    $f$
\end{center}
Input of sender: $(\td,\bfc,\bfc^*_1,\ldots,\bfc^*_{\ell},\sk_2)$\\
Input of receiver: $(\bfd,r_1,\ldots,r_{\ell})$
\ \\
\begin{itemize}
    \setlength\itemsep{1em}
    \item If $\left( \bfc \leftarrow \comm\left( 1^{\secparam},(r_1,\ldots,r_{\ell});\bfd\right) \right) \bigwedge \left(\forall i \in [\ell], \bfc^*_i \leftarrow \comm\left( 1^{\secparam},\td_{i}; r_i \right) \right)$, output $\sk_2$. Here, $\td_i$ denotes the $i^{th}$ bit of $\td$. 
    \item Otherwise, output $\bot$.
\end{itemize}
\ \\
\hline
\end{tabular}
\caption{Description of the function $f$ associated with the $\protpc$.}
\label{fig:sfefunc}
\end{figure}

\begin{figure}[!htb]
\begin{tabular}{|p{13cm}|}
\hline 

Input of sender: $(\inst,\witness)$. \\
Input of receiver: $\inst$ \\

\begin{itemize}
\setlength\itemsep{1em}

    \item $\receiver$: sample $(r_1,\ldots,r_{\ell}) \xleftarrow{\$}  \{0,1\}^{\ell \cdot \poly(\secparam)}$. Compute $\bfc \leftarrow \comm\left( 1^{\secparam},(r_1,\ldots,r_{\ell});\bfd\right)$, where $\ell=\secparam$ and $\bfd$ is the randomness used to compute $\bfc$. Send $\bfc$ to $\sender$.

    \item $\sender$:
    \begin{itemize}
    
        \item Compute the $\fhe.\setup$ twice;  $(\pk_i,\sk_i) \leftarrow \fhe.\setup(1^{\secparam})$ for $i \in \{1,2\}$. 
        \item Compute $\ct_1 \leftarrow \fhe.\enc(\pk_1,(\td || \witness))$, where $\td \xleftarrow{\$} \{0,1\}^{\secparam}$.
      
        \item Compute $\obfC \leftarrow \lobf(1^{\secparam},\lockC[\bfr,\bfk,\sk_1,\ct^*])$, where $\bfr \xleftarrow{\$} \{0,1\}^
        {\secparam}$ and $\bfk \xleftarrow{\$} \{0,1\}^{\secparam}$, $\ct^*$ is defined below and $\lockC[\bfr,\bfk,\sk_1,\ct^*]$ is defined in Figure~\ref{fig:lockfunc}. \begin{itemize}
          
          \item $\ct^* \leftarrow \fhe.\enc \left(\pk_2,\bfr \right)$
      \end{itemize}
    \end{itemize}
      Send $\msg_1=\left(\ct_1,\obfC, \otpad := \bfk \oplus \sk_1 \right)$. 
    
    \item $\receiver$: compute $\bfc^*_i \leftarrow \comm \left(1^{\secparam},0;r_i \right)$ for $i \in [\ell]$. Send $(\bfc^*_1,\ldots,\bfc^*_{\ell})$ to $\sender$.   
    
    \item $\sender$ and $\receiver$ run $\protpc$, associated with the two-party functionality $f$ defined in Figure~\ref{fig:sfefunc}; $\sender$ takes the role of $\protpc.\sender$ and $\receiver$ takes the role of $\protpc.\receiver$.  The input to $\protpc.\sender$ is $(\td,\bfc,\bfc_1^*,\ldots,\bfc_{\ell}^*,\sk_2)$ and the input to $\protpc.\receiver$ is $(\bfd,r_1,\ldots,r_{\ell})$.
 
\end{itemize}
\ \\
\hline
\end{tabular}
\caption{Quantum Extraction Protocol $(\sender,\receiver)$}
\label{fig:QEXTcons}
\end{figure}

\paragraph{Construction.} We construct a protocol $(\sender,\receiver)$ in Figure~\ref{fig:QEXTcons} for a NP language $\lang$, and the following lemma shows that $(\sender,\receiver)$ is a quantum extraction protocol.

%%%%%%%%%%%%%%%%%%%%%%%%%%%%%%%%%%%%%%%%%%%%%%%%%%%%%%
%%%%%%%%%%%%%%%%%%%%%%%%%%%%%%%%%%%%%%%%%%%%%%%%%%%%%%%%

\submversion{

We prove the following lemma in the full version.

\begin{lemma}
\label{lem:proof:qqext}
Assuming the quantum security of $\comm$, $\protpc$, $\fhe$ and $\lockobfscheme$,  $(\sender,\receiver)$ is a quantum extraction protocol for $\lang$ secure against quantum adversaries.
\end{lemma}
\newpage
}

%%%%%%%%%%%%%%%%%%%%%%%%%%%%%%%%%%%%%%%%%%%%%%%%%%%%%%
%%%%%%%%%%%%%%%%%%%%%%%%%%%%%%%%%%%%%%%%%%%%%%%%%%%%%%%%

%%%%%%%%%%%%%%%%%%%%%%%%%%%%%%%%%%%%%%%%%%%%%%%%%%%%%%
%%%%%%%%%%%%%%%%%%%%%%%%%%%%%%%%%%%%%%%%%%%%%%%%%%%%%%%%
\fullversion{
%We prove the following lemma in Section~\ref{sec:proofs-qqext}. 

\begin{lemma} Assuming the quantum security of $\comm$, $\protpc$, $\fhe$, and  $\lockobfscheme$, $(\sender,\receiver)$ is a quantum extraction protocol for $\lang$ secure against quantum adversaries.
\end{lemma}
\begin{proof} 
\
\paragraph{Quantum Zero-Knowledge.} Let $(\inst,\witness) \in \rel$, and let $\receiver^*$ be a QPT malicious receiver.  Associated with $\receiver^*$ is the QPT algorithm $\Simu$ -- in fact, $\Simu$ is a classical PPT algorithm that only uses $\receiver^*$ as a black-box -- defined below.

%In the first round, $\Simu$ will compute the same things as the honest sender, except that it will encrypt $\bot$ in $\ct_1$ instead of the witness $\witness$.

\paragraph{Description of $\Simu$.} 

\begin{itemize}
    
    \item It first receives $\bfc$ from $\receiver^*$. It performs the following operations: 
    \begin{itemize}
    \item Compute the $\fhe.\setup$ to obtain $(\pk_1, \sk_1)$.
    \item Compute $\ct_1 \leftarrow \fhe.\enc(\pk_1, \bot)$.
    \item Compute the obfuscated circuit $\obfC \leftarrow \locksimr\left(1^{\secparam},1^{|\lockC|} \right)$. 
    \item Sample $\otpad \xleftarrow{\$} \{0,1\}^{|\sk_1|}$.  
    
\end{itemize}
\noindent Send $(\ct_1,\obfC,\otpad)$.
    
    \item It then receives $(\bfc_1^*,\ldots,\bfc_{\ell}^*)$ from the receiver.  
    
    \item It executes $\protpc$ with $\receiver^*$; $\Simu$ takes the role of $\protpc.\sender$ with the input $\bot$.
     
    \item Finally, it outputs the final state of $\receiver^*$.  
     
\end{itemize}

\noindent We show below that the view of $\receiver^*$ when interacting with the honest sender is indistinguishable, by a QPT distinguisher, from the output of $\Simu$. Consider the following hybrids:\\

\noindent $\underline{\hybrid_1}$: In this hybrid, $\receiver^*$ is interacting with the honest sender $\sender$. The output of this hybrid is the output of $\receiver^*$.\\

\noindent $\underline{\hybrid_2}$: In this hybrid, we define a hybrid sender, denoted by $\hybrid_2.\sender$: it behaves exactly like $\sender$ except that in $\protpc$, the input of $\protpc.\sender$ is $\bot$.
\par Consider the following claim.  

\begin{claim}
\label{clm:hyb1hyb2}
$\view_{\receiver^*}\left( \langle \sender(1^\secparam,\inst,\witness),\receiver^*(1^\secparam,\inst, \cdot)\rangle\right) \approx_{Q} \view_{\receiver^*}\left( \langle \hybrid_2.\sender(1^\secparam,\inst,\witness),\receiver^*(1^\secparam,\inst, \cdot)\rangle\right).$
\end{claim}
\begin{proof}
To prove this claim, we first need to show that the probability that the receiver $\receiver^*$ commits to $\witness$ is negligible. Consider the following claim. 

\begin{claim}\label{clm:mainclm} Assuming the quantum security of $\comm$, $\lockobfscheme$ and $\fhe$, the following holds:
$$\Pr\left[\substack{\exists r_1,\ldots,r_{\ell},\bfd,\\ \left( \bfc = \comm\left( 1^{\secparam},(r_1,\ldots,r_{\ell});\bfd\right) \right)\\ \bigwedge\\ \left(\forall i \in [\ell], \bfc^*_i = \comm\left( 1^{\secparam},\td_{i}; r_i\right) \right) = 1}\ :\\ \substack{\bfc \leftarrow \receiver^*(1^{\secparam},\inst,\cdot)\\
\td \xleftarrow{\$} \{0,1\}^{\secparam}\\
(\pk_i,\sk_i) \leftarrow \fhe.\setup(1^{\secparam}), \forall i \in \{1,2\}\\
\ct_1 \leftarrow \fhe.\enc(\pk_1,(\td||\witness))\\
\bfr \xleftarrow{\$} \{0,1\}^{\secparam}\\
\bfk \xleftarrow{\$} \{0,1\}^{|\sk_1|}\\
\ct^* \leftarrow \fhe.\enc(\pk_2,\bfr) \\
\obfC \leftarrow \lobf(1^{\secparam},\lockC[\bfr,\bfk,\sk_1,\ct^*]) \\
\otpad =\bfk \oplus \sk_1\\
(\bfc_1^*,\ldots,\bfc_{\ell}^*) \leftarrow \receiver^*(1^{\secparam},\inst,\cdot) }\right] \leq \negl(\secparam),$$
for some negligible function $\negl$.
\end{claim}

\begin{proof}  
We define the event $\badevent_1$ as follows: 

\begin{quote}
    $\badevent_1=1$ if there exists $ r_1,\ldots,r_{\ell},\bfd$ such that $$\left( \bfc = \comm\left( 1^{\secparam},(r_1,\ldots,r_{\ell});\bfd\right) \right) \bigwedge \left(\forall i \in [\ell], \bfc^*_i = \comm\left( 1^{\secparam},\td_{i}; r_i\right) \right) = 1,$$ where:
    \begin{itemize}
        \item $\bfc \leftarrow \receiver^*(1^{\secparam},\inst,\cdot)$,
\item $\ct_1 \leftarrow \fhe.\enc(\pk_1,(\td||\witness))$, where $(\pk_i,\sk_i) \leftarrow \fhe.\setup(1^{\secparam}), \forall i \in \{1,2\}$ and $\td \xleftarrow{\$} \{0,1\}^{\secparam}$,
\item $\obfC \leftarrow \lobf(1^{\secparam},\lockC[\bfr,\bfk,\sk_1,\ct^*])$, where $\bfr \xleftarrow{\$} \{0,1\}^{\secparam}$, $\bfk \xleftarrow{\$} \{0,1\}^{|\sk_1|}$ and $\ct^* \leftarrow \fhe.\enc(\pk_2,\bfr)$,
\item $\otpad =\bfk \oplus \sk_1$ and,
\item $\receiver^*(1^{\secparam},\inst,\cdot)$ on input $(\ct,\obfC,\otpad)$ outputs $(\bfc_1^*,\ldots,\bfc_{\ell}^*)$. 
    \end{itemize}
    Otherwise, $\badevent_1=0$. 
\end{quote}

\noindent Define $\bfp_1$ to be $\bfp_1=\prob[\badevent_1=1]$.  \\

\noindent We define a hybrid event $\badevent_{1.1}$ as follows: 
\begin{quote}
    $\badevent_{1.1}=1$ if there exists $ r_1,\ldots,r_{\ell},\bfd$ such that $$\left( \bfc = \comm\left( 1^{\secparam},(r_1,\ldots,r_{\ell});\bfd\right) \right) \bigwedge \left(\forall i \in [\ell], \bfc^*_i = \comm\left( 1^{\secparam},\td_{i}; r_i\right) \right) = 1,$$ where:
    \begin{itemize}
        \item $\bfc \leftarrow \receiver^*(1^{\secparam},\inst,\cdot)$,
\item $\ct_1 \leftarrow \fhe.\enc(\pk_1,(\td||\witness))$, where $(\pk_i,\sk_i) \leftarrow \fhe.\setup(1^{\secparam}), \forall i \in \{1,2\}$ and $\td \xleftarrow{\$} \{0,1\}^{\secparam}$,
\item $\obfC \leftarrow \lobf(1^{\secparam},\lockC[\bfr,\bfk,\sk_1,\ct^*])$, where $\bfr \xleftarrow{\$} \{0,1\}^{\secparam}$, $\bfk \xleftarrow{\$} \{0,1\}^{|\sk_1|}$ and\\ \underline{$\ct^* \leftarrow \fhe.\enc(\pk_2,\bot)$}, 
\item $\otpad =\bfk \oplus \sk_1$ and,
\item $\receiver^*(1^{\secparam},\inst,\cdot)$ on input $(\ct,\obfC,\otpad)$ outputs $(\bfc_1^*,\ldots,\bfc_{\ell}^*)$. 
    \end{itemize}
    Otherwise, $\badevent_{1.1}=0$. 
\end{quote}
We define $\bfp_{1.1}$ as $\bfp_{1.1} = \prob[\badevent_{1.1}=1]$.
\par From the quantum security of $\fhe$, it holds that $|\bfp_1-\bfp_{1.1}| \leq \negl(\secparam)$ for some negligible function $\negl$. Note that we crucially rely on the fact that $\protpc$, that requires the sender to input $\sk_2$, is only executed after the receiver sends $(\bfc^*_1,\ldots,\bfc^*_{\ell})$.\\

\noindent We define a hybrid event $\badevent_{1.2}$ as follows: 
\begin{quote}
    $\badevent_{1.2}=1$ if there exists $ r_1,\ldots,r_{\ell},\bfd$ such that $$\left( \bfc = \comm\left( 1^{\secparam},(r_1,\ldots,r_{\ell});\bfd\right) \right) \bigwedge \left(\forall i \in [\ell], \bfc^*_i = \comm\left( 1^{\secparam},\td_{i}; r_i\right) \right) = 1,$$ where:
    \begin{itemize}
        \item $\bfc \leftarrow \receiver^*(1^{\secparam},\inst,\cdot)$,
\item $\ct_1 \leftarrow \fhe.\enc(\pk_1,(\td||\witness))$, where $(\pk_i,\sk_i) \leftarrow \fhe.\setup(1^{\secparam}), \forall i \in \{1,2\}$ and $\td \xleftarrow{\$} \{0,1\}^{\secparam}$,
\item \underline{$\obfC \leftarrow \locksimr\left(1^{\secparam},1^{|\lockC|} \right)$}, 
\item $\otpad =\bfk \oplus \sk_1$ and,
\item $\receiver^*(1^{\secparam},\inst,\cdot)$ on input $(\ct,\obfC,\otpad)$ outputs $(\bfc_1^*,\ldots,\bfc_{\ell}^*)$. 
    \end{itemize}
    Otherwise, $\badevent_{1.2}=0$. 
\end{quote}
\noindent We define $\bfp_{1.2}$ as $\bfp_{1.2}=\prob[\badevent_{1.2}=1]$. From the quantum security of $\lockobfscheme$, it follows that $|\bfp_{1.1}-\bfp_{1.2}| \leq \negl(\secparam)$. Note that we crucially use the fact that the lock $\bfr$ is uniformly sampled and independently of the function that is obfuscated.\\

\noindent We define a hybrid event $\badevent_{1.3}$ as follows: 
\begin{quote}
    $\badevent_{1.3}=1$ if there exists $ r_1,\ldots,r_{\ell},\bfd$ such that $$\left( \bfc = \comm\left( 1^{\secparam},(r_1,\ldots,r_{\ell});\bfd\right) \right) \bigwedge \left(\forall i \in [\ell], \bfc^*_i = \comm\left( 1^{\secparam},\td_{i}; r_i\right) \right) = 1,$$ where:
    \begin{itemize}
        \item $\bfc \leftarrow \receiver^*(1^{\secparam},\inst,\cdot)$,
\item $\ct_1 \leftarrow \fhe.\enc(\pk_1,(\td||\witness))$, where $(\pk_i,\sk_i) \leftarrow \fhe.\setup(1^{\secparam}), \forall i \in \{1,2\}$ and $\td \xleftarrow{\$} \{0,1\}^{\secparam}$,
\item $\obfC \leftarrow \locksimr\left(1^{\secparam},1^{|\lockC|} \right)$, 
\item \underline{$\otpad \xleftarrow{\$} \{0,1\}^{|\sk_1|}$} and,
\item $\receiver^*(1^{\secparam},\inst,\cdot)$ on input $(\ct,\obfC,\otpad)$ outputs $(\bfc_1^*,\ldots,\bfc_{\ell}^*)$. 
    \end{itemize}
    Otherwise, $\badevent_{1.3}=0$. 
\end{quote}
We define $\bfp_{1.3}$ as $\bfp_{1.3}=\prob[\badevent_{1.3}=1]$. Observe that $\bfp_{1.2}=\bfp_{1.3}$. \\

\noindent We define a hybrid event $\badevent_{1.4}$ as follows: 
\begin{quote}
    $\badevent_{1.4}=1$ if there exists $ r_1,\ldots,r_{\ell},\bfd$ such that $$\left( \bfc = \comm\left( 1^{\secparam},(r_1,\ldots,r_{\ell});\bfd\right) \right) \bigwedge \left(\forall i \in [\ell], \bfc^*_i = \comm\left( 1^{\secparam},\td_{i}; r_i\right) \right) = 1,$$ where:
    \begin{itemize}
        \item $\bfc \leftarrow \receiver^*(1^{\secparam},\inst,\cdot)$,
\item \underline{$\ct_1 \leftarrow \fhe.\enc(\pk_1,\bot)$}, where $(\pk_i,\sk_i) \leftarrow \fhe.\setup(1^{\secparam}), \forall i \in \{1,2\}$ and $\td \xleftarrow{\$} \{0,1\}^{\secparam}$,
\item $\obfC \leftarrow \locksimr\left(1^{\secparam},1^{|\lockC|} \right)$, 
\item $\otpad \xleftarrow{\$} \{0,1\}^{|\sk_1|}$ and,
\item $\receiver^*(1^{\secparam},\inst,\cdot)$ on input $(\ct,\obfC,\otpad)$ outputs $(\bfc_1^*,\ldots,\bfc_{\ell}^*)$. 
    \end{itemize}
    Otherwise, $\badevent_{1.4}=0$. 
\end{quote}
\noindent We define $\bfp_{1.4}$ as $\bfp_{1.4}=\prob[\badevent_{1.4}=1]$. From the quantum security of $\fhe$, it follows that $|\bfp_{1.3}-\bfp_{1.4}| \leq \negl(\secparam)$. Moreover, note that $\bfp_{1.4}=2^{-\secparam}$ since $\td$ is information-theoretically hidden from $\receiver^*$. Thus, we have that $\bfp_{1} \leq \negl(\secparam)$.    

\fullversion{the non-uniformity requirement of the primitives needs to be stated explicitly above..}

\end{proof}

\noindent We now use Claim~\ref{clm:mainclm} to prove Claim~\ref{clm:hyb1hyb2}. Conditioned on $\badevent_1 \neq 1$, it holds that the view of $\receiver^*$ after its interaction with $\sender$ is indistinguishable (by a QPT algorithm) from the view of $\receiver^*$ after its interaction with $\hybrid_2.\sender$; this follows from the IND-security of $\protpc$ against quantum receivers since $f((\td,\bfc,\bfc_1^*,\ldots,\bfc_{\ell}^*,\sk_2),(\bfd,r_1,\ldots,r_{\ell}))=f((\bot),(\bfd,r_1,\ldots,r_{\ell}))$.     

\end{proof}

\noindent $\underline{\hybrid_3}$: We define a hybrid sender, denoted by $\hybrid_3.\sender$: it behaves exactly like $\hybrid_2.\sender$ except that $\ct^*$ in $\obfC$ is generated as $\ct^* \leftarrow \fhe.\enc(\pk_2,\bot)$. 
\par Assuming the quantum security of $\fhe$, we have:
$$\view_{\receiver^*}\left( \langle \hybrid_{2}.\sender(1^\secparam,\inst,\witness),\receiver^*(1^\secparam,\inst, \cdot)\rangle\right) \approx_{Q} \view_{\receiver^*}\left( \langle \hybrid_3.\sender(1^\secparam,\inst,\witness),\receiver^*(1^\secparam,\inst, \cdot)\rangle\right)$$

\noindent $\underline{\hybrid_4}$: We define a hybrid sender, denoted by $\hybrid_4.\sender$: it behaves exactly like $\hybrid_3.\sender$ except that $\obfC$ is generated as $\obfC \leftarrow \locksimr\left(1^{\secparam},1^{|\lockC|} \right)$. 
\par Assuming the quantum security of $\lockobfscheme$, we have: 
$$\view_{\receiver^*}\left( \langle \hybrid_{3}.\sender(1^\secparam,\inst,\witness),\receiver^*(1^\secparam,\inst, \cdot)\rangle\right) \equiv \view_{\receiver^*}\left( \langle \hybrid_4.\sender(1^\secparam,\inst,\witness),\receiver^*(1^\secparam,\inst, \cdot)\rangle\right)$$

\noindent $\underline{\hybrid_5}$: We define a hybrid sender, denoted by $\hybrid_5.\sender$: it behaves exactly like $\hybrid_4.\sender$ except that $\otpad$ is generated uniformly at random. 
\par The following holds unconditionally: 
$$\view_{\receiver^*}\left( \langle \hybrid_{4}.\sender(1^\secparam,\inst,\witness),\receiver^*(1^\secparam,\inst, \cdot)\rangle\right) \equiv \view_{\receiver^*}\left( \langle \hybrid_5.\sender(1^\secparam,\inst,\witness),\receiver^*(1^\secparam,\inst, \cdot)\rangle\right)$$

\noindent $\underline{\hybrid_6}$: We define a hybrid sender, denoted by $\hybrid_6.\sender$: it behaves exactly like $\hybrid_5.\sender$ except that $\ct_1$ is generated as $\ct_1 \leftarrow \fhe.\enc(\pk_1,\bot)$. 
\par Assuming the quantum security of $\fhe$, we have:
$$\view_{\receiver^*}\left( \langle \hybrid_{5}.\sender(1^\secparam,\inst,\witness),\receiver^*(1^\secparam,\inst, \cdot)\rangle\right) \approx_{Q} \view_{\receiver^*}\left( \langle \hybrid_6.\sender(1^\secparam,\inst,\witness),\receiver^*(1^\secparam,\inst, \cdot)\rangle\right)$$

\noindent Since $\hybrid_6.\sender$ is identical to $\Simu$, the proof of quantum zero-knowledge follows.  

\newcommand{\RegR}{\mathsf{Reg}_\receiver}
\newcommand{\RegS}{\mathsf{Reg}_\sender^*}

\paragraph{Extractability.} Let $\sender^* = (\sender_1^*,\sender_{2}^*)$ be a semi-malicious QPT, where $S_2^*$ is the QPT involved in $\protpc$. Denote by $\receiver=(\receiver_1,\receiver_2,\receiver_3)$ the PPT algorithms of the honest receiver.  In particular, $\receiver_3$ is the algorithm that the receiver runs in $\protpc$ protocol.  Let $$\cE_{\protpc}(\cdot \;;\bfd,r_1,...,r_\ell,\td, \witness, \bfc, \bfc^*):= \left\langle \receiver_3(1^\secparam, \bfd,r_1,\ldots,r_\ell), \sender_2^*(1^\secparam, \td, \witness,\bfc,\bfc^*, \cdot ) \right\rangle $$
be the interaction channel induced on the private quantum input of $\sender^*$ by the interaction with $\receiver$ in the $\protpc$ protocol for the functionality $f$ with inputs $\bfd,r_1,...,r_\ell, \td, \witness, \bfc, \bfc^*$.  Without loss of generality, assume that this channel also outputs the classical message output of $\protpc$.

Consider the following  extractor $\ext$, that takes as input the efficient quantum circuit description of $\sender^*(1^\secparam, \inst,\witness,\cdot)$, and the instance $\inst$.
\paragraph{$\ext(1^\secparam,S^*,\inst, \cdot)$:}
\begin{itemize}
    \item Run $\receiver_1$ to compute $\bfc$, $\bfd$, and $r_1,\ldots, r_\ell$.
    \item Apply the channel $S_1^*(1^\secparam, \inst, \witness, \bfc,\cdot)$.
    \item Let $(\ct_1, \obfC, \otpad)$ denote the classical messages outputted by $S_1^*$, and let $\rho$ denote the rest of the state.
    \item  With $\ct_1$, homomorphically commit to $\td$, obtaining
    $$\fhe.\enc(\pk_1, \bfc^*:=\comm(1^\secparam, \td))$$. 
    \item Encrypt $(\bfd, \bfc, r_1,\ldots,r_\ell)$, and $\rho$, and homomorphically apply the channel $\cE_{\protpc}(\cdot \;;\bfd,r_1,...,r_\ell,\td, \witness, \bfc, \bfc^*)$
    \item Let $\fhe.\enc(\pk_1, \protpc.\out \otimes \rho')$ be the output of the previous step, where $\protpc.\out$ is the classical output of the $\protpc$ protocol.
    
    \item Apply $\obfC$ to the $\fhe$ encryption of $\protpc.\out$. Note that we are assuming that classical messages have classical ciphertexts, so this computation is a classical one. Let $k$ be the output of $\obfC\left(\fhe.\enc(\pk_1,\protpc.\out)\right)$.
    \item Let $\sk_1 := k \oplus \otpad$, and decrypt $\ct_1$  with $\sk_1$. If the decryption is successful and the message $\witness$ is recovered, let $\ext_2$ output $\witness$.
    \item Use $\sk_1$ to decrypt the ciphertext $\fhe.\enc(\pk_1, \protpc.\out\otimes \rho')$, and let $\ext_1$ output $\rho'$.

\end{itemize}

\begin{claim}
$\view_{\sender^*}\left( \langle \sender^*(1^\secparam,\inst,\witness, \cdot), \receiver(1^\secparam,\inst) \rangle\right) \approx_{Q} \ext_1 \left(1^\secparam,\sender^*,\inst, \cdot \right)$
\end{claim}

\newcommand{\RegD}{\mathsf{R}_{\cD}}
\begin{proof}
Let $\RegD$ be the quantum register of a distinguisher $\cD$.
Let $\cF : \RegD \rightarrow \RegD$ be the following channels, parametrized by $\bfd,r_1,...,r_\ell,\td, \witness, \bfc, \bfc^*$,
$$\cF(\rho; \bfd,r_1,...,r_\ell, \witness, \bfc, \bfc^*):= \left(\left[\cE_{\protpc}(\cdot \; ; \bfd,r_1,...,r_\ell, \td, \witness, \bfc, \bfc^*) \circ S_1^*(1^\secparam,\inst, \witness,\bfc,\cdot)\right]\otimes \Id \right)\left(\rho\right).$$
The identity is acting on the distinguisher's private state, and the composition $\cE_{\protpc}(\cdot \; ; \bfd,r_1,...,r_\ell, \td, \witness, \bfc, \bfc^*) \circ S_1^*(1^\secparam,\inst,\witness,\bfc,\cdot)$ acts on the private state of $\sender^*$. We do not write $\td$ as a parameter to $\cF$, because $\td$ is generated by $S_1^*$ and assumed to be part of the sender's private state. We do add it as a parameter to $\cE_{\protpc}$ to be consistent and to remind ourselves that the $\td$ is input into the $\protpc$ protocol. 

Note that when $\bfd,r_1,\ldots,r_\ell, \bfc$ and $\bfc^*$ are generated by the honest $\receiver$ in the protocol, we have
$$\cF(\rho; \bfd,r_1,...,r_\ell,\witness, \bfc, \bfc^*)= \left(\view_{\sender^*}\left( \langle \sender^*(1^\secparam,\inst,\witness, \cdot), \receiver(1^\secparam,\inst) \rangle\right) \otimes \Id \right) (\rho)$$

We will show that when $\bfd,r_1,\ldots,r_\ell, \bfc$ are generated the same way as the honest $\receiver$ would generate them in the first round $\receiver_1$, but the commitment $\bfc^*=\bfc^*_1,\ldots,\bfc_\ell^*$ is a commitment to the witness, $\witness$, instead, we have
$$\cF(\rho; \bfd,r_1,...,r_\ell,\witness, \bfc, \bfc^*_{\witness})= \left(\ext_1 \left(1^\secparam,\sender^*,\inst, \cdot \right) \otimes \Id \right) (\rho)$$
Our goal is to show that these two cases, $\bfc^*$ and $\bfc^*_{\witness}$, are quantum computationally indistinguishable.
 
 To see why this last equation is true, we are using the perfect correctness of both the $\fhe$ scheme and of the lockable obfuscator, as well as the fact that the $\sender^*$ is semi-malicious, which means it has to follow the protocol.  This means that when $S_1^*$ outputs $(\ct_1,\obfC,\otpad)$, the extractor has a valid ciphertext $\ct_1$ encrypted with a key $\pk_1$, which in turn is one-time padded, $\sk_1 \oplus k = \otpad$. Furthermore, the one-time pad value $k$ is the output of $\obfC$ if an input releases the lock, and $\obfC$ is a correct lockable obfuscation of the desired circuit. 
 
 After this, the extractor performed $\cE_{\protpc}(\cdot \; ; \bfd,r_1,...,r_\ell, \td, \witness, \bfc, \bfc^*_{\witness})$ homomorphically, which results in the extractor having an encryption of $\sk_2$ under $\pk_1$. This is true because the extractor is able to commit to the witness inside the encryption, and the semi-malicious sender has to engage correctly in the $\protpc$. Since the extractor can now use the $\obfC$ to obtain $\sk_1$, we can summarize the whole operation of the extractor as follows. Let $(\ct_1, \obfC, \otpad) \otimes \rho'$ be the state of the distinguisher after $S_1^*$. Then, the extractor performs
 $$\left(\left(\dec(\sk_1, \cdot) \circ \eval\left(\cE_{\protpc}\left(\cdot \; ; \bfd,r_1,...,r_\ell,\td, \witness, \bfc, \bfc^*_{\witness}\right),\cdot\right) \circ \enc(\pk_1, \bfc^*_{\witness},\cdot) \right) \otimes \Id \right) \left( \rho' \right)$$
 
 By correctness of the $\fhe$ scheme, this is the same as the extractor performing
 $$\left(\left[\cE_{\protpc}(\cdot \; ; \bfd,r_1,...,r_\ell,\td, \witness, \bfc, \bfc^*_{\witness}) \circ S_1^*(1^\secparam,\inst,\witness,\bfc,\cdot)\right]\otimes \Id \right)\left(\rho\right)$$
 on the distinguisher's state.
 
%To show that the $\bfc^*$ and the $\bfc^*_{\witness}$ cases are quantum computationally indistinguishable, we will proceed with some hybrids.\\
To show that the view of the sender when interacting with the honest receiver is indistinguishable (against polynomial time quantum algorithms) from the view of the sender when interacting with the extractor. 
%\todo[inline]{I am confused on what is going on in the sentence above. (i.e. its not even a complete sentence (?)}

\noindent $\underline{\hybrid_1}$: The output of this hybrid is the view of the sender when interacting with the honest receiver.\\

\noindent $\underline{\hybrid_2}$: We define a hybrid receiver $\hybrid_2.\receiver$ that behaves like the honest receiver except that the input of $\hybrid_2.\receiver$ in $\protpc$ is $\bot$. The output of this hybrid is the view of the sender when interacting with $\hybrid_2.\receiver$.
\par The quantum indistinguishability of $\hybrid_1$ and $\hybrid_2$ follows from the semantic security of $\protpc$ against quantum polynomial time adversaries.\\ 

\noindent $\underline{\hybrid_3}$: We define a hybrid receiver $\hybrid_3.\receiver$ that behaves like $\hybrid_2.\receiver$ except that it sets $\bfc$ to be $\bfc=\comm(1^{\secparam},0;\bfd)$. The output of this hybrid is the view of the receiver when interacting with $\hybrid_3.\receiver$.

\par The quantum indistinguishability of $\hybrid_2$ and $\hybrid_3$ follows from the quantum computational hiding of $\comm$. \\

\noindent $\underline{\hybrid_4}$: We define a hybrid receiver $\hybrid_4.\receiver$ that sets $\bfc^*_i=\comm(1^{\secparam},\td_i;r_i)$, for every $i \in [\ell]$, where $\td$ is extracted inefficiently.

To prove that $\hybrid_3$ and $\hybrid_4$ are indistinguishable, we first establish some notation. Let $p_x$ be the probability that the sender samples $\td=x$, and let $\varepsilon_x$ denote the probability that the sender distinguishes $\hybrid_3$ and $\hybrid_4$ when $\td=x$.  Let $E_x$ denote the event that sender chooses $\td=x$ and that it distinguishes correctly. 
\par Suppose a QPT distinguisher can distinguish  $\hybrid_3$ and $\hybrid_4$. Then it follows that  $\Pr[\cup_x E_x]$ is non-negligible. Moreover, we have the following: 
\begin{align*}
\Pr[\cup_x E_x] &= \sum_{x} p_x \varepsilon_x \\
&\leq \max_x (\varepsilon_x)
\end{align*}
where we used the fact that $\{E_x\}$ are mutually exclusive events.
Since $\Pr[\cup_x E_x]$ is non-negligible, this means that there exists an $x$ such that $\varepsilon_x$ is non-negligible. This further implies that $\comm(0)$ and $\comm(x)$ are distinguishable with non-negligible probability, thus contradicting the quantum computational hiding security of $\comm$. 
\par Thus, the computational indistinguishability of $\hybrid_3$ and $\hybrid_4$ follows from the quantum computational hiding of $\comm$. \\

\noindent $\underline{\hybrid_5}$: We define a hybrid receiver $\hybrid_5.\receiver$ that behaves as $\hybrid_4.\receiver$ except that it sets $\bfc$ to be $\bfc=\comm(1^{\secparam},(r_1,\ldots,r_{\ell});\bfd)$, where $r_i$ is the randomness used in the commitment $\bfc^*_i$.
\par The quantum indistinguishability of $\hybrid_4$ and $\hybrid_5$ follows from the quantum computational hiding of $\comm$. \\

\noindent $\underline{\hybrid_6}$: The output of this hybrid is the output of the extractor. 
\par The quantum indistinguishability of $\hybrid_5$ and $\hybrid_6$ follows from the semantic security of $\protpc$ against polynomial time quantum adversaries. 

\end{proof}
\end{proof}
}

\section*{Acknowledgements}
We are grateful to Kai-Min Chung for many clarifications regarding quantum zero-knowledge proof and argument systems. We thank Thomas Vidick and Urmila Mahadev for answering questions about noisy trapdoor claw-free functions. We thank Abhishek Jain for helpful discussions and pointing us to the relevant references.

\fullversion{ 
\bibliographystyle{alpha}
\bibliography{crypto}
}

\submversion{ 
\bibliographystyle{plain}
\bibliography{crypto}
}

\clearpage

%\submversion{
%\section*{\underline{Supplementary Material}}

%\appendix 

%\input{addprelims}
%\input{proofs-clqext}
%\input{applzk}
%\input{proofs-qqext}
%}

\end{document}